\documentclass[acmsmall]{acmart}

\AtBeginDocument{%
  \providecommand\BibTeX{{%
    \normalfont B\kern-0.5em{\scshape i\kern-0.25em b}\kern-0.8em\TeX}}}

\usepackage{epsfig}
\usepackage{wrapfig}
\usepackage{balance}
\usepackage{pifont}

\usepackage{float}
\usepackage{times}
\usepackage{graphicx} 
\usepackage{subcaption}
\usepackage{tabularx}

\usepackage{xspace}
\usepackage{listings}
\usepackage{verbatim}
\usepackage{booktabs}
\usepackage{colortbl}

\usepackage{nicefrac}
\usepackage{siunitx}
\usepackage{stackengine}
\usepackage[small, compact]{titlesec}
\usepackage{enumitem}
\setlist{nolistsep}

\usepackage{amsmath}
\usepackage[subtle]{savetrees}

\usepackage{color}
\definecolor{darkred}{rgb}{0.7,0,0}
\definecolor{darkgreen}{rgb}{0,0.5,0}
\hypersetup{colorlinks=true,
        linkcolor=darkred,
        citecolor=darkgreen}

\newcommand{\Fig}[1]{Fig.~\ref{fig:#1}\xspace}

\newcommand{\ma}[1] {{\textcolor{red}{MA: #1}}}
\newcommand{\tom}[1] {{\textcolor{red}{tom: #1}}}
\newcommand{\pg}[1] {{\textcolor{black}{{#1}}}}
\newcommand{\todo}[1] {{\textcolor{red}{TODO: {#1}}}}
\newtheorem{theorem}{Theorem}
\newtheorem{lemma}{Lemma}
\newtheorem{definition}{Definition}
\newtheorem{conjecture}{Conjecture}
\newtheorem{proposition}{Proposition}
\newtheorem{corollary}{Corollary}
\newtheorem{sublemma}{Sub Lemma}[lemma]

\newcommand{\cut}[1]{}

\def\compactify{\itemsep=0pt \topsep=0pt \partopsep=0pt \parsep=0pt}
\let\latexusecounter=\usecounter

\renewcommand\footnotetextcopyrightpermission[1]{}
\begin{document}
\settopmatter{printacmref=false}
\fancyhead{}
\fancyfoot{}
\pagestyle{plain}

%\title{\Large \pg{Congestion Control Behaviour on Time-Varying Wireless Links}}
\title{\Large Optimal Congestion Control for Time-Varying Wireless Links}
%\pg{Swanky Title Needed}
% \title{\LARGE Designing Congestion Control Loops for Time-Varying Wireless Links}
\author[]{Prateesh Goyal}
\affiliation{\institution{Microsoft Research}\country{USA}}
\author[]{Mohammad Alizadeh}
\affiliation{\institution{MIT CSAIL}\country{USA}}
\author[]{Thomas E. Anderson}
\affiliation{\institution{University of Washington}\country{USA}}

\renewcommand{\shortauthors}{Goyal et al.}

\begin{abstract}
    Modern networks exhibit a high degree of variability in link rates. Cellular network bandwidth inherently varies with receiver motion and orientation, while class-based packet scheduling in datacenter and service provider networks induces high variability in available capacity for network tenants.  Recent work has proposed numerous congestion control protocols to cope with this variability, offering different tradeoffs between link utilization and queuing delay.  In this paper, we develop a formal model of congestion control over time-varying links, and we use this model to derive a bound on the performance of {\em any} congestion control protocol running over a time-varying link with a given distribution of rate variation. Using the insights from this analysis, we derive an optimal control law that offers a smooth tradeoff between link utilization and queuing delay. We compare the performance of this control law to several existing control algorithms on cellular link traces to show that there is significant room for optimization.

\end{abstract}

%%
%% The code below is generated by the tool at http://dl.acm.org/ccs.cfm.
%% Please copy and paste the code instead of the example below.
%%

%\ccsdesc[500]{Networks~Transport protocols}
%\ccsdesc[300]{Networks~Network protocol design}
%\ccsdesc[300]{Networks~Network performance modeling}

%%
%% Keywords. The author(s) should pick words that accurately describe
%% the work being presented. Separate the keywords with commas.

\maketitle
\section{Introduction}
\label{s:intro}
Traditional end-to-end congestion control algorithms for computer networks combine two orthogonal ideas
into a single mechanism. The first is to discover the capacity of the
path to carry traffic through probing~\cite{Jacobson88}, and the second is to respond to congestion in a way that
approximates fairness with respect to how the link is shared~\cite{jain1996congestion}.  Provided that traffic flows
persist over multiple round trips~\cite{homa}, and the network provides timely and regular feedback to endpoints~\cite{xcp,rcp}, 
a number of practical algorithms exist that can achieve both low queuing delay and high throughput~\cite{homa,hpcc,swift}.

\begin{wrapfigure}{r}{0.45\textwidth}
    %\centering
    \vspace{-2mm}
    \includegraphics[width=0.43\textwidth]{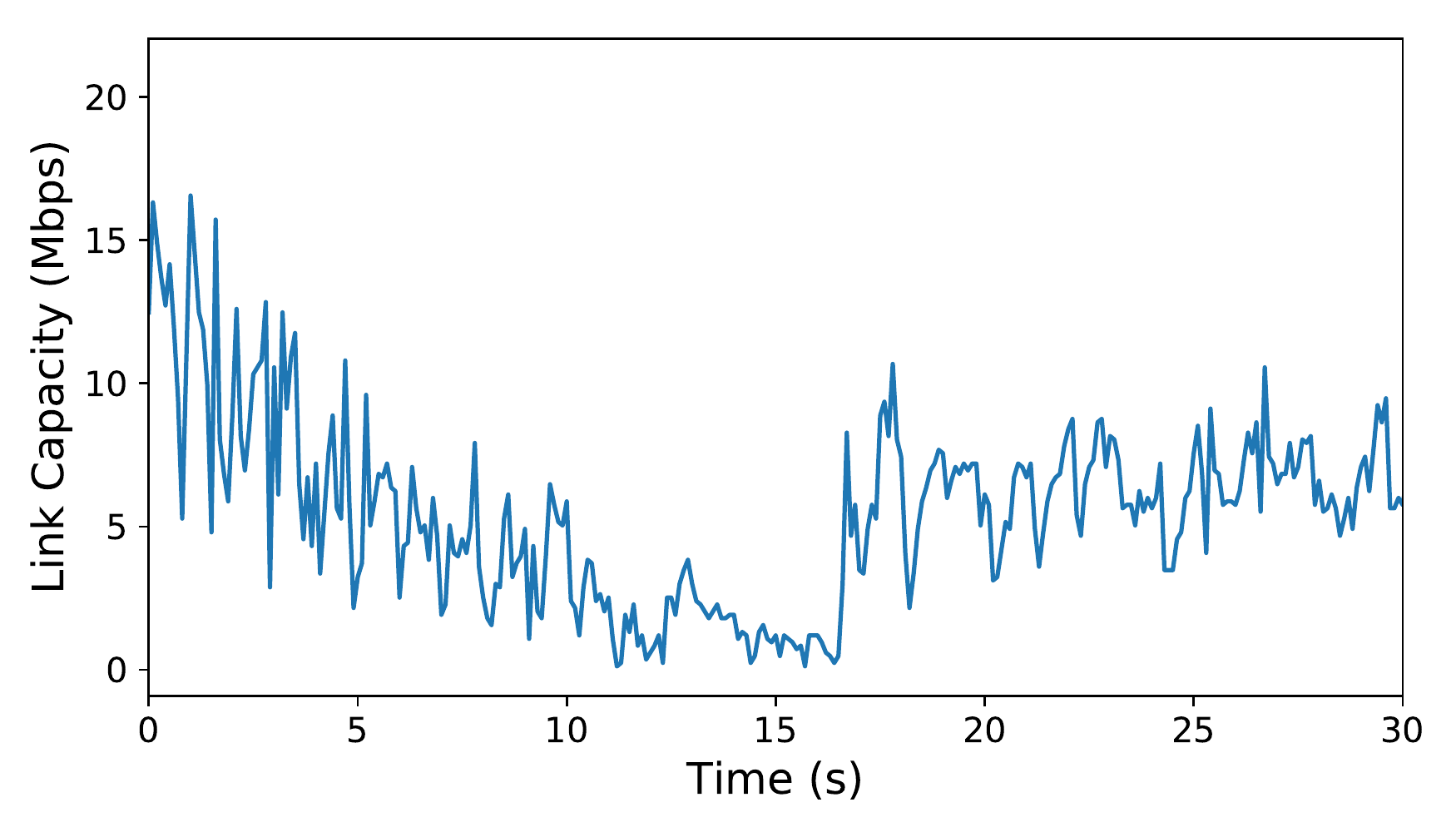}
    \vspace{-4mm}
    \caption{\small{\bf Variations in link capacity for a sample Verizon LTE link.}}
    \label{fig:intro_ts}
    \vspace{-3mm}
\end{wrapfigure}
Increasingly, however, modern networks pose a third challenge: path or link capacity itself is variable 
on a fine-time scale.  For example, in cellular networks, changes in receiver orientation
can cause available bandwidth to oscillate up or down by a factor of two within a few seconds~\cite{sprout}.  \Fig{intro_ts}, a sample
trace from an LTE network, illustrates this effect. Class-based packet scheduling for virtual networks induces similar
behavior.  To isolate different types of customers from each other, network switches can assign traffic from different
virtual networks to different queues, with different scheduling weights or priorities given to each queue.  
Whether traffic in other classes is present or absent can have a large effect on observed network capacity within the virtual network. Because of the inherent delay in transmitting
information about network state, link variability means that the endpoint must {\em guess} the 
capacity of the network in the near future, and will often guess wrong.  This can lead to link underutilization (if the capacity is underestimated) 
or excess queuing (if it is overestimated).

A number of recent research papers have attempted to adapt traditional congestion control
to this new setting~\cite{abc, goyal2017rethinking, verus, sprout, pbe_cc, copa, pcc, bbr}. 
These protocols differ from each other in many aspects; in particular, which \emph{feedback signals} are used by the sender and the exact \emph{control law} to adjust the sending rate. \Fig{intro_perf} shows the performance of various schemes on a
sample Verizon LTE trace, reproduced from earlier work~\cite{abc}. 
While several schemes lie on what appears to be the Pareto frontier, there is no clear winner:
different schemes provide different tradeoffs between underutilization and delay. Thus, we must pick different algorithms
and protocols depending on how much we value each goal. Low queueing reduces short request latency;
high throughput benefits longer requests.

\begin{wrapfigure}{r}{0.45\textwidth}
    %\centering
    \vspace{-2mm}
    \includegraphics[width=0.43\textwidth]{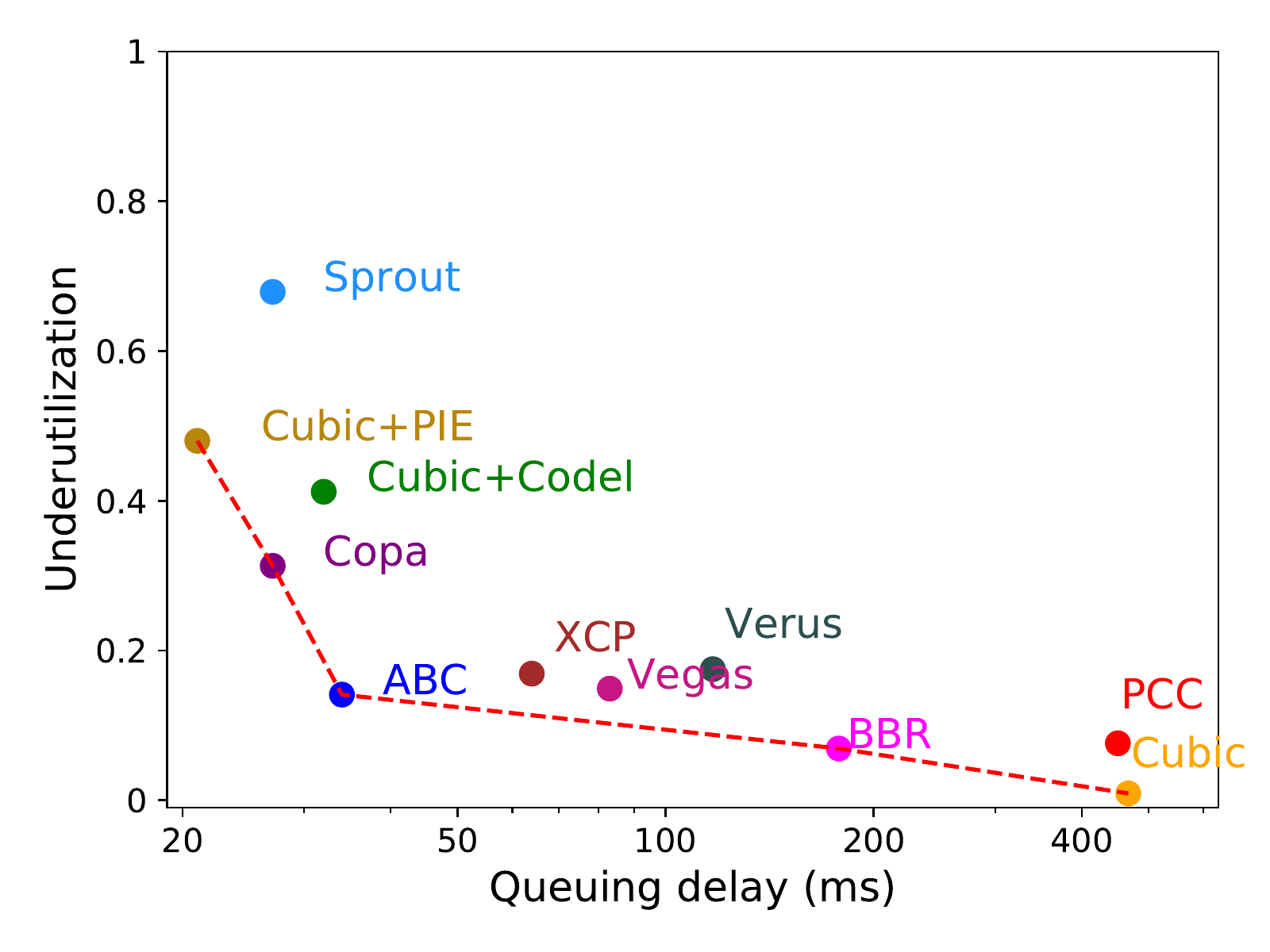}
    \vspace{-4mm}
    \caption{\small{\bf Performance of existing schemes on a Verizon LTE trace.} The experiments consist of a single long-running flow. Numbers are from ~\cite{abc}. The dashed red line represents the Pareto frontier of performance. }
    \label{fig:intro_perf}
    \vspace{-3mm}
\end{wrapfigure}
In this paper, we step back and ask a more fundamental question. Is there a bound on the optimal performance for congestion control
over time varying links \pg{based on the statistical properties of the variation.}  To make this problem tractable, we reduce the setting to the simplest imaginable: a single, long-lived
flow over a time-varying link that changes only once per round trip according to a \pg{simple} Markov process~\cite{howard1960dynamic}\pg{ (with restrictions on the transition probabilities). For each round trip, the link capacity differs from the previous round trip by a multiplicative I.I.D. random variable factor.}
%\ma{Maybe say: Markov process with I.I.D. multiplicative increments?}  
We show that it is possible to derive a bound for this case analytically. 
The bound is a function of the uncertainty in link capacity at the sender: the more 
uncertain the link, the worse the bound.

We can then ask, for this simplified model of real network links, is it possible to achieve this bound? We show that the answer is yes, provided the link variability is small. 
% \pg{We don't achieve the (loose) performance bound if the condition on variability is not met.} \pg{Alt text: To this end, we formulate the congestion control task as a MDP to derive the optimal control law for time-varying links.} 
For the more general case, we formulate the congestion control task as a Markov Decision Process~\cite{sutton2018reinforcement} (MDP) to derive an optimal control law for time-varying links. 
% \pg{Not sure if we should say that control law is not closed form. The parameters in the optimal control law are not closed-form, and the achievable performance with the control law is not closed-form.} 
Interestingly, the structure of this optimal control law is quite different from the control laws
used by much of the previous work in congestion control for time-varying links.
% that have been previously proposed for traditional and time-varying links.  
In particular, the optimal sending rate is unrelated 
to the previous sending
rate, but depends only on the previous link capacity and queue occupancy.  This control structure allows
us to smoothly trade queuing delay for better utilization across a broad range. 
% \pg{I am unsure about this paragraph. The performance bound only gives us a hint that the optimal control loop is of the particular form. We derive the optimal control Loop for all values of trade-off using the MDP.}
We then extend our analysis to prove an optimal control law for the case where the link
variability is predictable in some manner (e.g., using machine learning or other knowledge). \pg{Finally, we consider scenarios where the link variability has a more general Markov structure with few restrictions on the transitions.}%some deeper Markov structure.  
%\ma{Deeper Markov structure is an odd term. What does that mean? -> ``where link variability has a more general Markov structure.'' (another reason to be explicit about the I.I.D. increments earlier}
% \pg{I guess its fine not to say anything about the generic markovian model here.}
% Finally, we are able to generalize our model to quantify the benefit of improved predictability (e.g., by leveraging link layer feedback on signal quality) and techniques that exploit the structure of the link variability.

Note that our models are simplifications, and real wireless links need not follow the models we use. Our goal is not to construct perfect models for link capacity based on the specifics of the link and the physical layer. Instead, we present general and tractable models to capture key components of variability in link capacity on the scale of the round trip delay, draw insights, and show that these insights can translate to performance improvements. We use real cellular traces to validate our findings.

\subsection{Key Findings}
\begin{wrapfigure}{r}{0.45\textwidth}
    %\centering
    \vspace{-15mm}
    \includegraphics[trim={1cm 0 11cm 0},clip, width=0.43\textwidth]{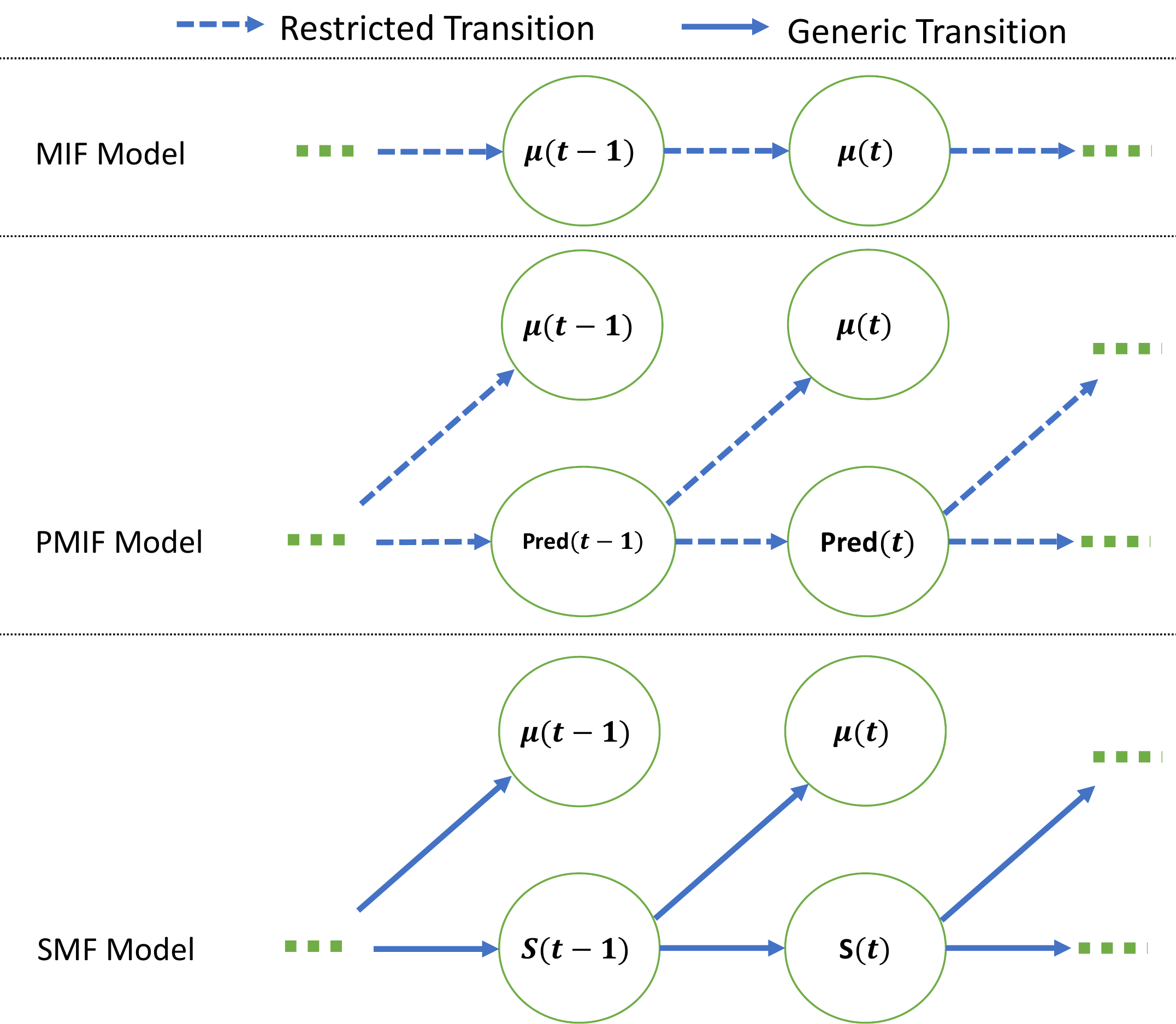}
    \vspace{-3mm}
    \caption{\small{\bf Transition diagram for different models of link capacity.}}
    \label{fig:intro_markovmodel}
    \vspace{-16mm}
\end{wrapfigure}
To capture variations in the link capacity, we develop three discrete-time Markov chain models for the link capacity with increasing complexity. \Fig{intro_markovmodel} shows the schematics for the three link models we consider in this paper.
For each model, we establish that there is a \emph{fundamental} bound on congestion control performance in terms of expected link underutilization and queuing delay
% \footnote{In \S\ref{s:discussion}, we show how we can extend our analysis to calculate the bound for other performance metrics.} 
(Theorem~\ref{thm:model1:perfbound}, Corollary~\ref{cor:model2:perfbound} and Theorem~\ref{thm:model3:perfbound}) \pg{, and we derive the optimal congestion control law (Theorem~\ref{thm:model1:mdp}, Corollary~\ref{cor:model2:mdp} and \ref{cor:model3:mdp}).}

\vspace{10mm}
\smallskip
\noindent
%\textbf{Model 1: Restricted markov chain (\S\ref{s:model1}):} 
\textbf{Multiplicative I.I.D. Factors (MIF) model (\S\ref{s:model1}):} 
We begin by analysing a basic model for variations in link capacity. The  analysis of this model is easy to follow yet the model is still expressive enough to provide insights that improve performance on real-world cellular traces. Later, we will use the basic analysis framework we establish for this model to extend our analysis to more comprehensive models.
In this model, the link capacity at time step $t$ ($\mu(t)$) is dependent on the capacity in the previous round trip ($\mu(t-1)$) as follows,
\begin{align}
    \mu(t) = \mu(t-1) \cdot X_t,
\end{align}
where $X_t$ are I.I.D random variables. % In this model, the relative uncertainty in link capacity is same in every round and is governed by characteristics of $X$. %This model is relatively easier to understand and analyse. %We use insights from analysis of this model for other more complex models.
Under the assumptions of the model, the optimal control law is of the simple form,
\begin{align}
    s(t) = \left(C \cdot \mu(t-1) - \frac{Q(t-1)}{T}\right)^+,    
\end{align}
where $T$ is the base round trip time (RTT), $Q(t)$ is the queue size. $C$ is a positive constant that depends on the desired trade-off between underutilization and queuing delay, and characteristics of the random variable $X$.  

This control law is interesting for several reasons. First, in several end-to-end schemes~\cite{verus, sprout, copa, cubic}, the current link capacity is not signalled back to
the sender. Instead, the sender has to infer the link capacity indirectly. Such inference might have errors~\cite{abc}, and these errors can lead to poor performance.
%\ma{i don't understand this point. of course end-to-end schemes are at an disadvantage; our analysis focuses on explicit schemes}.\pg{I understand ABC makes the same point but we theoretically show that the optimal control law needs link capacity. I like it. See the first heading in \S\ref{ss:model1:implications}}
Second, in 
existing explicit signalling protocols\footnote{Explicit protocols directly signal congestion or rate information to endpoints from routers.}  such as XCP~\cite{xcp} and RCP~\cite{rcp},
the sending rate is derived iteratively based on the sending rate in previous rounds. The prototypical control law adjusts the sending rate up/down by increments that depend on a congestion signal fed back by the routers. By contrast, our analysis shows that the optimal controller does {\em not} depend on the prior sending rate but only on the link capacity and queuing in the previous round.
If in a given round, we guess wrong (as we must from time to time), we should not 
compound that error by basing our new sending rate on the old erroneous one.
Finally, to adjust the trade-off in underutilization and delay, this rule argues for only changing the constant $C$ and keeping the queuing penalty term the same. In contrast, protocols like ABC~\cite{abc} incorrectly treat $C$ as a ``target utilization'', setting it to a value close to one and adjusting the queuing penalty term to achieved a desired performance trade-off. 
We show that is not optimal. In our evaluation over real cellular traces, our optimal control law outperforms these existing schemes.

% old version
% First, the form of the control loop presents a strong case for explicit congestion control. This is because to calculate the sending rate the sender needs the link capacity in the previous round. With explicit feedback, the bottleneck router can simply communicate the link capacity directly. In contrast, for end-to-end schemes~\cite{verus, sprout, copa, CoDel}, the sender has to infer the link capacity indirectly. Such inference might have errors~\cite{abc}, and consequently lead to sub-optimal performance. Additionally, this rule differs from existing explicit protocols like XCP~\cite{xcp} and RCP~\cite{rcp}. In particular, the sending rate is not dependent on the sending rate in previous rounds and is directly based on the link capacity and queue size. Finally, to adjust the trade-off in underutilization and delay this rule argues for only changing the constant $C$ and keeping the queueing penalty term same. In contrast, explicit protocols like ABC~\cite{abc} incorrectly treat $C$ as target utilization, setting it to a value close to one and adjusting the queuing penalty term instead to achieved a desired performance trade-off. In our evaluation over real cellular traces, the optimal control loop outperforms existing schemes. %We use analysis of this model as a basis to analyse more complex models.

\smallskip
\noindent
%\textbf{Model 2: Prediction for link capacity (\S\ref{s:model2}):}
\textbf{Prediction-based Multiplicative I.I.D. Factors (PMIF) model (\S\ref{s:model2}):}
Next, we consider scenarios where the sender has access to  an imperfect prediction ($Pred(t-1)$) about the current link capacity ($\mu(t)$). The uncertainty in link capacity at the sender is governed by the following equations:
\begin{align}
    \mu(t) &= Pred(t-1) \cdot X^p_t, & Pred(t) = Pred(t-1) \cdot X^{pred}_t,
\end{align}
where $X^p_t$ and $X^{pred}_t$ are I.I.D random variables.  In other words, prediction
error is structurally similar to the error induced by feedback delay in the MIF model.
\pg{A corollary from the analysis of the MIF model is that the optimal control law in this new model is of the form,}
\begin{align}
    s(t) = \left(C^p \cdot Pred(t-1) - \frac{Q(t-1)}{T}\right)^+,   
\end{align}
where $C^p$ is a positive constant that depends on the desired performance trade-off and characteristics of $X^p$ and $X^{pred}$. We find that a good predictor for link capacity can substantially improve performance on cellular traces. Thus, a promising future direction is to design algorithms to predict link capacity for time-varying wireless links, e.g., using physical layer information~\cite{predictwifi}.%\ma{i think there is some work on this? useful to cite some papers}.\pg{couldn't find much, will add if I find anything relevant}

\smallskip
\noindent
%Prediction-based Multiplicative I.I.D. Factors (PMIF)
\textbf{State-dependent Multiplicative Factors (SMF) (\S\ref{s:model3}):}
%\textbf{Model 3: Generic Markov chain (\S\ref{s:model3}):} 
Finally, we relax the constraints in the MIF Model for how capacity changes between each time step. We assume a generic Markov chain for how the underlying link state governing the link capacity (depicted by $S(t)$) evolves. \pg{In particular, the link state is a set that includes the current link capacity and any other quantities that might impact the link capacity in the next round trip. In this model, the uncertainty in link capacity at the sender ($\mu(t)$) is governed by the link state in the previous round trip ($S(t-1)$).
The optimal control law is of the form,}
\begin{align}
    s(t) = \left(C^{S(t-1)} \cdot \mu(t-1) - \frac{Q(t-1)}{T}\right)^+    
\end{align}
\pg{where $\textbf{S}$ represents the space of all possible states $S(t)$, and, $C^k, \forall k \in \textbf{S}$ is a positive constant that depends on the probability transition matrix for the link state. Compared to the MIF Model, this model is more realistic as it does not restrict the relative uncertainty in the link capacity to be the same in every round and allows for the uncertainty to be dependent on the underlying link state. For example, the SMF model (but not the MIF model) allows there to be a minimum and maximum capacity. As expected, on cellular traces, the control law above improves performance over the (simpler) control law based on the MIF model.}

\section{Related Work}
\label{s:related}

% \pg{network calculus earlier link models?}

\smallskip
\noindent
\textbf{Queueing Theory:} A substantial body of work in the field of queuing theory pertains to analyzing the dynamics of queues in computer systems~\cite{kleinrock1976queueing}. Existing techniques are useful for analysing systems in steady state equilibrium where both the rate of arrival of jobs and the job size can be modeled as probabilistic processes (e.g., as Markovian chains). An important characteristic of our work is that
the queue service rate is variable with time. A branch of queuing theory considers
time-varying behavior as an exception to steady state analysis. 
Some work has focused on modelling time-varying arrival rates, e.g., as a way
to model diurnal customer requests~\cite{massey}. Time-varying service rates have also been considered
for the case where the number of servers scales up to meet increased demand~\cite{tvq}. By contrast,
the service rate of a cellular link is endogenous, with variation that is largely 
independent of the workload demand. Thus, even though the link rate is time varying, 
we can analyse it as a steady state Markov process.
Cecchi and Jacko~\cite{tvwifi} develop a general Markov model, similar
to our SMF model, for representing changes in wireless transmission capacity. Their
focus is on studying multi-user fairness for media access rather than congestion control. 
 %\pg{on second thought, happy with what you wrote. We make the point about deriving optimal control law in the next para anyways.}

% Such techniques might potentially be useful for analyzing closed loop systems such as the performance/properties of a \emph{particular} congestion control law on the time-varying markovian links we consider in this paper. %\pg{say more? maybe model the system as a state dependent M/M/1 queue in kendall notation?}
% In contrast to existing techniques, this paper provides a framework for quantifying the performance space of \emph{all} the congestion control algorithms on  markovian links. Additionally, we show that  by treating the problem of choosing the sending rate as a markov decision process (MDP) one can also derive the optimal congestion control for such markovian links.

\smallskip
\noindent
\textbf{Analysis of real-world congestion control protocols:} This is a crowded space with a lot of work spanning theoretical modeling~\cite{chiu1989analysis, mathis1997macroscopic}, experimental analysis~\cite{hock2017experimental, yan2018pantheon}, and recently formal verification~\cite{arun2021toward, nsdi_rw_tcp_bugs}. While prior work can be used for analysing the performance of a particular congestion control protocol, they cannot be used to reason about the performance space of all congestion control laws or to deduce the optimal control law itself. Moreover, the theoretical analysis in existing works is primarily for wired links with relatively fixed capacity.\footnote{Some works do allow for minor variations in link capacity to account for networks components such as token bucket filters~\cite{arun2021toward}.} In contrast, we use Markovian models to capture the time-varying nature of the wireless link. 

\smallskip
\noindent
\textbf{Existing congestion control protocols:} We can classify existing protocols based on the feedback signals used. Cubic~\cite{cubic}, TCP Reno~\cite{newreno}, Sprout~\cite{sprout}, Verus~\cite{verus}, and Copa~\cite{copa} use traditional feedback signals such as end-to-end delay, packet drops and explicit congestion notification (ECN) marks (with various active queue management (AQM) schemes such as CoDel~\cite{CoDel}, PIE~\cite{pie}, and RED~\cite{red} to mark the ECN) as a basis to adjust rates. These signals are good for inferring congestion to reduce the sending rate. However, when the link is underutilized these signals do not offer much information about the current link capacity. In such cases, the sender has to resort to a blind increase which can cause poor performance on time-varying links~\cite{abc}. Explicit congestion control protocols~\cite{xcp,rcp, abc, hpcc} where the bottleneck router explicitly communicates a rate to the sender based on the link capacity can overcome this limitation. These protocols also differ from each other in the control law they use to adjust the rate (see \S\ref{ss:model1:implications}). For cellular networks, some recent work~\cite{xie2015pistream, lu2015cqic, pbe_cc} proposes techniques for leveraging physical layer information at the sender to compute link capacity without any explicit feedback from the base station.

%Where we differ in this paper is our analysis for performance of all congestion control loops  

%Instead,  pose the problem of designing an optimal congestion control loop on time-varying links as a markov decision process to govern the rate of incoming packets at the link.

%It might be possiblt to analyse closed control loops 
%However, congestion control loops inherently present a closed loop system. 
%congestion control is inherently a closed loop system where the rate . The problem of designing an optimal control loop where the link capacity is a markov chain can naturally be framed as a MDP process.
\section{Preliminaries}
\label{s:setup}

%\ma{We use round and time step interchangeably, but I think it's better to be consistent. I'd just use time step everywhere.}
\noindent{\bf Model.} We study a discrete-time model of a single long-lived flow transmitting data on a link with time-varying capacity. Each discrete time step (or ``round'') corresponds to one base round trip time, $T$, defined as the minimum possible round trip time (RTT) in the absence of queuing delay. Let $\mu(t)$ and $Q(t)$  denote the link capacity and queue length in time step $t$. At the start of time step $t$, the sender receives feedback about the {\em state} of the system in the previous time step, i.e. $\mu(t-1)$ and $Q(t-1)$. It then selects a non-negative sending rate, $s(t)$, which it uses throughout time step $t$. Consequently, the queue length evolves according to:
\begin{align}
    Q(t) &= (Q(t-1) + s(t) \cdot T- \mu(t) \cdot T)^+,
    \label{eq:setup:queue_size_def}
\end{align}
where $y^+ = max(0, y)$. We assume $Q(0)=0$. 
%It is evident that $Q(t)$ depends on the sending rates and link capacities in all previous time steps: ($\cup_{i=1}^{t}\{s(i)\}$) and ($\cup_{i=1}^{t}\{\mu(i)\}$). 

We make several simplifying assumptions to keep the analysis tractable. First, our model ignores sub-RTT dynamics, since all network feedback and congestion control decisions occur at the beginning of each time step. Second, we assume that the bottleneck router has unlimited buffer space and never drops packets. Third, we focus on {\em explicit} congestion control mechanisms, where the bottleneck router provides direct feedback about its state to the sender in each time step. Studying explicit schemes is natural because our goal is to determine performance bounds on achievable performance for any congestion control algorithm. Since explicit schemes have more information, any bound for these schemes also applies to approaches that must infer link conditions from end-to-end signals like packet loss and delay. %Moreover, explicit schemes have attracted considerable attention in congestion control research. Examples of such schemes include XCP~\cite{xcp} and RCP~\cite{rcp}, and recent protocols like ABC~\cite{abc} and HPCC~\cite{hpcc} that have been designed for wireless and datacenter networks that exhibit high variability. \pg{if we include explict vs implicit heading \S\ref{ss:model1:implications}, do we need to edit this?}\tom{I would drop the last two sentences. its true our focus is on explicit schemes. we mention all this in the intro so no need to dwell on it here.}

Although we focus on the single-flow case, our analysis extends to settings in which multiple long-lived flows with the same base round trip time share a bottleneck link. In such cases, $s(t)$ refers to the aggregate sending rate across all the senders. %\pg{footnote it?}

\noindent{\bf Metrics.} Define $q(t)$ as the ratio of the queue length to the link capacity in time step $t$. The evolution of $q(t)$ is governed by: 
\begin{align}
    q(t) &\triangleq \frac{Q(t)}{\mu(t)} = \frac{(Q(t-1) + s(t) \cdot T - \mu(t) \cdot T)^+}{\mu(t)}.
    \label{eq:setup:queue_delay_def}
\end{align}
$q(t)$ represents the time that it would take to drain the last packet in the queue at time step $t$, assuming the link capacity stays constant at $\mu(t)$ until the packet departs. Since the link capacity is time-varying, $q(t)$ may differ from the actual delay experienced by the last packet in the queue at time step $t$. Nevertheless, it provides an approximate measure of the instantaneous queuing delay at time $t$, and henceforth we will refer to $q(t)$ as the queuing delay. 

We define the link underutuilization, $U(t)$, as the ratio of the unused capacity to the available link capacity in time step $t$. The link transmits $\min\{s(t)\cdot T + Q(t-1), \mu(t)\cdot T\}$ worth of data in time step $t$. Therefore the underutilization in time step $t$ is given by:  
\begin{align}
    U(t) &\triangleq \frac{(\mu(t) \cdot T - s(t) \cdot T - Q(t - 1))^+}{\mu(t) \cdot T}.
    \label{eq:setup:undeutilization_def}
\end{align}

Finally, we define the average underutilization ($\bar{U}(t))$) and the average queuing delay ($\bar{q}(t)$) as follows:
\begin{align}
    \bar{U}(t) &= \frac{\sum_{i=1}^{t}U(i)}{t}, & \bar{q}(t) = \frac{\sum_{i=1}^{t}q(i)}{t}. 
    \label{eq:setup:qu_multiple_round_def}
\end{align}
We will use $\bar{q}(t)$ and $\bar{U}(t)$ as our main performance metrics in this paper. Specifically, we will assume that the goal for any congestion control protocol is to achieve low average queuing delay and low average underutilization according to the above definitions.

\noindent{\bf Algorithms.}
In our analysis, we will only consider causal congestion protocols.
\begin{definition}{\textit{Causal Congestion Control protocols}---} A congestion control protocol is considered causal if, the sending rate in time step $t$, is only dependent on observations of the system (e.g., link capacity, queue size, etc.) in earlier time steps.%, i.e., $s(t)$ is dependent only on $\cup_{i=0}^{t-1}\{\mu(i)\}$, $\cup_{i=0}^{t-1}\{U(i)\}$, $\cup_{i=0}^{t-1}\{q(i)\}$.
\end{definition}
Causal congestion control protocols decide the sending rate based on past observations. Note that causal congestion control protocols also include randomized protocols where the sending rate is not-deterministic.% but is still dependent on past observations only.
%\ma{Doe we need this? I'm guessing the section on predictive algorithms doesn't really need this assumption. The analysis doesn't really care how you come up with prediction, it just tells you what your performance bounds are based on your prediction accuracy. I want to revisit this after reading that section, but if we include this paragraph, we should probably mention predictive algorithms too.}  \pg{Agreed, I removed the definition of predictive congestion control from there already. Maybe here we can just state that we consider causal congestion protocol where the sender has access to state of the \emph{state of the system}}

%\smallskip
%\noindent
%\textbf{Network Environment:}  We target proxied cellular networks that split the connection at the edge~\cite{wang2011untold,ravindranath2013timecard}). In such networks, the cellular base station schedules cellular users out of separate per-user queues. Because of the proxy, all the flows of a particular user have similar base round trip times. Typically, flows of a user are bottleneck at the cellular link. In such situations, our analysis can be applied on a per-user basis. We leave extending our analysis to network environments with multiple bottlenecks links or multiple flows with different base round trip times as future work.
%\pg{Do we need to create a subsection for multiple flows?} 

%\input{SIGMETRICS/Model1}
%\input{SIGMETRICS/Model2}
\section{A Simple Link Model}
\label{s:model1}
%\subsection{Model}
\label{ss:model1:model}
%We consider a network with a single sender and receiver pair connected via a time-varying link. We consider a discrete time model, where each discrete round's duration is $T$, the base RTT for the sender-receiver pair.
We begin with a simple model where $\mu(t)$ deviates from $\mu(t-1)$ by an I.I.D. multiplicative factor. We refer to this as the Multiplicative I.I.D. Factors (MIF) model. Our analysis of the MIF model forms the basis for studying more complex link models later in the paper. 

Formally, $\mu(t)$ evolves according to: 
\begin{equation}
    \mu(t) = \mu(t - 1) \cdot X_t,
    \label{eq:model1:model}
\end{equation}
where $X_t$ is a I.I.D random variable with probability density function (PDF) $f(\cdot)$. Since we don't want the link capacity to drop below zero, we assume that $X_t$ is positive. 
%\ma{The link capacity evolution is independent of the state of the link and the sending rate decisions of the congestion control algorithm, i.e. $X_t$ is independent of $Q(t-1)$, $\mu(t-1)$ and $s(t)$.} 
The link capacity evolution is independent of the sending rate decisions of the congestion control algorithm, formally, $X_t$ is independent of $Q(i)$, $\mu(i)$ and $s(i), \forall i < t$. Additionally, independence of $s(t)$ and $X_t$ follows from the assumption that the congestion control algorithm is causal and does not have advanced knowledge of future capacity variations (although it is permitted to know the {\em distribution} of the variations). %\pg{Check this.}

%Also we assume there is some distribution on the starting link rate ($\mu(0)$), and $\mu(0)>0$. 

%We also assume that in round 1, $\mu(0)$ is known. 

We will first establish a performance bound in terms of the expected underutilization ($E[\bar{U}(t)]$) and expected queuing delay ($E[\bar{q}(t)]$) that can be achieved by {\em any} causal congestion control algorithm. We use this analysis to draw insights about the optimal congestion controller for this setting (\S\ref{ss:model1:achievable_bound}). Next, we will derive the exact optimal control law by posing the problem of deciding the sending rate as a Markov Decision Process (MDP) (\S\ref{ss:model1:mdp}).
\pg{In \S\ref{s:model2}, we will consider an extension of the MIF model where the sender can predict the current link capacity with certain accuracy. We will derive the performance bound and the optimal control law for this extension.
Finally, in \S\ref{s:implications}, we will discuss the implications of our analysis and validate our findings using real-world cellular traces.}
%We will also draw insights about the optimal congestion control law for causal protocols.

\subsection{Performance Bound}
\label{ss:model1:performance_bound}

\begin{theorem} 
In the MIF model, for any causal congestion control protocol and any $t \in \mathbb{N}$, the point $(E[\bar{q}(t)]$, $E[\bar{U}(t)])$ always lies in a convex set $\mathbb{C}^f \triangleq \{(x,y) | y \geq g^f(x)\}$, where $g^f(\cdot)$ is defined as follows:
\begin{align}
    x &= T \cdot \int_{0}^{b}\left(\frac{b}{a} - 1 \right) \cdot f(a) \cdot da, & g^f(x) &= \int_{b}^{\infty}\left( 1 - \frac{b}{a}\right) \cdot f(a) \cdot da, 
\label{eq:model1:thm:perfbound}
\end{align}
where $b \in \mathbb{R}^+$.
\label{thm:model1:perfbound}
\end{theorem}

\if 0
\begin{theorem} 
In the MIF model, for any causal congestion control protocol, and any $t \in \mathbb{N}$, there is a fundamental trade-off between $E[\bar{q}(t)]$ and $E[\bar{U}(t)]$. The point \{$E[\bar{q}(t)]$, $E[\bar{U}(t)]$\} always lies in a convex set $\mathbb{C}^f := \{\{x,y\} | y>=g^f(x)\}$. Where $g^f$ is defined as follows,
\begin{align}
    x &= T \cdot \int_{0}^{b}\left(\frac{b}{a} - 1 \right) \cdot f(a) \cdot da\nonumber\\
    g^f(x) &= \int_{b}^{\infty}\left( 1 - \frac{b}{a}\right) \cdot f(a) \cdot da 
\label{eq:model1:thm:perfbound}
\end{align}
where $b \in \mathbb{R}^+$.
\label{thm:model1:perfbound}
\end{theorem}
\fi

The curve $g^f(\cdot)$ forms the boundary of the set $\mathbb{C}^f$. Any causal congestion control strategy cannot achieve better performance in expectation than the feasible set $\mathbb{C}^f$. Note that in a particular run (i.e., for a given $\cup_{i=0}^{t}\{\mu(i)\}$), a scheme might get lucky and achieve performance $(\bar{q}(t), \bar{U}(t))$ that lies outside $\mathbb{C}^f$ (i.e. below $g^f(\cdot)$), however, in expectation the performance is bounded by $\mathbb{C}^f$. %\ma{I think we can drop the superscript $f$ from $C^f$ and $g^f$. It's implicit in the model and simple notation is better.}\pg{the f subscript is useful in model 3}

\if 0
\begin{wrapfigure}{r}{0.48\textwidth}
    %\centering
    \vspace{-2mm}
    \includegraphics[width=0.45\columnwidth]{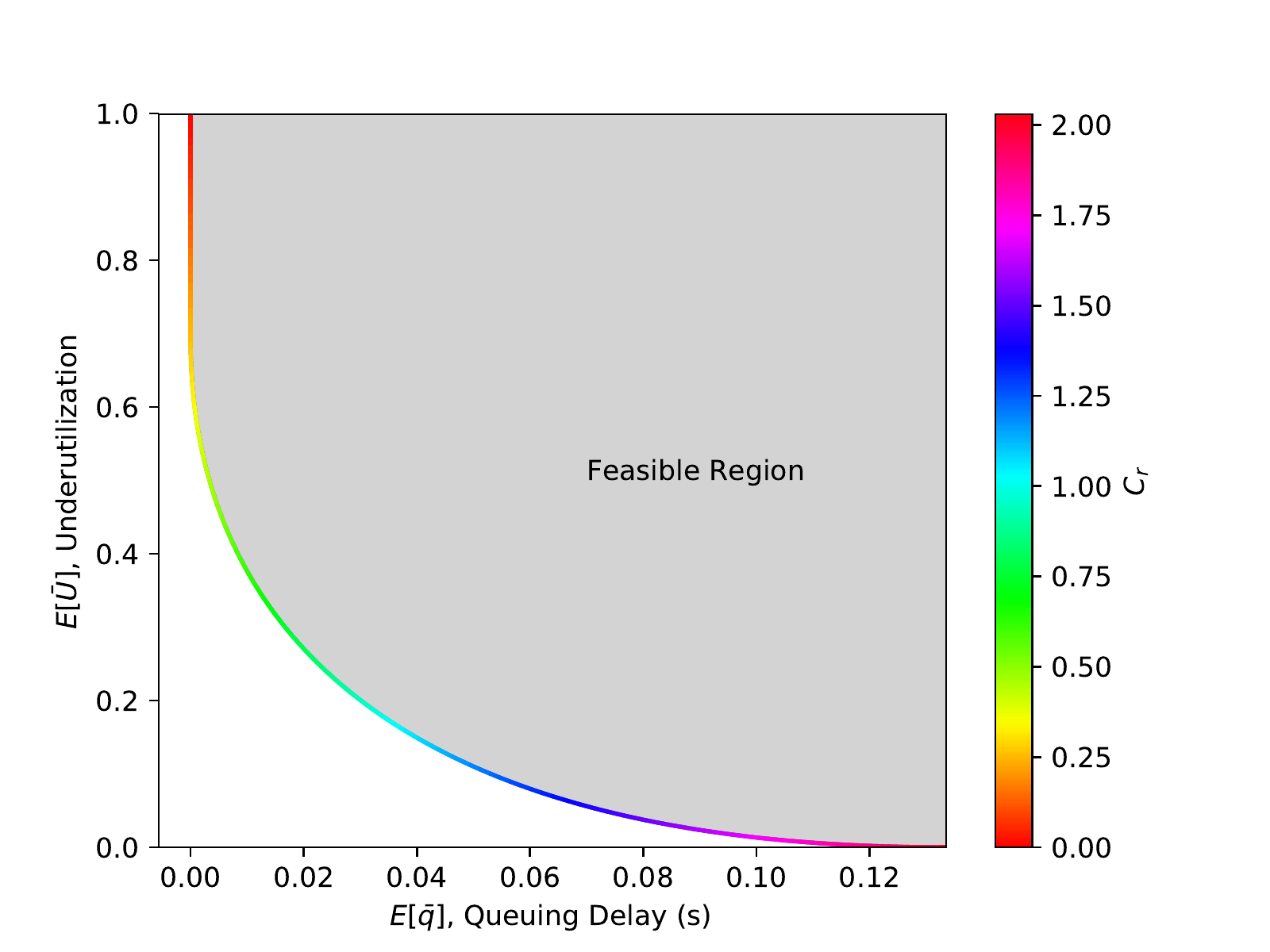}
    \vspace{-4mm}
    \caption{\small {\bf Trade-off between underutilization and queuing delay, for $X \sim \textit{U}(0.27, 2) $, $T = 0.1$s ---} The gray region presents the feasible region of any congestion control protocol. The color of the performance bound ($g^f$) corresponds to different values of $b$ in Eq.~\eqref{eq:model1:thm:perfbound}.}
    \label{fig:model1:feasible_uniform}
    \vspace{-4.5mm}
\end{wrapfigure}

  \begin{figure}[t]
     \centering
    \includegraphics[width=0.6\columnwidth]{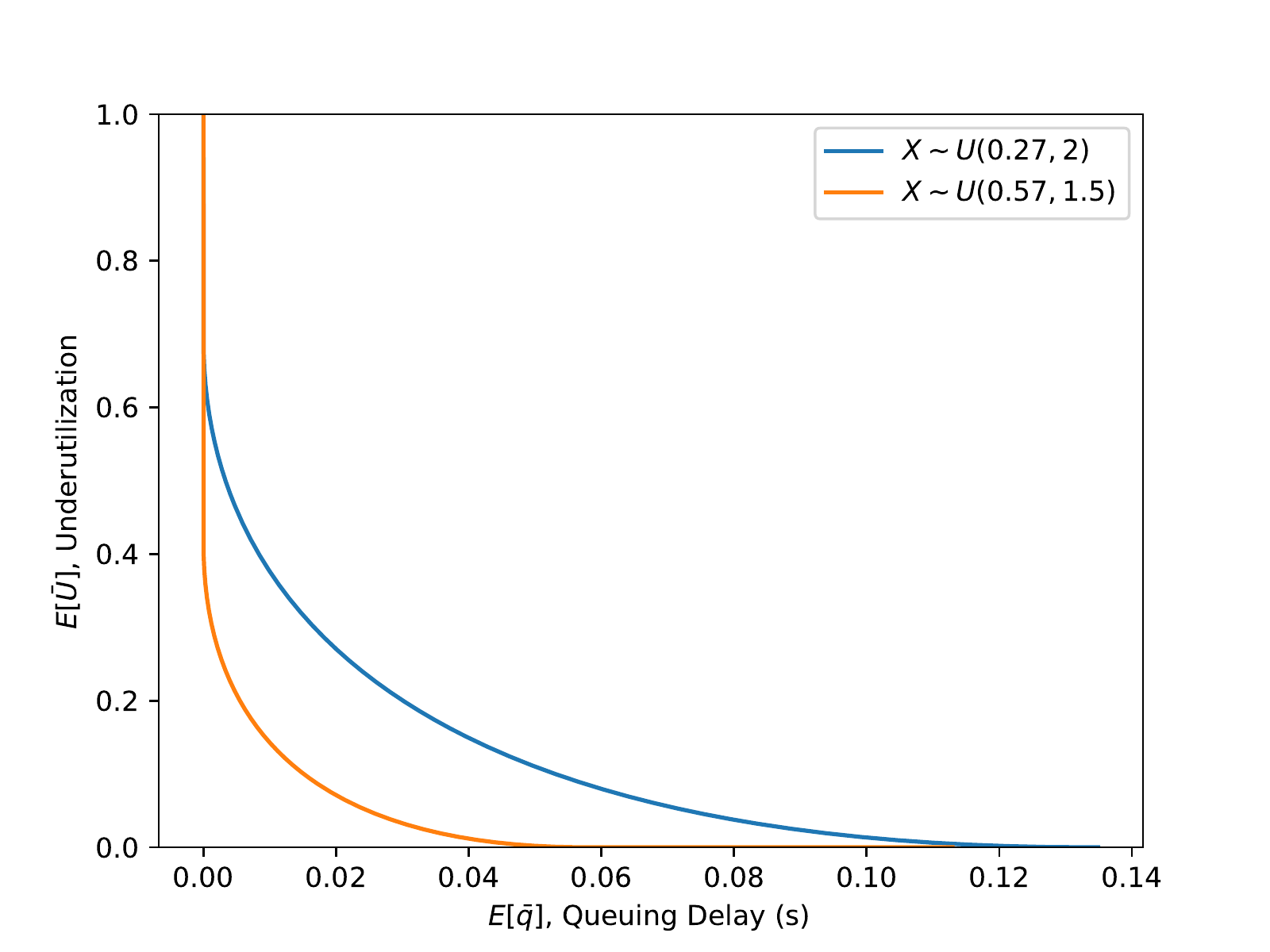}
    \vspace{-4mm}
    \caption{\small {\bf Trade-off is dependent on the PDF $f$ ---} More variable the link is worse the performance bound.}
    \label{fig:model1:f_variation}
    \vspace{-4.5mm}
 \end{figure}
 \fi
 
 \begin{figure}[t]
\begin{minipage}[t]{.48\textwidth}
\centering
    \includegraphics[width=1\textwidth]{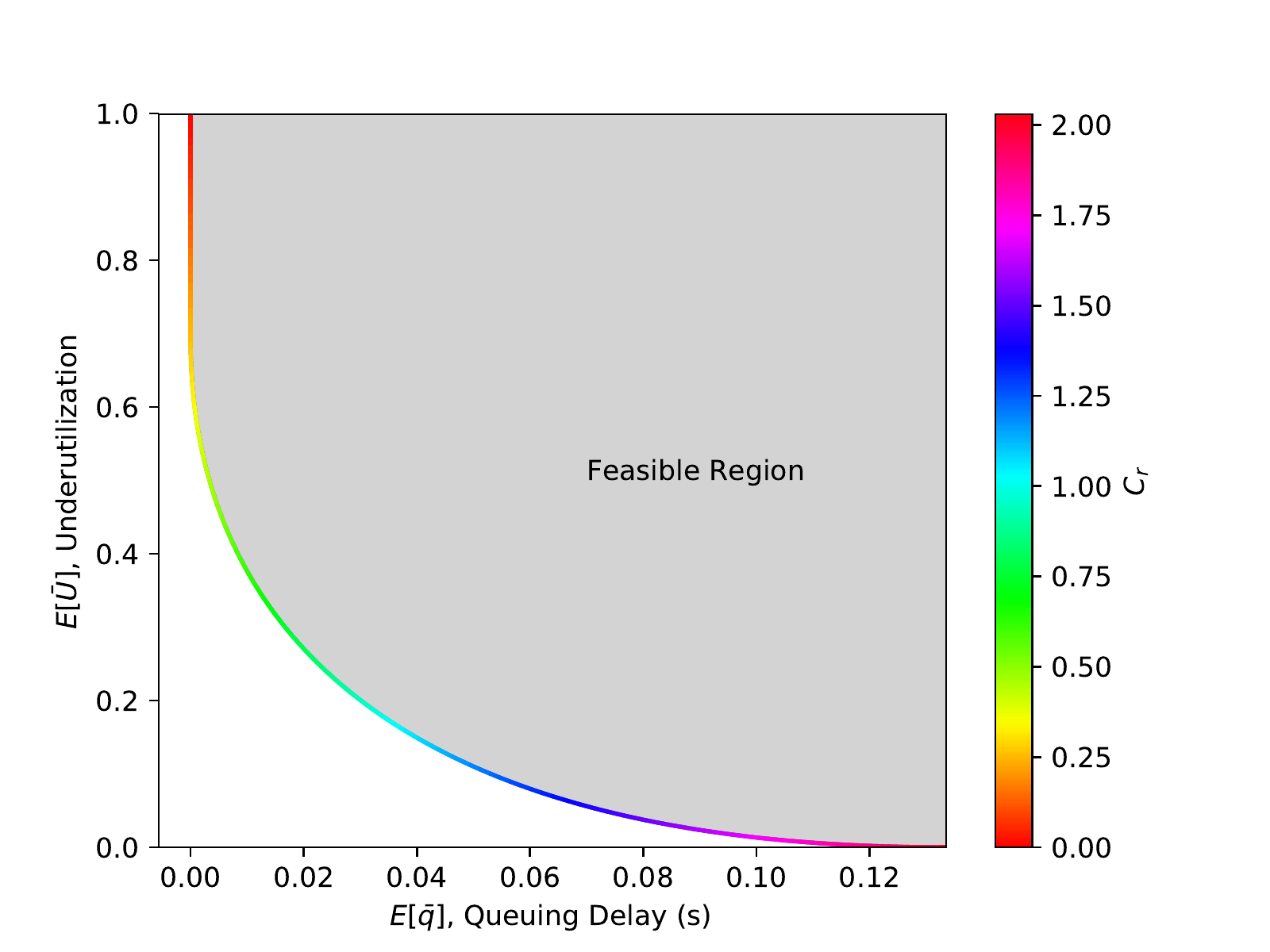}
    \vspace{-7mm}
    \caption{\small {\bf Trade-off between underutilization and queuing delay, for $X \sim \textit{U}(0.27, 2) $, $T = 0.1$s ---} The grey region presents that feasible region of any congestion control protocol. The color of the performance bound ($g^f(\cdot)$) corresponds to different value of $b$ in Eq.~\eqref{eq:model1:thm:perfbound}.}
    \label{fig:model1:feasible_uniform}
\end{minipage}
\hfill
\begin{minipage}[t]{.48\textwidth}
\centering
    \includegraphics[width=\textwidth]{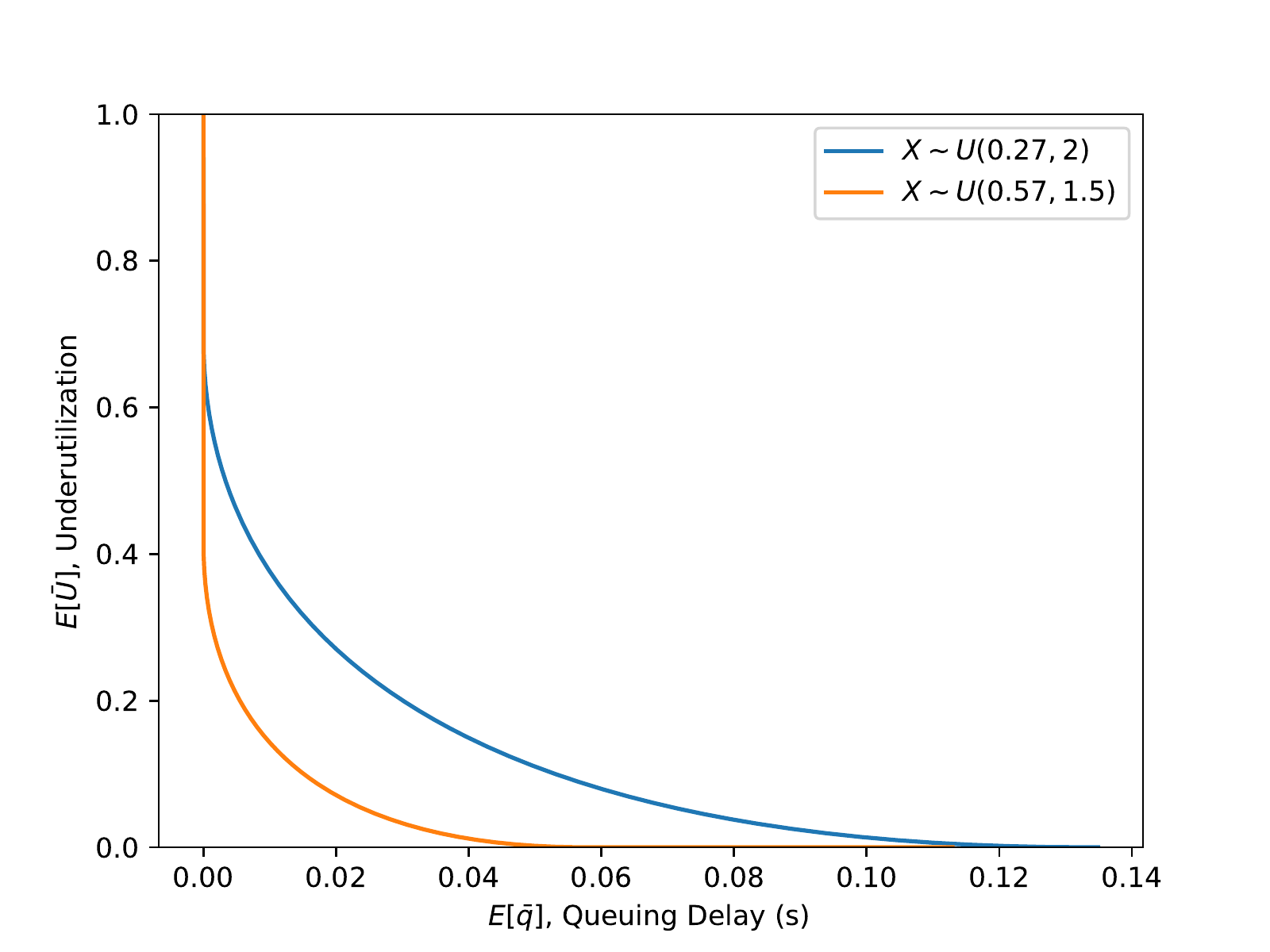}
    \vspace{-7mm}
    \caption{\small {\bf Trade-off is dependent on the PDF $f$ ---} The performance bound becomes worse (further up and to the right) as the extent of link rate variation increases.}
    \label{fig:model1:f_variation}
\end{minipage}
\vspace{-6mm}
\end{figure}
 
The performance bound depends on the distribution $f$ of the multiplicative factors $X_t$. \Fig{model1:feasible_uniform} shows this performance bound and the feasible region for a particular distribution. The distribution $f$ captures the {\em uncertainty} about the future link rate at each time step. If $X_t$ is highly variable, then the link rate $\mu(t)$ is highly uncertain (even with knowledge of $\mu(t-1)$). Our bound quantifies how this uncertainty imposes a limit on the achievable performance of any congestion control algorithm. For example, \Fig{model1:f_variation} shows the performance bound for two uniform distributions with different ranges. The more variable distribution imposes a worse bound on performance, shifting it upwards and to the right.% \ma{Can you use simple round numbers like U(0.3, 3) and U (0.5, 1.5) for the example. Something where it's immediately obvious what we mean by more variable.}

\subsubsection{Proof\\}
\label{ss:model1:proof_perfbound}
We first define a new quantity $\rho(t)$ in terms of $Q(t-1)$, $s(t)$ and $\mu(t-1)$ as follows,
\begin{align}
  \rho(t) =  \frac{\frac{Q(t-1)}{T} + s(t)}{\mu(t - 1)}.  
\label{eq:model1:proof_perfbound:rt}
\end{align}

In time step $t$, the total amount of data available (that can be serviced) at the link is $Q(t-1) + T \cdot s(t)$. Thus, $\rho(t)$ can be thought of as the ``load'' on the link in time step $t$ relative to the link capacity in previous time step $\mu(t-1)$.  With this new quantity, using $\mu(t) = \mu(t-1) \cdot X_t$, we can simplify Equations~\eqref{eq:setup:queue_delay_def} and ~\eqref{eq:setup:undeutilization_def} as follows: %\ma{use \eqref{} instead of \ref for equations}
\begin{align}
        q(t) &= \frac{(Q(t-1) + s(t) \cdot T - \mu(t) \cdot T)^+}{\mu(t)}= T \cdot \left(\frac{\rho(t)}{X_t} - 1 \right)^+,\nonumber\\
        U(t) &= \frac{(\mu(t) \cdot T - s(t) \cdot T - Q(t - 1))^+}{\mu(t) \cdot T} = \left(1 - \frac{\rho(t)}{X_t}\right)^+.
    \label{eq:model1:proof_perfbound:qU_rt}
    %\label{eq:model1:proof_perfbound:system}
\end{align}
Notice that both $q(t)$ and $U(t)$ depend only on $\rho(t)$ and the value of the random variable $X_{t}$. Moreover, $\rho(t)$ itself is a random variable dependent on $Q(t-1)$, $\mu(t-1)$, and the congestion control strategy (which selects $s(t)$). Crucially, $\rho(t)$ and $X_t$ are independent. We now use this independence to establish a bound on the expected underutuilization and queuing delay for a single time step $t$. But first, we need the following Lemma showing that the function $g^f$ is convex. 

%Notice that both $q(t)$ and $U(t)$ depend only on $\rho$ and the value of the random variable $X_{t}$. This implies that regardless of the state of the system in round $t-1$, by setting relative link load $\rho$ to the same value, the sender can achieve the same profile (or probability distribution) for $q(t)$ and $U(t)$. \ma{Confusing because $\rho$ obviously depends on the state in the previous round as well.} We will leverage this fact throughput our analysis.

\begin{lemma}
The function $g^f(\cdot)$ and the set $\mathbb{C}^f$ are convex for any positive PDF function $f$, where $f(x)=0$ for all $x \leq 0$.
\label{lemma:model1:convexity_of_bound}
\end{lemma}
\begin{proof}
See Appendix~\ref{app:model1:perfbound}.
\end{proof}

%\begin{lemma}
%In the model \S\ref{ss:model1:model}, for any causal congestion %control protocol, for any $t \in \mathbb{N}$, the point \{$E[q(t)]$, %$E[U(t)]$\} always lies in the convex set $\mathbb{C}^f$.
%\end{lemma}

\if 0
\begin{align}
        q(t) &= \frac{(Q(t-1) + s(t) \cdot T - \mu(t) \cdot T)^+}{\mu(t)}= T \cdot \left(\frac{\frac{Q(t-1)}{T} + s(t)}{\mu(t - 1) \cdot X_t} - 1 \right)^+\nonumber\\
        U(t) &= \frac{(\mu(t) \cdot T - s(t) \cdot T - Q(t - 1))^+}{\mu(t) \cdot T} = \left( 1 - \frac{\frac{Q(t-1)}{T} + s(t)}{\mu(t - 1) \cdot X_t}\right)^+
    \label{eq:model1:proof_perfbound:system}
\end{align}
Let's define a new variable $\rho(t)$ in terms of $Q(t-1)$, $s(t)$ and $\mu(t-1)$,
\begin{align}
  \rho =  \frac{\frac{Q(t-1)}{T} + s(t)}{\mu(t - 1)}  
\label{eq:model1:proof_perfbound:rt}
\end{align}

Intuitively, $\rho(t)$ can be thought of as the load on the link in time step $t$ relative to the link capacity in previous time step $\mu(t-1)$. Rewriting Eq. ~\eqref{eq:model1:proof_perfbound:system},
\begin{align}
    q(t) &= T \cdot \left(\frac{\rho}{X_t} - 1 \right)^+ &U(t) = \left(1 - \frac{\rho}{X_t}\right)^+
    \label{eq:model1:proof_perfbound:qU_rt}
\end{align}
\fi

To compute $E[q(t)]$ and $E[U(t)]$, we will first condition on the value of $\rho(t)$: 
\begin{align}
    E[q(t)|\rho(t)=b] &= E[T \cdot \left(\frac{b}{X_t} - 1 \right)^+] = T \cdot \int_{0}^{b}\left(\frac{b}{a} - 1 \right) \cdot f(a) \cdot da,\nonumber\\
                        %&= T \cdot \int_{0}^{z}\left(\frac{z}{x} - 1 \right) \cdot f(x) \cdot dx & \text{where } z=\rho\nonumber\\
    E[U(t)|\rho(t)=b] &= E[\left( 1 - \frac{b}{X_t}\right)^+] = \int_{b}^{\infty}\left( 1 - \frac{b}{a}\right) \cdot f(a) \cdot da.
\label{eq:model1:proof_perfbound:single_round_givenrt}
\end{align} 
Notice that Eq.~\eqref{eq:model1:proof_perfbound:single_round_givenrt} is identical to the curve $g^f$ from Eq.~\eqref{eq:model1:thm:perfbound}, i.e., the point $(E[q(t)|\rho(t)=b], E[U(t)|\rho(t)=b])$ lies on the curve $g^f$. Indeed, as we vary $b \in \mathbb{R}^+$, this point traces the entirety of $g^f$.

\if 0
%As we would expect, increasing the rate we send, in an uncertain environment, poses a direct trade-off between underutilization and queuing delay. 
%Increasing $\rho$ increases the conditional expected value of $q(t)$ and reduces the expected value of $U(t)$ . Varying $\rho$ in the above Eq. gives us the performance bound for round $t$.%curve between the two values. 

\begin{figure}[t]
     \centering
    \includegraphics[width=0.8\columnwidth]{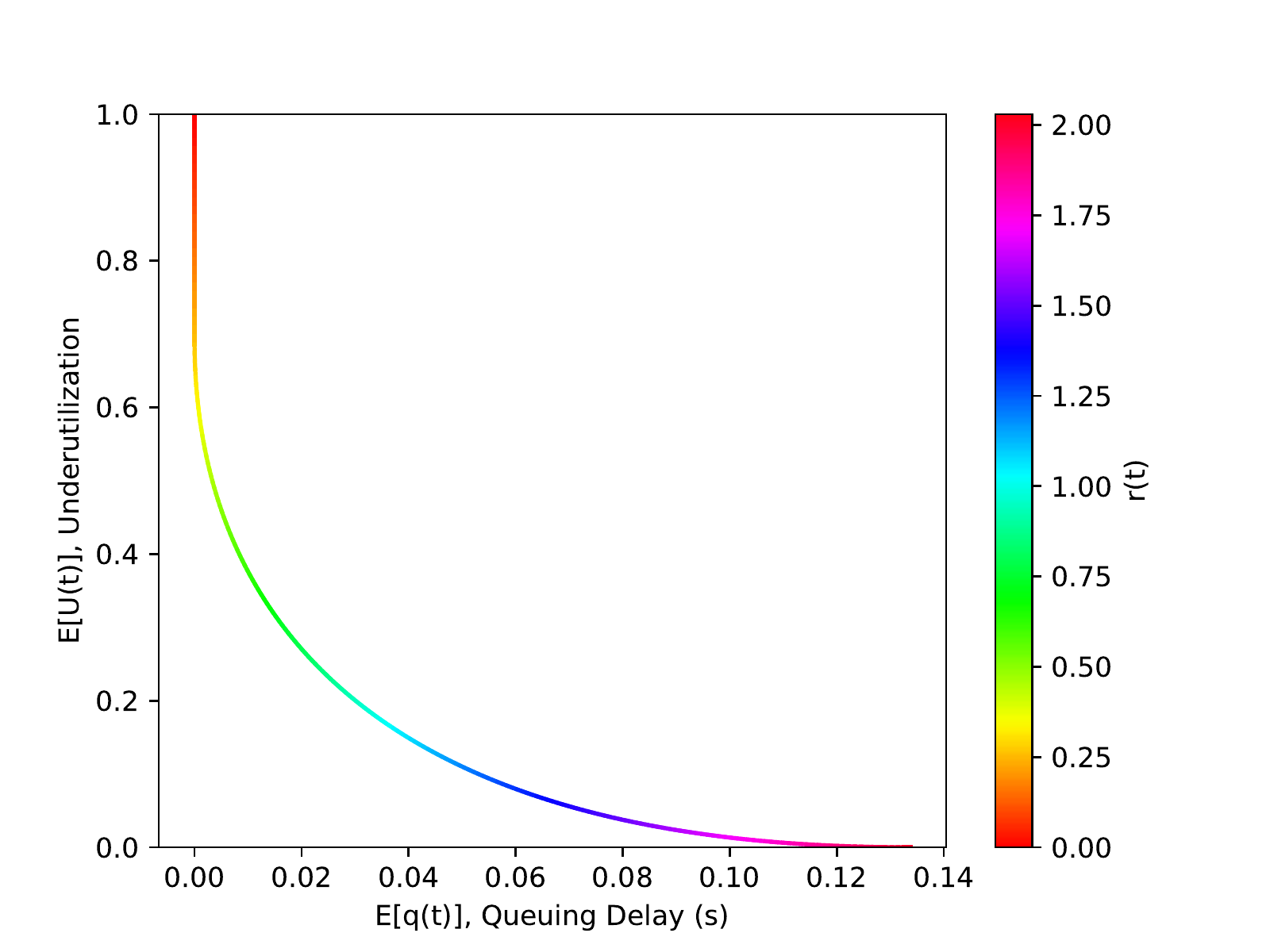}
    \vspace{-4mm}
    \caption{\small Trade-off between conditional expectation in a given time step t. $X \sim \textit{U}(0.27, 2) $, $T = 0.1$s. The color of the line corresponds to different values of $\rho=y$.}
    \label{fig:roundt_uniform}
    \vspace{-4.5mm}
 \end{figure}
\Fig{roundt_uniform} shows this trade-off curve for same distribution as \Fig{feasible_uniform}. 
\fi

%Since $\rho$ is a random variable it has a distribution dependent on the congestion strategy. 
Let the PDF of $\rho$ be $f_{r_t}$ in time step $t$. This distribution depends on the congestion control strategy and can be different for different time steps.  We have:
\begin{align}
    E[q(t)] &= E_{\rho(t)}[E[q(t)|\rho(t)]] = \int_{b=0}^{\infty}E[q(t)|\rho(t)=b]\cdot f_{r_t}(b) \cdot db,\nonumber\\
    E[U(t)] &= E_{\rho(t)}[E[U(t)|\rho(t)]] = \int_{b=0}^{\infty}E[U(t)|\rho=b] \cdot f_{r_t}(b) \cdot db.
\label{eq:model1:single_round}
\end{align}

Eq.~\eqref{eq:model1:single_round} shows that the point $(E[q(t)], E[U(t)])$ is a weighted average of the points $(E[q(t)|\rho(t)=b], E[U(t)|\rho(t)=b])$ for different values of $b$, all of which lie on the curve $g^f$. The actual value of $(E[q(t)], E[U(t)])$ depends on the distribution $f_{r_t}$, which determines the weight given to each point on $g^f$. But regardless of $f_{r_t}$, the point $(E[q(t)], E[U(t)])$ will lie in the set $\mathbb{C}^f$ since $g^f$ is convex.

%The actual value of the point $\{E[q(t)], E[U(t)]\}$ will depend on the distribution $fr^t$. Regardless of $fr^t$, the point $\{E[q(t)], E[U(t)]\}$ is a weighted average (with $fr^t(b)$ as weights) of points $\{E[q(t)|\rho=b], E[U(t)|\rho=b]\}$ with different values of $b$, all of which lie on the curve $g^f$. Since $g^f$ is convex with increasing slope, for any $fr^t \in \mathbb{R}^+ \to \mathbb{R}^+ $ (i.e., for any causal congestion control strategy), the point $\{E[q(t), E[U(t)\}$  will lie in the set $\mathbb{C}^f$. Lemma proved.

Finally, let us consider the expected queuing delay and underutilization over multiple time steps: 
\begin{align}
    E[\bar{q}(t)] &= E\left[\frac{\sum_{i=1}^{t}q(i)}{t}\right] = \frac{\sum_{i=1}^{t}E[q(i)]}{t}, & E[\bar{U}(t)] &= E\left[\frac{\sum_{i=1}^{t}U(i)}{t}\right] = \frac{\sum_{i=1}^{t}E[U(i)]}{t}.
\end{align}
The point $(E[\bar{q}(t)], E[\bar{U}(t)])$ is the average of the points $(E[q(i)], E[U(i)])$ for $1 \leq i \leq t$. Since all these points belong to the convex set $\mathbb{C}^f$, their average must also lie in the set $\mathbb{C}^f$. This completes the proof. \qed

 \subsection{Is the Performance Bound Achievable?}
 \label{ss:model1:achievable_bound}
The analysis in the previous section provides insight into how to achieve a particular point on performance bound. The key observation is that we need to maintain the same ``load'' $\rho(t)$ in every time step. Recall that $\rho(t)$ determines the operating point $(E[q(t)], E[U(t)])$ achieved on the curve $g^f$ in each time step. If $\rho(t)$ changes from round to round, then over multiple time steps we will end up at a sub-optimal point in the interior of $\mathbb{C}^f$. However, if the congestion control algorithm can keep $\rho(t)$ constant over time, it can achieve an optimal tradeoff between underutilization and queuing delay.

Setting $\rho(t) = C$ (for some constant $C$) in Eq.~\eqref{eq:model1:proof_perfbound:rt} suggests the following control law: 
\begin{align}
    s(t) &= C \cdot \mu(t-1) - \frac{Q(t-1)}{T}, \quad \forall t \in \mathbb{N}.
    \label{eq:simple-rule}
\end{align}
%where $E[q(t)|\rho(t)=C]=x$.%$C \in R+$. 
The first term selects a sending rate proportional to the current link capacity $\mu(t-1)$. This is an intuitive choice since the next capacity $\mu(t)$ is also proportional to $\mu(t-1)$ (albeit with an unknown multiplicative factor). The second term deducts a rate based on the current queue backlog from the sending rate. This deduction corresponds to the rate necessary to drain the current queue backlog within one time step. The parameter $C$ controls the tradeoff between underutilization and queuing delay. As $C$ increases, the congestion control sends more aggressively, favoring high utilization at the expense of queuing delay. This corresponds to moving to the right on the curve $g^f$ (see Fig.~\ref{fig:model1:feasible_uniform}).

However, it may not always be possible to strictly follow the above rule. In Eq.~\eqref{eq:simple-rule}, $s(t)$ can become negative, but of course the sending rate must be non-negative. For example, it might so happen by chance that the link capacity drops rapidly over a number of consecutive time steps, leading to a large queue build up. If in any time step $t$, $Q(t-1)$ exceeds $T \cdot C \cdot \mu(t-1)$, then the sender will have to pick a negative sending rate to follow $\rho(t) = C$, which is not possible. 

This argument shows that depending on the distribution $f$ of the link variation and the value $C$, it may not be possible to achieve some  points on the curve $g^f$. This is more likely for high values of $C$, which makes the controller more aggressive, and for distributions $f$ where the link capacity can drop significantly each time step. The following proposition gives sufficient necessary conditions on $C$ such that the sender can always follow the rule in Eq.~\eqref{eq:simple-rule} and enforce $\rho(t)=C$.

\begin{proposition}
In the MIF model, it is possible to follow the strategy $\rho(t) = C$,  $\forall t \in \mathbb{N}$ iff
\begin{align}
    C &\leq \frac{X_{min}}{1 - X_{min}},
    \label{eq:prop:model1:achievable_bound}
\end{align}
where $X_{min}$ is the minimum value of the random variable $X$ ($P(X<X_{min}) = 0$). 
\label{prop:model1:achievable_bound}
\end{proposition} 
\begin{proof}
Please see Appendix~\ref{app:model1:prop}
\end{proof}

\if 0
\begin{proof}

\ma{I think this proof could go into the appendix. This proposition isn't central to the story.} \pg{Will move.}

We will prove this by induction. We will show that if the above condition holds then with the policy $\rho(t) = C$, $Q(t)  \leq T \cdot C \cdot \mu(t)$ $\forall t \in \mathbb{N}$.  

The base case holds because $Q(0) = 0$. Lets assume that $Q(t - 1) 
 \leq T \cdot C \cdot \mu(t-1)$, then we can pick a non negative $s(t)$, such that $\rho(t) = C$. Then in the next round $t$,
 \begin{align}
     Q(t) &= (Q(t-1) + s(t) \cdot T- \mu(t) \cdot T)^+,\nonumber\\
          &= (T \cdot C \cdot \mu(t - 1) - T \cdot \mu(t))^+,\nonumber\\
          &=  T \cdot C \cdot \mu(t) \cdot (\frac{1}{X_t} - \frac{1}{C})^+, \nonumber\\
          &\leq T \cdot C \cdot \mu(t) \cdot (\frac{1}{X_{min}} - \frac{1}{C})^+,\nonumber\\
          &\leq T \cdot \mu(t) \cdot C,
 \end{align}
 where in the last step we used Eq.~\eqref{eq:prop:model1:achievable_bound}.
\end{proof}
\fi
%Notice that this condition is dependent on the value of $X_{min}$. Lower $X_{min}$ implies that the link capacity can reduce by a higher amount and it is more likely for the sender to build an excessive queue.

This analysis shows us that the performance bound need not be tight. For any $f(\cdot)$, the points $(x, g^f(x))$ where $E[q(t)|\rho(t)=C] = x$ but $C$ does not meet the condition above might not be achievable. %\ma{This implies that the condition in Prop. 1 is actually an ``if only if'' condition. But we didn't show that right? Leaving this here until we resolve this comment. } \pg {I think if the above condition is violated with some probability the sender \emph{will} have to pick a higher $r(t) = C_1$ > $C$ some time. But if a portion of curve $g^f$ is a straight line then it might still be possible to achieve the desired trade-off exactly.}

\subsection{Optimal Control Law}
\label{ss:model1:mdp}

What happens when the above condition is not met and $Q(t-1)$ can exceed $T \cdot C \cdot \mu(t-1)$? A natural question to ask is: in the event of excessive queuing, should the sender just pick the minimum value for $\rho(t)$ such that $s(t)$ is non-negative? This corresponds to $\rho(t) = \max(C, \frac{q(t-1)}{T})$, which implies a control law of the following form:
\begin{align}
    s(t) &= \left(C \cdot \mu(t-1) - \frac{Q(t-1)}{T}\right)^+, \quad \forall t \in \mathbb{N}.
\end{align}

We now show that the above rule can indeed achieve an optimal tradeoff between underutilization and queuing delay. To this end, we formulate the congestion control task as a Markov Decision Process~\cite{sutton2018reinforcement} (MDP) and show that the above rule is an optimal policy for this MDP. 

The MDP is defined as follows. The {\em state} at time step $t$ is given by $(q(t-1), \mu(t-1))$. The congestion control ``agent'' observes this state and selects an {\em action}\,---\,the sending rate $s_t \geq 0$. The state then transitions to the next state $(q(t), \mu(t))$, and the agent incurs a weighted {\em cost} $(w\cdot q(t) + U(t))$ for this time step, depending on the relative importance of low queuing delay versus low underutilization. To see that this is an MDP, recall that the transitions occur according to:
\begin{align}
    \mu(t)&=\mu(t-1)\cdot X_t \quad \text{(per the MIF model)}, \nonumber\\
    q(t)&=T\cdot(\rho(t)/X_t-1)^+ \quad \text{(per Eq.~\eqref{eq:model1:proof_perfbound:qU_rt})}, \nonumber
\end{align} 
Therefore the next state is determined by the current state and the congestion control action, since $\rho(t) =  \frac{q(t-1)\mu(t-1)+T\cdot s(t)}{T\cdot \mu(t-1)}$  (Eq.~\eqref{eq:model1:proof_perfbound:rt}).

The agent uses a control policy $\pi$ mapping the current state to a probability distribution over possible actions (i.e., sending rates $s(t)\geq 0$). The policy can be deterministic or stochastic. The goal of the agent is to minimize the following objective function over all policies:
\begin{align}
     J(\pi) = E\left[\sum_{t=1}^{\infty}\gamma^{t-1} \Big( w \cdot q(t) + U(t)] \Big)\right].
\label{eq:model1:mdp:objective}
\end{align}
Here we use a standard discounted cost formulation, with discount factor $\gamma \in (0,1)$.

\begin{theorem}
    In the MDP defined above, the optimal control policy that minimizes $J(\pi)$ takes the form:  
    \begin{align}
        s(t) = \left( C(w, \gamma) \cdot \mu(t-1) - \frac{Q(t-1)}{T} \right)^+,  
        \label{eq:thm:model1:mdp}
    \end{align}
where $C(w,\gamma)$ is a constant. 
    \label{thm:model1:mdp}
\end{theorem}
\begin{proof}
Define the optimal value function~\cite{sutton2018reinforcement} for our MDP as:
\begin{align}
    V(q,\mu) = \min_\pi E\left[\sum_{t=1}^{\infty}\gamma^{t-1} \Big( w \cdot q(t) + U(t)] \Big) \;\middle|\; q(0) = q, \mu(0) = \mu, s(t) \sim \pi(\cdot|q(t-1),\mu(t-1))\right]. 
\end{align}
This is the minimum possible discounted expected sum of costs, starting from initial state $(q,\mu)$ and following policy $\pi$, over all policies $\pi$.
\footnote{We only consider $f(\cdot)$ for which $V(q,\mu)$ exists.} 
%Otherwise, if $V(q,\mu) \to \infty$, then for the state $(q, \mu)$, every control law has the same value function ($\infty$) and the optimal control law is not defined. \ma{As defined, we only talk about the optimal value function, which doesn't depend on a specific policy, so this could be confusing. Can we this say something like: if $V(q,\mu) = \infty$ for some states $(q, \mu)$, then it is easy to show that $V(q,\mu) = \infty$ for {\em every}) state $(q, \mu)$ with $q \neq 0$, and therefore the problem becomes meaningless.} \pg{I don't think this statement is true.} \ma{but i also don't see why $V(q,\mu)$ infinity for one state can imply anything about the optimal control law being undefined. hypothetically the optimal law would try to avoid that state!} \pg{I was wondering if we could formulate the objective function in itself as a function of the state of the system. then this goes away. If we do this we might have to get rid of the assumption Q(0) = 0.} \pg{alternatively we could get rid of the footnote altogether.}} %See discussion in Appendix~\ref{appxxx}. $\sum_{k\geq 1} \gamma^k E[(1/X)]^k  < $ => $\gamma E(1/X) < 1 $}

The optimal value function satisfies the following Bellman Equation:
\begin{align}
    V(q,\mu) = \min_{s(1) \geq 0} E \left[ w\cdot q(1) + U(1) + \gamma V(q(1), \mu(1)) \;\middle|\; q(0) = q, \mu(0) = \mu \right]. 
\end{align}
The condition $s(1) \geq 0$ is equivalent to $\rho(1) \geq q/T$. Rewriting,  
\begin{align}
    V(q,\mu) = \min_{\rho(1) \geq \frac{q}{T}} E \left[ w\cdot q(1) + U(1) + \gamma V(q(1), \mu(1)) \;\middle|\; q(0) = q, \mu(0) = \mu \right]. 
\end{align}

\begin{lemma}
$V(q,\mu)$ is a convex function of $q$ and does not depend on $\mu$. 
\label{lemma:model1:combined_mdp_lemma}
\end{lemma}
\begin{proof}
See Appendix~\ref{app:model1:combined_mdp_lemma}.
\end{proof}
Since $V$ does not depend on $\mu$, henceforth we will write the optimal value function as $V(q)$. To complete the proof, we exploit the fact that $q(t)$ and $U(t)$ only depend on the value of $\rho(t)$ and define a helper function $W(\cdot)$ as follows,
\begin{align}
    W(\rho) &= E\left[ w \cdot q(1) + U(1) + \gamma \cdot V(q(1)) \;\middle|\; \rho(1) = \rho\right], & V(q) &= \min_{\rho \geq \frac{q}{T}} (W(\rho)).
    %W_{t}(\rho(t))  &= \Bigg( w \cdot E[q(t)] + E[U(t)] + \gamma \cdot \int_{0}^{\infty}f(a) \cdot V_{t+1}(T \cdot (\frac{\rho(t)}{a} - 1)^+) \cdot da\Bigg), & \forall t \leq n, \nonumber\\
    %V_{t}^2(\rho(t)) &= \Bigg( w \cdot E[\frac{q(t)}{T}] + E[U(t)] \nonumber\\ 
    %&+ \int_{0}^{\infty}f(a) \cdot V_{t-1}^1((\frac{\rho(t)}{a} - 1)^+) \cdot da\Bigg)\nonumber &, \forall t<n\nonumber\\
    \label{eq:model1:mdp_helper_def}
\end{align}
$W(\rho)$ is the minimum expected cost when we restrict the first action to $\rho(1) = \rho$ and then act optimally from that step onward.  

\begin{lemma}
    $W(\rho)$ is a convex function in $\rho$.
\end{lemma}
\begin{proof}
The proof directly follows from proof of Sub lemma~\ref{sublemma:model1_helper} in Appendix~\ref{app:model1:combined_mdp_lemma}.
\end{proof}

Let the minimum value of $W(\cdot)$ occur at $\rho^*$. Then, the convexity of $W(\cdot)$ combined with Eq.~\eqref{eq:model1:mdp_helper_def} establishes that the optimal control law is of the form $\rho(t) = \max(\rho^*, q(t-1)/T)$. To see this, note that the minimizer of $W(\rho)$ in Eq.~\eqref{eq:model1:mdp_helper_def} is $\rho = q/T$ if $\rho^* < q/T$, and $\rho = \rho^*$ otherwise. Furthermore, note that $\rho^*$ depends on $\gamma$ and $w$. Defining $C(w, \gamma) = \rho^*$ and using Eq.~\eqref{eq:model1:proof_perfbound:rt} we get the optimal control law from the Theorem. 

\end{proof}

\subsection{Extension: Incorporating Prediction for Link Capacity}
\label{s:model2}
%\subsection{Model}
\label{ss:model2:model}
%\pg{Merge PMIF into MIF?}
It is natural to ask what if the sender could predict the link capacity with some degree of accuracy? %Can such a prediction improve performance? 
Like the MIF model, does there exist a fundamental performance bound? How can the sender use the prediction and what is the optimal control law? To answer these questions, we will now extend the MIF model to include prediction. We refer to this new model as Prediction-based Multiplicative I.I.D. Factors (PMIF) model. In this model, at the start of every time step, the sender has access to a prediction of the link capacity (e.g., based on past behavior or underlying link information such as signal strength). We denote the link capacity prediction for time step $t$ by $Pred(t-1)$ .

\if 0
\pg{Do we need this? We should just consider this as causal.}
\begin{definition}{\textit{Predictive Congestion Control protocols}---} A predictive congestion control protocol decides the sending rate in time step $t$ based on both observations in earlier time steps along with the prediction for the link capacity. I.e., $s(t)$ is dependent on $\cup_{i=0}^{t-1}\{\mu(i)\}$, $\cup_{i=0}^{t-1}\{U(i)\}$, $\cup_{i=0}^{t-1}\{q(i)\}$, and $Pred(t-1)$.
\end{definition}
\fi 

In the PMIF model, the predicted link capacity can deviate from the true link capacity. We use random variable $X^p$ to represent how far is the prediction relative to the true link capacity. In other words, the random variable $X^p$ is a measure of uncertainty in the true link capacity given the prediction. %Moreover, we assume this uncertainity for each round is the same.
Formally, the uncertainty in the link capacity is governed by,
\begin{equation}
    \mu(t) = Pred(t - 1) \cdot X_{t}^p,
    \label{eq:model2:model}
\end{equation}
where  $X_{i}^p$ is an I.I.D random variable with $f^p(\cdot)$ as its probability density function (PDF) and $P(X_{i}^p <= 0) = 0$. %The uncertainty is link capacity relative to its expected value has the same characteristics in every time step.%We assume $X_{i}^t$ is independent of $s(t), Q(t-1), \mu(t-1), Pred(t-1)$.

%As before, we derive a performance bound between $E[\bar{q}(t)]$ and $E[\bar{U}(t)]$ (\S\ref{ss:model2:performance_bound}). We use this analysis to draw insights about  the optimal controller (\S\ref{ss:model2:achievable_bound}). Next, we a MDP (\S\ref{ss:model2:mdp}). We sketch some implications of our analysis and validate our findings (\S\ref{ss:model2:implications}). 

%\pg{As before, we derive a performance bound between $E[\bar{q}(t)]$ and $E[\bar{U}(t)]$ (\S\ref{ss:model2:performance_bound}). We use this analysis to draw insights about the optimal controller and derive the optimal control by posing the problem as a particular MDP (\S\ref{ss:model2:mdp}). We sketch some implications of our analysis and validate our findings (\S\ref{ss:model2:implications}).}

Note that, we do not attempt to come up with a predictor ourselves; rather, we only analyse the impact of prediction error on the design of congestion control for time-varying links.

\subsubsection{Performance Bound}
\label{ss:model2:performance_bound}

\begin{corollary}

In the PMIF model, for any causal congestion control protocol, for any $t \in \mathbb{N}$, the point ($E[\bar{q}(t)]$, $E[\bar{U}(t)]$) always lies in a convex set $\mathbb{C}^{f^p} \triangleq \{(x,y) | y>=g^{f^p}(x)\}$, where $g^{f^p}(\cdot)$ is defined by Eq.~\eqref{eq:model1:thm:perfbound}.
\label{cor:model2:perfbound}
\end{corollary} 

The performance bound is dependent on the characteristics of the random variable $X^p$ and consequently the accuracy of prediction. The better the prediction, the better the performance bound.

\begin{proof}
To prove this corollary, we will use the techniques from the proof of Theorem ~\ref{thm:model1:perfbound}. Using $\mu(t) = Pred(t-1) \cdot X^p_t$, we can rewrite $q(t)$ and $U(t)$ as follows,
%the equation describing underutilization and queuing delay as follows,

\begin{align}
    q(t) &= T \cdot \left(\frac{\frac{Q(t-1)}{T} + s(t)}{Pred(t - 1) \cdot X_{t}^p} - 1 \right)^+, &  U(t) = \left( 1 - \frac{\frac{Q(t-1)}{T} + s(t)}{Pred(t - 1) \cdot X_{t}^p}\right)^+.\nonumber\\
    \label{eq:prediction_system}
\end{align}

Similar to $\rho(t)$ from Eq.~\eqref{eq:model1:proof_perfbound:rt}, we define a new quantity $\rho^p(t)$ as follows,
\begin{align}
    \rho^p(t) &=  \frac{\frac{Q(t-1)}{T} + s(t)}{Pred(t - 1)}.
    \label{eq:prediction_rt}
\end{align}

Intuitively, $\rho^p(t)$ can be thought of as the load on link in time step $t$ relative to the prediction of the link capacity ($Pred(t-1)$). We can rewrite $q(t)$ and $U(t)$ in terms of $\rho^p(t)$,

\begin{align}
    q(t) &= T \cdot \left(\frac{\rho^p(t)}{X_{t}^p} - 1 \right)^+, &U(t) &= \left(1 - \frac{\rho^p(t)}{X_{t}^p}\right)^+.
    \label{eq:prediction_qU_rt}
\end{align}

Notice that the equation above is analogous to Eq.~\eqref{eq:model1:proof_perfbound:qU_rt}. Thus, we can apply the same techniques as in the proof for Theorem ~\ref{thm:model1:perfbound} to show that the point $(E[q(t)|\rho^p(t)=b], E[U(t)|\rho^p(t)=b])$ lies on the curve $g^{f^p}(\cdot)$. Both the points $(E[q(t)], E[U(t)])$ and $(E[\bar{q}(t)], E[\bar{U}(t)])$ lie in the convex set $\mathbb{C}^{f^p}$. %Corollary proved.
\end{proof}

\subsubsection{Optimal Control Law\\}
\label{ss:model2:mdp}
The analysis in the previous section shows us that we can achieve a particular point $(x, g^{f^p}(x))$ on the performance bound by following a strategy of the form,
\begin{align}
    \rho^p(t) &= C^p, &
    s(t) &= C^p \cdot Pred(t-1) - \frac{Q(t-1)}{T}, & \forall t \in \mathbb{N},
    \label{eq:model2:ab_guess_control}
\end{align}
where $E[q(t)|\rho^p(t)=C^p] = x$. The key is keeping the link load relative to the prediction of the link capacity the same in every time step, regardless of the state of the system. However, similar to \S\ref{ss:model1:achievable_bound}, following this strategy might not be possible for all values of $C^p$. In the event of excessive queuing ($Q(t-1) > T \cdot C^p \cdot Pred(t-1)$), for $\rho^p(t)= C^p$, the sender needs to pick a negative value for $s(t)$ (from Eq. \eqref{eq:prediction_rt}) which is not possible. Is simply following $s(t) = (C^p \cdot Pred(t-1) - \frac{Q(t-1)}{T})^+$ optimal in just situations?

%\pg{is it even needed? We should only include math for which can de experiments}
%What happens when the above condition is not met? Will the following variant of the strategy in Proposition~\ref{prop:model2:achievable_bound} -- $\rho^p(t) = max(C^p, \frac{Q(t-1)}{Pred(t-1) \cdot T})$ or $s(t) = (C^p \cdot Pred(t-1) - \frac{Q(t-1)}{T})^+$ -- be optimal? 

Similar to the MIF model, we can formulate the congestion control task as a MDP and show that the above control law is optimal. The key change in the MDP from the MIF model is that the state in time step $t$ is now given by ($q(t-1), \mu(t-1), Pred(t-1)$). The optimal control law for the MDP depends on the dynamics of the predictions. 
%At a high level, the dynamics of the prediction affect how the link capacity varies from one time step to the next and consequently the dynamics of (excessive) queue build up.
Here, we only derive the exact optimal control law for a simple model governing the evolution of $Pred$. In this MDP, the evolution of $Pred$ is analogous to evolution of $\mu$ from the MIF model. Formally,%\tom{I would drop the sentence "at a high level" - its obvious and distracts from the next sentence which is that you are only solving the optimal control law for a specific case.}
%Additionally, we must assume some dynamics for the predictions in the MDP. This dynamics impacts the control law \ma{This is the tricky sentence. How? is it relevant to constants or the form of the controller. Right now $f^{pred}$ doesnt seem to be showing up at all in the result. I am assuming it impacts the constant.}. For this MDP, we use a simple model governing the evolution of $Pred$,
\begin{align}
    Pred(t) = Pred(t-1) \cdot X^{pred}_{t},    
\end{align}
where  $X_{i}^{pred}$ is an I.I.D random variable with $f^{pred}$ is its probability density function (PDF) and $P(X_{i}^{pred} \leq 0) = 0$. Other state transitions are governed by,
\begin{align}
    \mu(t)&=\mu(t-1)\cdot X^p_t \quad \text{(per the PMIF model)}, \nonumber\\
    q(t)&=T\cdot(\rho^p(t)/X^p_t-1)^+ \quad \text{(per Eq.~\eqref{eq:prediction_qU_rt})}, \nonumber
\end{align}

The goal of the agent is to minimize the objective function ($J(\cdot)$ from Eq.~\eqref{eq:model1:mdp:objective}) over all policies ($\pi$).

\if 0
\begin{align}
     J^p(\pi^p) = E\left[\sum_{t=1}^{\infty}\gamma^{t-1} \Big( w \cdot q(t) + U(t)] \Big)\right], 
\end{align}
where $\gamma \in (0,1)$ is the discount factor. 
\fi

\begin{corollary}
    In the MDP defined above, the optimal control policy that minimizes $J(\pi)$ takes the form:  
    \begin{align}
        s(t) = \left( C^p(w, \gamma) \cdot Pred(t-1) - \frac{Q(t-1)}{T} \right)^+,  
        \label{eq:cor:model2:mdp}
    \end{align}
where $C^p(w,\gamma)$ is a constant. 
    \label{cor:model2:mdp}
\end{corollary}
\begin{proof}
Please see Appendix~\ref{app:model2:cor_mdp}
\end{proof}

\section{Implications and Validation}
\label{s:implications}
In this section, we discuss some of the implications of the formal analysis given above.

\subsection{MIF Model}
\label{ss:model1:implications}
%\pg{Regenerate graphs. Feedback on the legend name appreciated.}
 First, we validate our work against real cellular traces (Verizon LTE Uplink and Downlink) taken from earlier work~\cite{mahimahi}. Recall that the cellular traces need not follow our model. For comparison, we also generate a synthetic trace based on the MIF model with $X \sim e^{\textit{U(-1, 1})}$.
 %\ma{confusing organization of sentences. i can't tell if we mean the next two sentences are only for the synthetic traces and then the stuff that comes after that is for only the cellular traces or what.} 
For these three traces, we plot the expected underutilization ($E[\bar{U}]$) versus the expected normalized queuing delay ($E(\bar{q})$). We compare the performance bound from Theorem~\ref{thm:model1:perfbound} and the optimal control law from Theorem~\ref{thm:model1:mdp}. 
We use the empirical distribution of values $\frac{\mu(t)}{\mu(t-1)}$ from the traces as our PDF $f$ to compute the performance bound and we simulate the optimal control law for different values of $C(w, \gamma)$.

For the cellular traces, we also show the performance for various previously proposed congestion control algorithms. 
We consider Cubic+Codel~\cite{cubic, CoDel}, XCP~\cite{xcp}, ABC~\cite{abc, goyal2017rethinking}, Sprout~\cite{sprout}, and Copa~\cite{copa} as representing the best performing alternatives. 
For the existing algorithms, we use
Mahimahi to emulate the traces~\cite{mahimahi}.\footnote{The dynamic range of the link capacity in the synthetic trace is not supported by Mahimahi, and so we omit the comparison with prior work for that trace.}%\tom{The dynamic range of the synthetic trace exceeds Mahimahi's design parameters, and so we omit the comparison with prior work for that trace.}}. 
For Copa and Sprout, we use the authors' UDP implementations.
For the Cubic endpoints, we use the Linux TCP implementation.
For Codel, XCP, and ABC, we use the implementation and configuration parameters from ~\cite{abc}. 
To make these systems comparable to the derived bounds, we compute the queuing delay once per baseline round trip.
\if 0
For these traces, we plot the expected underutilization ($E[\bar{U}]$) versus the expected normalized queuing delay ($E(\bar{q})$). We compare the performance bound from Theorem~\ref{thm:model1:perfbound}, the optimal control law from Theorem~\ref{thm:model1:mdp}, and various previously proposed congestion control algorithms. 
For the cellular traces, we use the empirical distribution of values $\frac{\mu(t)}{\mu(t-1)}$ from the traces as our PDF $f$ to compute the performance bound; we simulate the optimal control law for different values of $C(w, \gamma)$; and we consider Cubic+Codel~\cite{cubic, CoDel}, XCP~\cite{xcp}, ABC~\cite{abc, goyal2017rethinking}, Sprout~\cite{sprout}, and Copa~\cite{copa} as representing the best performing alternatives. We emulate the traces using the Mahimahi network emulation tool~\cite{mahimahi}. For Copa and Sprout, we use the authors' UDP implementations.
For the Cubic endpoints, we use the Linux TCP implementation.
For Codel, XCP, and ABC, we use the implementation and configuration parameters from ~\cite{abc}. 
To make these systems comparable to the derived bound, we compute the queuing delay once per baseline round trip.
\fi
% \tom{are we including cubic? if not, presumably we aren't using the Linux TCP implementation?}\pg{We have cubic+Codel, end point is running cubic while the bottleleck router is using the codel scheme.}
% 

\Fig{model1:basic} shows our main result: it is possible to achieve a smooth tradeoff across a large range of
underutilization and queuing delay. There is a significant gap between the five 
existing congestion control algorithms and the performance of the optimal control law on the Verizon traces.
Note that the formatting differs from \Fig{intro_perf} in that the x-axis is linear rather than log scale. 
For the cellular
traces, there is only a small gap between the predicted performance bound and the optimal control law.  This difference
is larger in the case of the synthetic trace. To achieve very low levels of underutilization on the synthetic trace, 
the control law must be very aggressive, sometimes resulting in large queues. As discussed in \S\ref{ss:model1:achievable_bound}, achieving the bound would require a negative sending rate to drain these queues in one time step, which isn't possible.  For the cellular traces, this rarely happens. %\ma{can we say somehting about why we didn't run existing congestion controllers on the synthetic trace? for example, if you did and found that they all struggle to keep the queues from blowing up, that's actually interesting because the optimal control law apparently still does a decent job of this} \pg{Synthetic trace went to bandwidths that mahimahi could not support both high and low.} \ma{is the yellow curve a calculation or actual protopcol run on mahimahi for the verizon links? basically do we know that the performance we're talking aout is actually achievable?}\pg{yellow curve is a simulation of the optimal control law. }%actual achievable performance}

%\pg{Say in practice the performance bound is tight?}

%\pg{Should we say optimizing for E$[\bar{q}]$ and $E[\bar{U}]$ corresponds to $\gamma \to 1$ in the MDP?}

% For each trace, we plot: 1) The performance bound obtained in Theorem~\ref{thm:model1:perfbound}. For cellular traces, we use the distribution of values $\frac{\mu(t)}{\mu(t-1)}$ from the trace as our PDF $f$ to compute the bound;  2) The performance when following the optimal control law obtained in Theorem~\ref{thm:model1:mdp}. We simulate the optimal control law for different values of $C(w, \gamma)$; \cut{The queuing delay and underutilization evolution is governed as per Equations~\eqref{eq:setup:queue_delay_def} and \eqref{eq:setup:undeutilization_def} } 3) Performance of a few existing schemes such as Cubic+Codel~\cite{cubic, CoDel}, XCP~\cite{xcp}, ABC~\cite{abc, goyal2017rethinking}, Sprout~\cite{sprout}, and Copa~\cite{copa}. Here, we emulate the traces using Mahimahi a network emulation tool. For Copa and Sprout, we use the respective author's UDP implementation. For Cubic, we use the linux TCP implementation. For Codel, XCP, and ABC, we use the implementation from ~\cite{abc}. %\pg{include verus?} for this model. There is a fundamental trade-off between underutilization and queuing delay. In particular, we see that achieving very low (close to 0) queuing delay comes at the cost of heavy underutilization and vice-versa.

  \begin{figure}[tb]
 \centering
    \begin{subfigure}[h]{0.32\textwidth}
        \includegraphics[width=\textwidth]{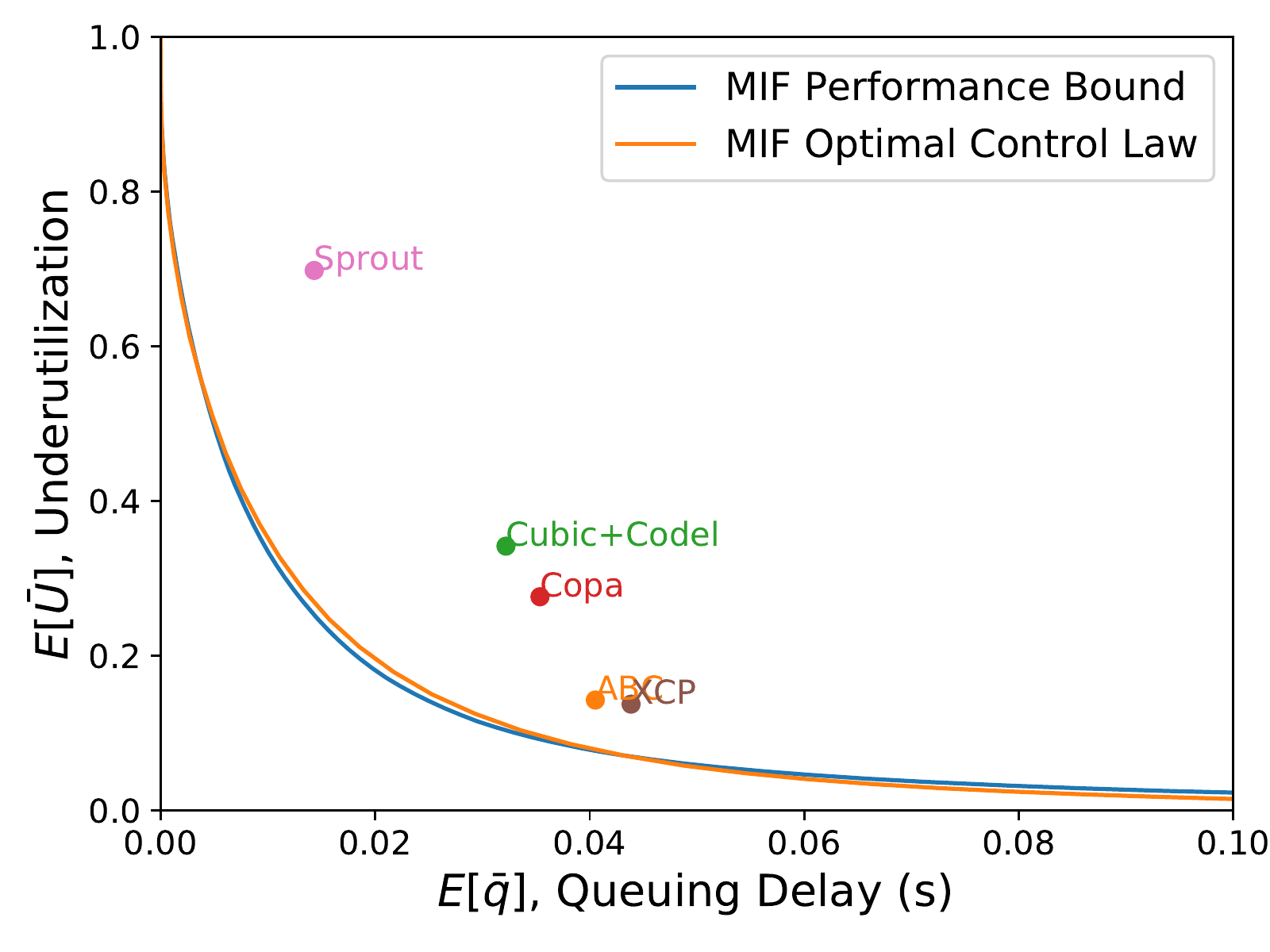}
        \vspace{-6mm}
        \caption{Verizon Downlink}
        \label{fig:model1:basic:verizon_down}
    \end{subfigure}
    \begin{subfigure}[h]{0.32\textwidth}
        \includegraphics[width=\textwidth]{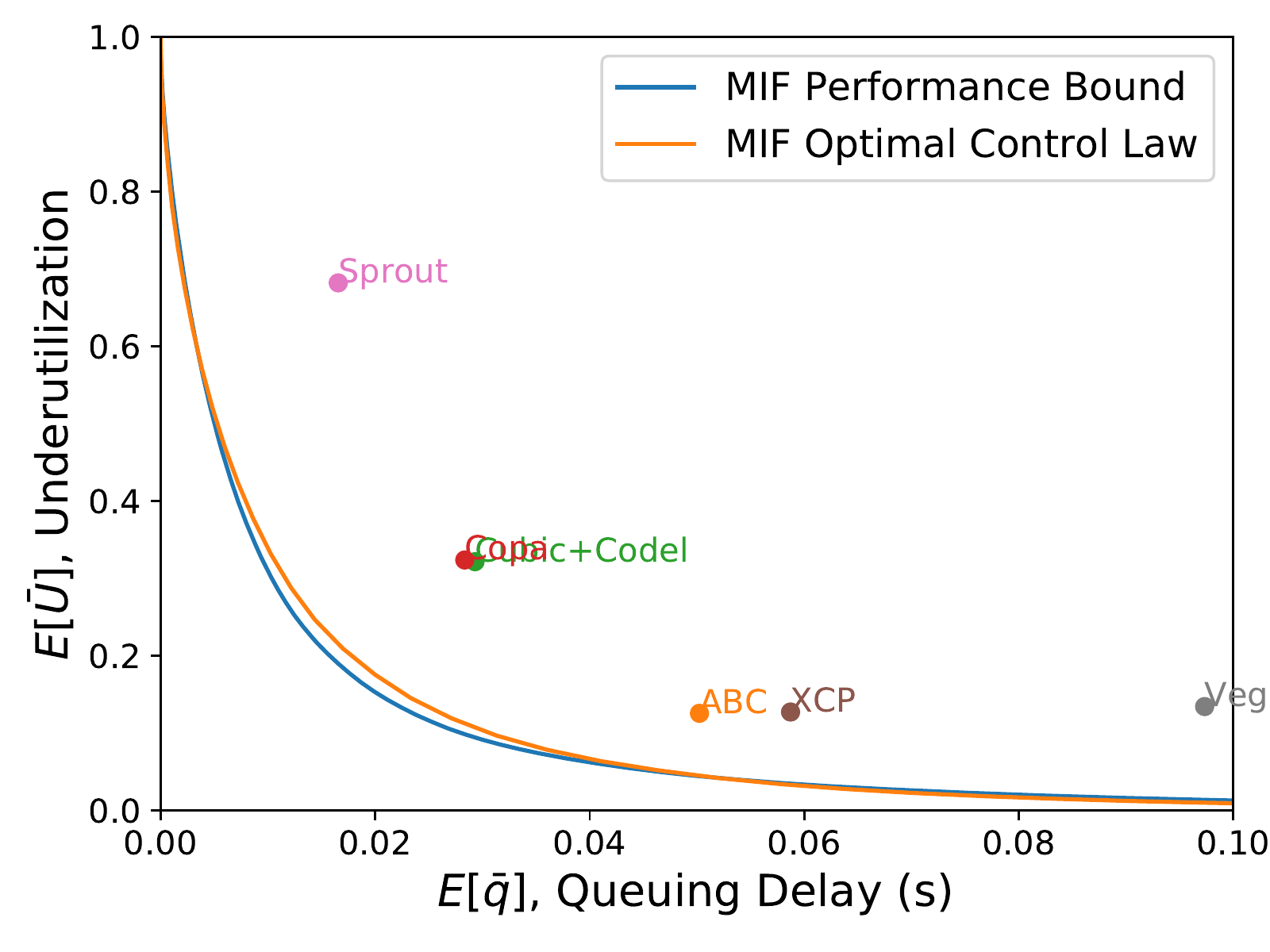}
        \vspace{-6mm}
        \caption{Verizon Uplink}
        \label{fig:model1:basic:verizon_up}
    \end{subfigure}
    \begin{subfigure}[h]{0.32\textwidth}
        \includegraphics[width=\textwidth]{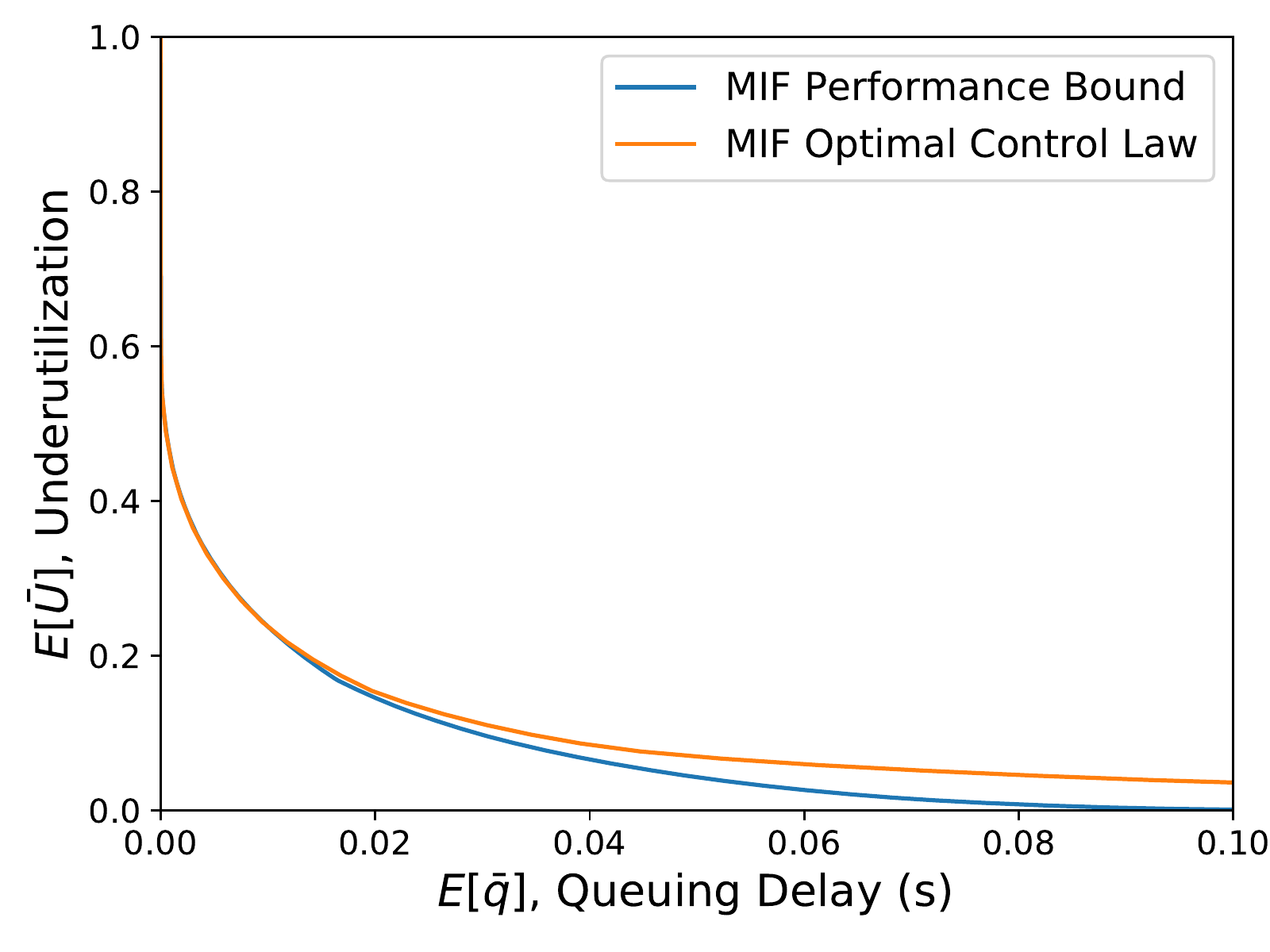}
         \vspace{-6mm}
        \caption{Synthetic Trace}
        \label{fig:model1:basic:generated}
    \end{subfigure}
    \vspace{-4mm}
    \caption{\small {\bf MIF model performance bound and performance curve for the optimal control law.}}
    \label{fig:model1:basic}
    \vspace{-4mm}
\end{figure} 

%\subsection{Implications for Design of the Control Law}
\smallskip
\noindent
\textbf{Explicit versus implicit feedback:} 
To compute the optimal sending rate, we need the link capacity ($\mu(t-1)$) and queue size ($Q(t-1)$) in the previous time step. Both these quantities are known at the link in time step $t-1$. 
% Therefore, the link itself can compute the optimal sending rate and explicitly communicate it to the sender (via packet headers).
The link can communicate these quantities explicitly to the sender for computing the optimal sending rate. However, explicit feedback is hard to deploy in many scenarios and hasn't seen much adoption in wireless networks~\cite{abc}. Without explicit feedback, the sender needs to infer the link capacity ($\mu(t-1)$) indirectly. This inference is straightforward when there is no link underutilization (the packet arrival rate at the receiver can serve as a measure for the link capacity). However, when the link is underutilized, it can be harder to infer the link capacity at either endpoint~\cite{abc}. 
%For example, if the queue is large enough, the optimal control law requires the sender to stop transmitting entirely until the queue drains in expectation.  \ma{what is the purpose of the last sentence? i don't see the connection to the rest of the paragraph}
% \pg{Do we need this line, the second half is likely wrong? With varying capacity, the link can become frequently underutilized when the link capacity increases, worsening this problem.}
Because implicit feedback introduces uncertainty in the link capacity ($\mu(t)$) at the sender, the resulting performance may be sub-optimal. We see this effect in \Fig{model1:basic}. Explicit schemes like ABC and XCP which compute feedback based on the link capacity outperform end-to-end schemes like Copa and Sprout that infer the link capacity indirectly. 
%\ma{i'm ok with this paragraph but it's also ok to cut it. the abc paper already made this point. also the figure isn't very clear about the explcit schemes being that much better}\pg{fine either way. I thing a better example would have been to simulate Copa and compare.}

% We hope that our analysis motivates the networking community to change the status-quo and finally adopt explicit feedback.

%Unlike Copa, Sprout, and Cubic+Codel, the sending rate in explict ABC and XCP is based on the link capacity in the previous round. Copa, Sprout, and Cubic+Codel, indeed perform worse than ABC and XCP. 

   \begin{figure}[tb]
 \centering
     \begin{subfigure}[h]{0.32\textwidth}
        \includegraphics[width=\textwidth]{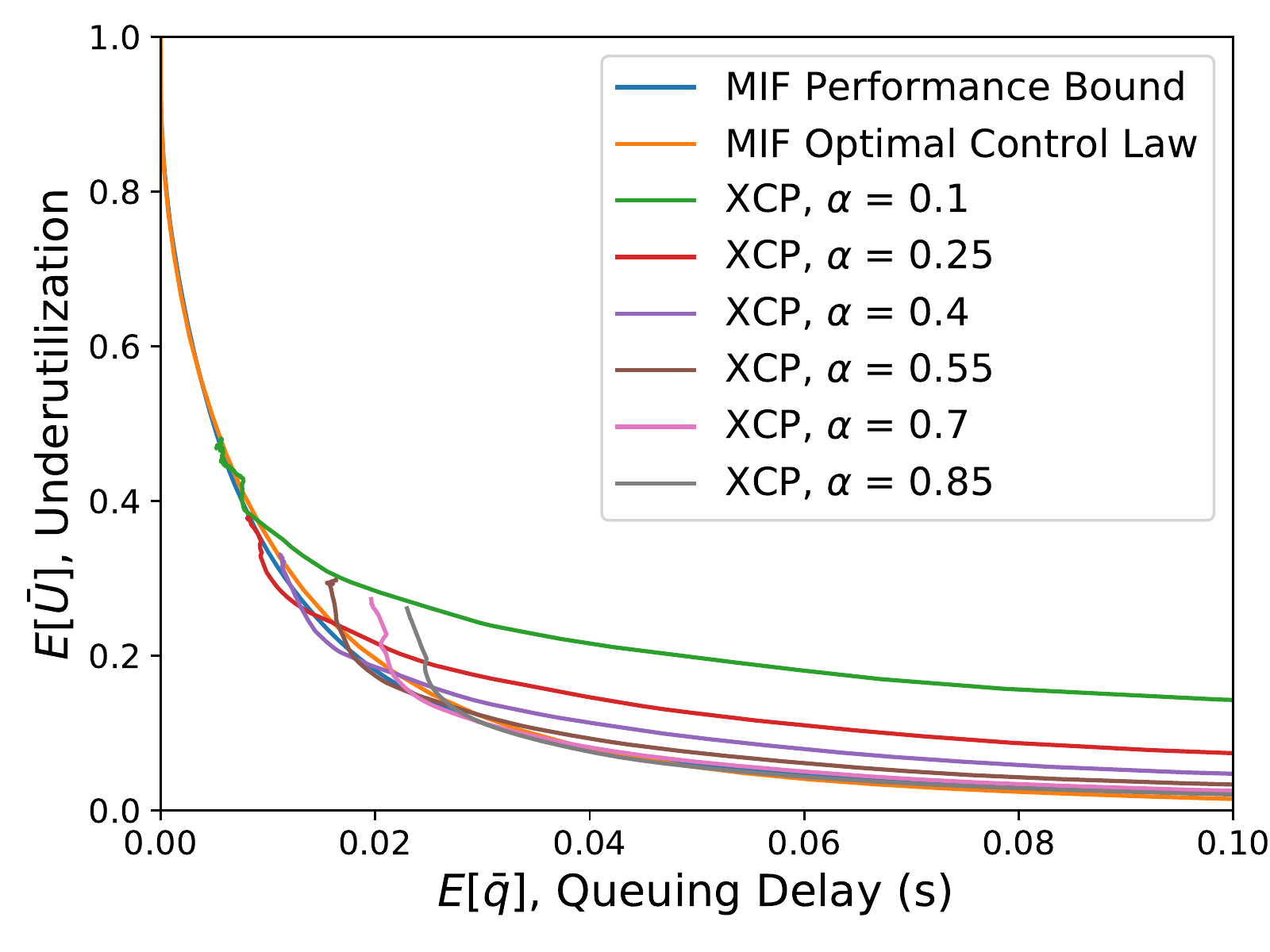}
        \vspace{-6mm}
        \caption{Verizon Downlink}
        \label{fig:model1:xcpcomp:verizon_down}
    \end{subfigure}
     \begin{subfigure}[h]{0.32\textwidth}
        \includegraphics[width=\textwidth]{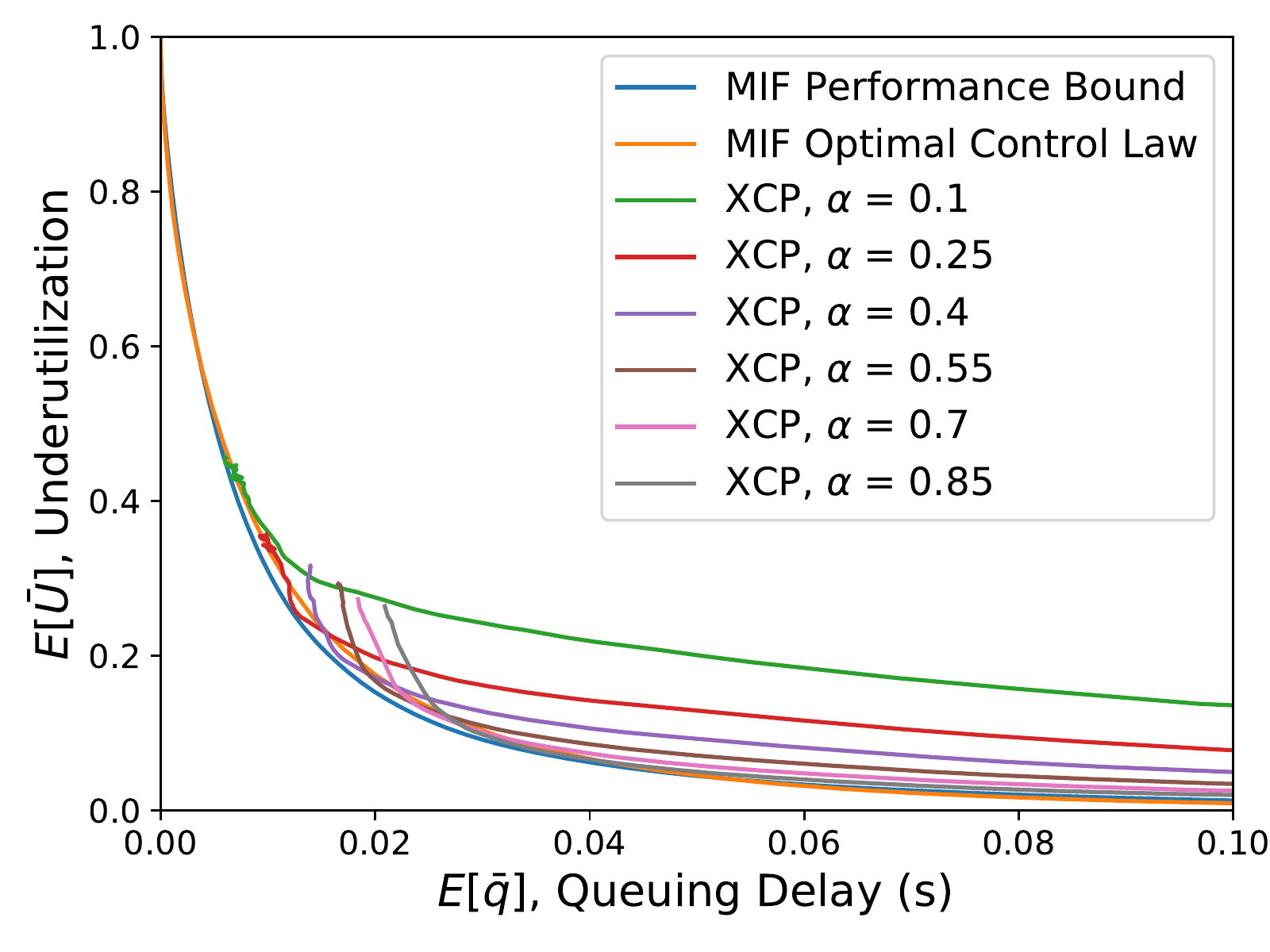}
        \vspace{-6mm}
        \caption{Verizon Uplink}
        \label{fig:model1:xcpcomp:verizon_up}
    \end{subfigure}
    \begin{subfigure}[h]{0.32\textwidth}
        \includegraphics[width=\textwidth]{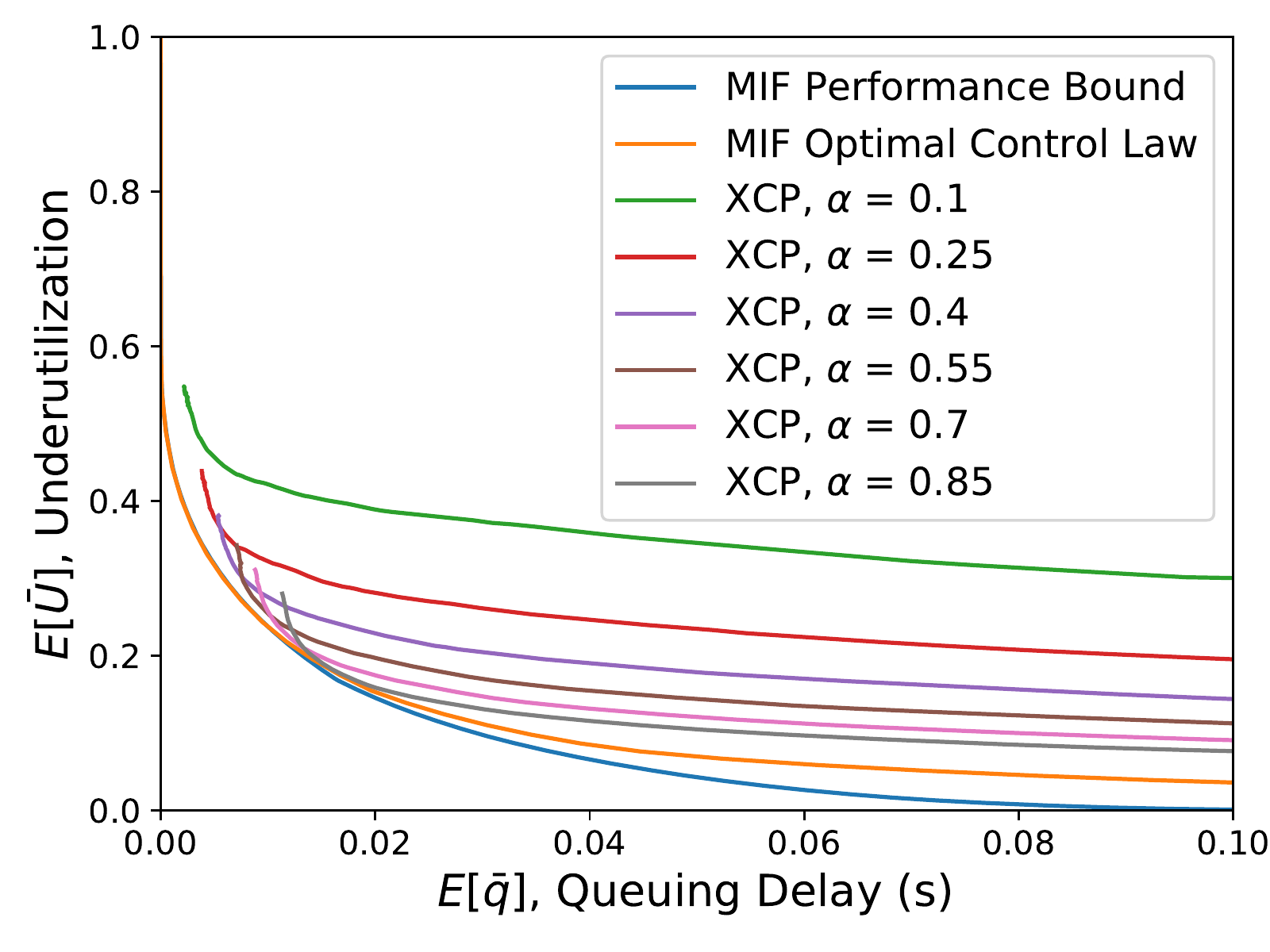}
         \vspace{-6mm}
        \caption{Synthetic Trace}
        \label{fig:model1:xcpcomp:custom}
    \end{subfigure}
    \vspace{-4mm}
    \caption{\small {\bf Comparison with XCP/RCP style control law ---} Each XCP curve shows performance for a fixed value of $\alpha$ and variable value of $\beta$ (0 to 10).}
    \label{fig:model1:xcpcomp}
    \vspace{-4mm}
\end{figure} 

\smallskip
\noindent
\textbf{Set sending rate based only on queue size and link capacity:} Some explicit congestion protocols, such as XCP and RCP (for a single sender), set the sending rate according to the following rule\footnote{RCP uses a variant of the rule in which it adjusts the rate in multiplicative  (rather than  additive) increments.}:
\begin{align}
    s(t) = \left( s(t-1) + \alpha \cdot (\mu(t-1) - s(t-1)) - \beta \cdot \frac{Q(t-1)}{T}\right)^+,
\end{align}
where $\alpha$ and $\beta$ are smoothing parameters (typically $< 1$)  that govern stability, how quickly the protocol adapts to spare capacity ($\mu(t-1) - s(t-1))$,
and how aggressively it drains the queue. Unlike our optimal control law, the sending rate in these schemes is dependent on the {\em sending rate} chosen in the previous time step. This does not attempt to keep the relative link load $\rho(t)$ fixed in every time step. Instead, $\rho(t)$ is dependent on congestion control decisions in the previous time steps, and any error in guessing the 
link rate is thus compounded.

%Our analysis argues that the performance of such control laws might be sub-optimal on links where the relative uncertainty in link capacity is the same in each time step \ma{on links where the link capacity can vary widely over time (isn't this the more basic point?)}\pg{fine with this, felt the original sentence is more accurate}. 
Our analysis argues that the performance of such control laws might be sub-optimal on links where the link capacity can vary widely over time.
To demonstrate this, we simulate an XCP/RCP style control law for various values of $\alpha$ and $\beta$. \Fig{model1:xcpcomp} shows the result for the three traces.  For certain values of $\alpha$ and $\beta$, the performance gets close to the optimal strategy curve. However, for most values the performance deviates from the optimal, and in order to achieve the optimal, we need to tune both $\alpha$ and $\beta$.  For any fixed value of $\alpha$, varying $\beta$ is not optimal across the entire range.

\smallskip
\noindent
\textbf{Do not change the queuing penalty term to adjust the performance trade-off:} Some explicit schemes, such as ABC, follow a different control law of the form:
\begin{align}
    s(t)= \left( \eta \cdot \mu(t-1) - \beta \frac{Q(t-1)}{T}\right)^+,    
\end{align}
where $\eta$ is the target utilization and $\beta$ is a scaling factor for queuing penalty.

This control law is similar to the optimal control law. However, in such control laws the parameter $\eta$ (analogous to $C(w, \gamma)$ in the optimal control law) is incorrectly viewed as the target utilization and is usually set to a value just below one~\cite{abc}. In order to adjust the trade-off between utilization and queuing delay such schemes vary $\beta$. The optimal control law argues that the sender should instead use a fixed value of $\beta = 1$ and vary the parameter $C(w, \gamma)$ (or $\eta$) instead; setting it much less than or greater than one depending on the desired trade-off between queueing delay and underutilization. 

We simulate the ABC style control law for a fixed value of $\eta = 1$ and different values of $\beta$. \Fig{model1:fcvq} shows the result. As expected, the ABC control law deviates from optimal performance, particularly for low values of queueing delay.

   \begin{figure}[tb]
 \centering
    \begin{subfigure}[h]{0.32\textwidth}
        \includegraphics[width=\textwidth]{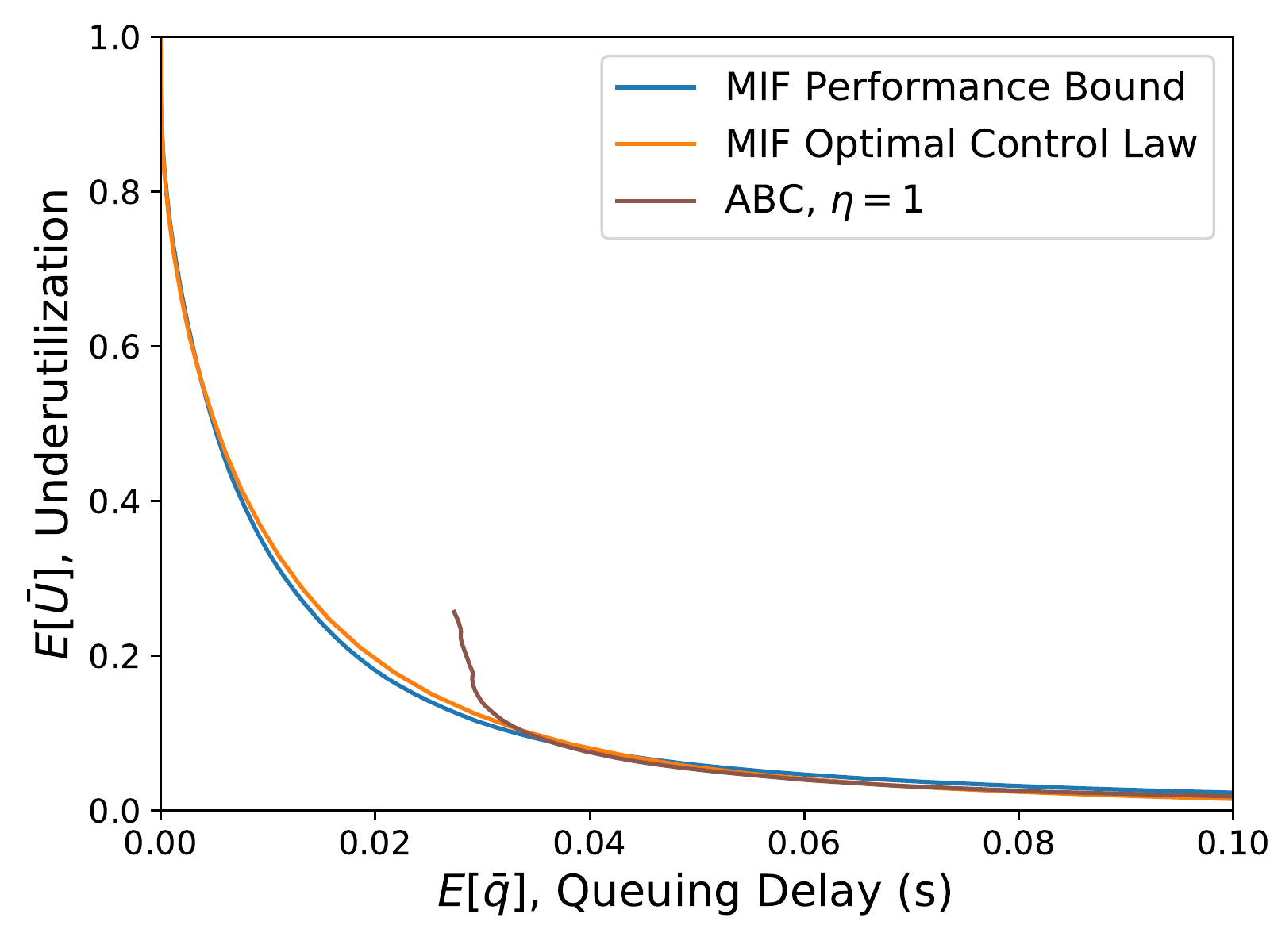}
        \vspace{-6mm}
        \caption{Verizon Downlink}
        \label{fig:fcvq:verizon_down}
    \end{subfigure}
     \begin{subfigure}[h]{0.32\textwidth}
        \includegraphics[width=\textwidth]{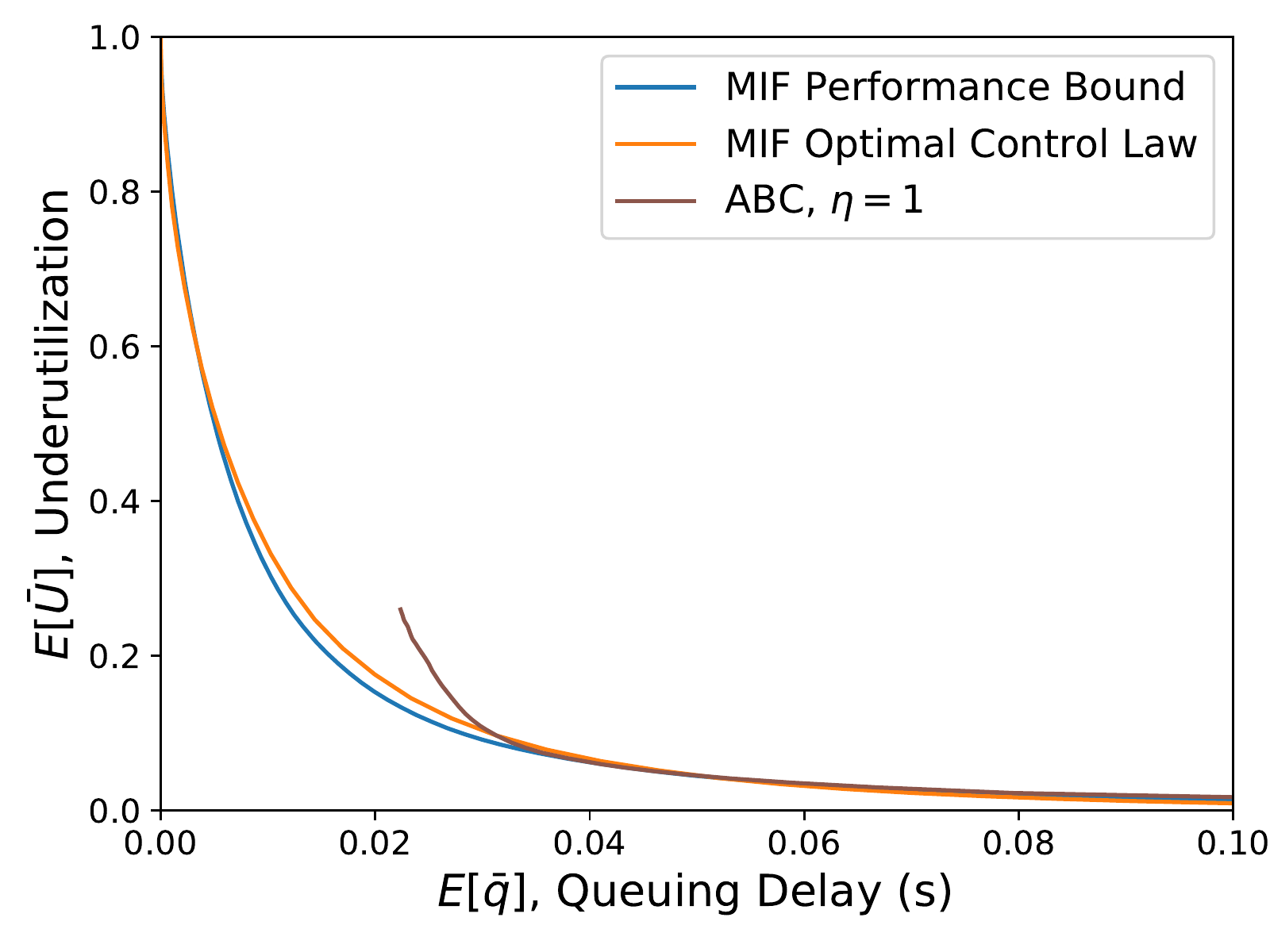}
        \vspace{-6mm}
        \caption{Verizon Uplink}
        \label{fig:fcvq:verizon_up}
    \end{subfigure}
    \begin{subfigure}[h]{0.32\textwidth}
        \includegraphics[width=\textwidth]{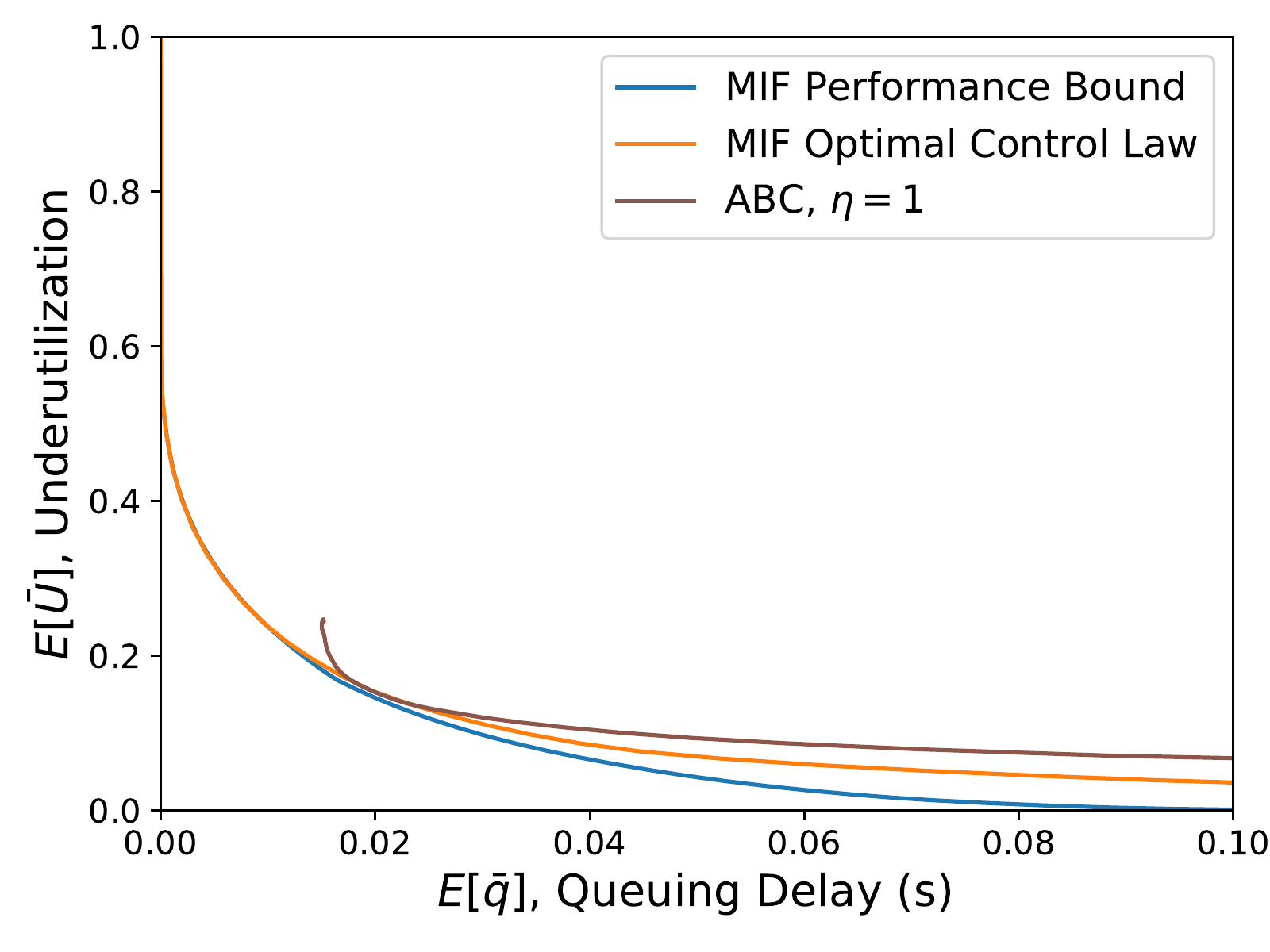}
         \vspace{-6mm}
        \caption{Synthetic Trace}
        \label{fig:fcvq:generate}
    \end{subfigure}
    \vspace{-4mm}
    \caption{\small {\bf Comparison with ABC style control law.}}
    \label{fig:model1:fcvq}
    \vspace{-4mm}
\end{figure}

\smallskip
\noindent
\textbf{Implications for Link Layer Design:} Our analysis also has implications for link layer design (see Appendix~\ref{app:implications_linklayer}). For example, our analysis reveals that variability and uncertainty in link capacity restricts achievable performance. For cellular networks, designing schedulers with a smoother allocation for link capacity can improve performance. Additionally, knowledge about the achievable performance can enable cellular operators to provide guaranteed service level agreements for latency-sensitive applications such as video conferencing.

\subsection{PMIF Model}
\label{ss:model2:implications}

%\textbf{Good prediction can improve performance substantially:} 
To demonstrate that good prediction can improve performance substantially this model, we use the traces from \S\ref{ss:model1:implications}. At the start of each time step $t$, we synthetically supply a prediction ($Pred(t-1)$) to the sender. To demonstrate improvements over the optimal control law (without prediction) from the MIF model, we intentionally pick a prediction with good accuracy ($X^p \sim U (0.8, 1.2)$) such that for each trace $g^{f^p}$ is better than $g^f$. Using $X^p$ and $\mu(t)$, we sample the value for $Pred(t-1)$ with Eq.~\eqref{eq:model2:model} as follows:
\begin{align}
    Pred(t-1) = \frac{\mu(t)}{X^p}.
\end{align}

\Fig{model2:basic} shows the performance bound ($g^{f^p}(\cdot)$) obtained in Corollary~\ref{cor:model2:perfbound}. The figure also shows the performance curve when following the optimal control law from Corollary~\ref{cor:model2:mdp} for different values of $C^p(w, \gamma)$. We also include the performance bound ($g^f(\cdot)$) and the performance curve of the optimal control law based on the MIF model. Indeed, the optimal control law with prediction  outperforms that based on the MIF model.
Additionally, the optimal control law achieves performance close to the performance bound. 

This result shows that that improving link rate prediction could provide large performance gains. This could come from better modelling of the statistical properties of wireless links, or from physical layer information
such as signal strength~\cite{predictwifi, pbe_cc, xie2015pistream}.
%\pg{This result shows that that improving link rate prediction could provide large performance gains. Besides better statistical modeling, it is likely that the largest gains could result from better signals from the link layer and  wireless schedulers~\cite{pbe_cc, xie2015pistream}.} \pg{Based on MA's suggestion. Though the citation is meh}
%\pg{I cut some text here (in comments).}%We can use $g^{f^p}(\cdot)$ to calculate an approximate value for $C^p(w, \gamma)$ in the optimal control law. To achieve a desired trade-off $w$, we can use  $g^{f^p}(\cdot)$ to calculate a $b$ such that $g^{f^p}{'}(x) = -w$ ($x$ is from Eq.~\eqref{eq:model1:thm:perfbound}) and correspondingly set $C^p(w, \gamma)$ to $b$(=$C^p(w, 0)$) in the optimal control law.

   \begin{figure}[tb]
 \centering
     \begin{subfigure}[h]{0.32\textwidth}
        \includegraphics[width=\textwidth]{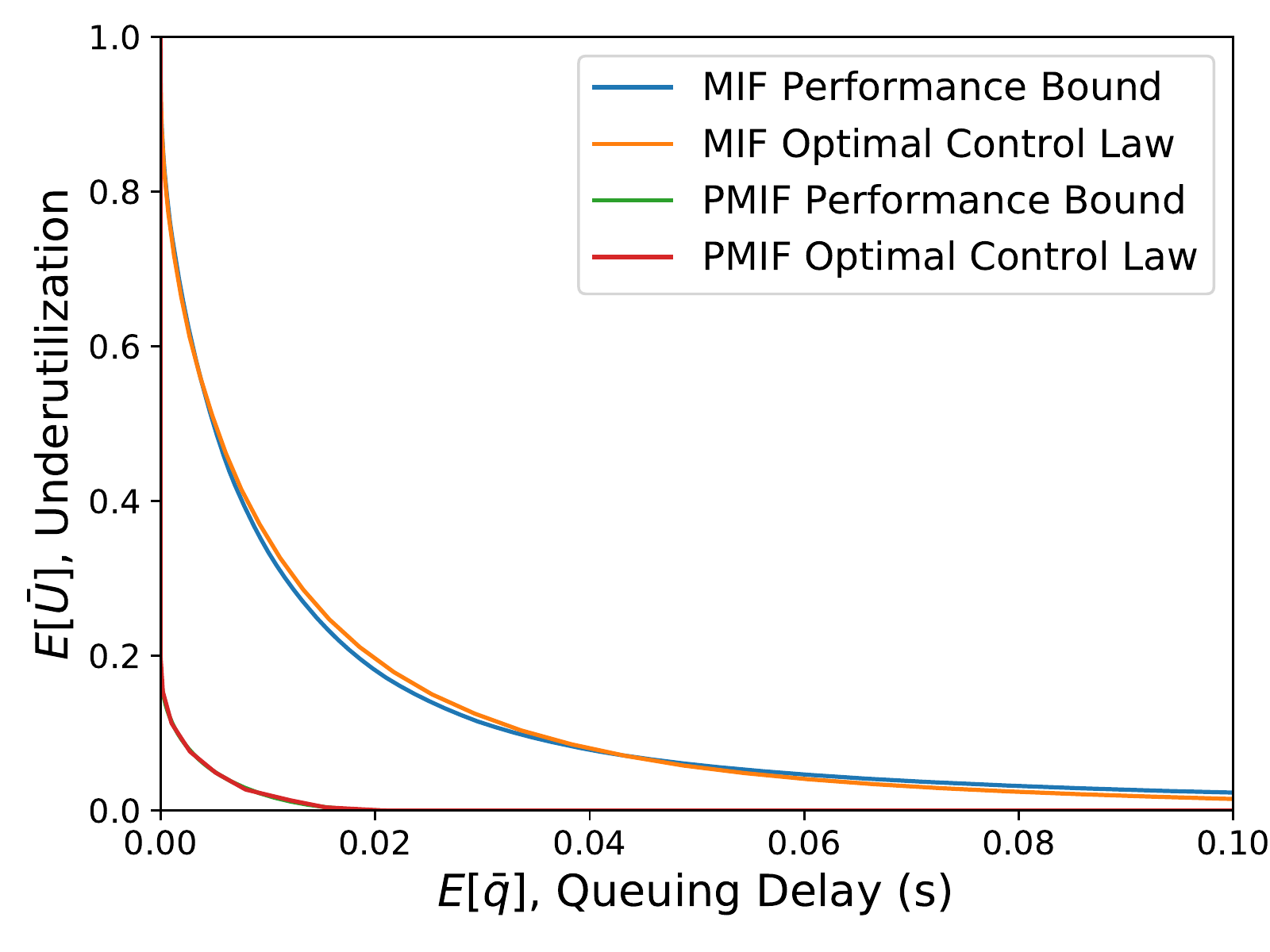}
        \vspace{-6mm}
        \caption{Verizon Downlink}
        \label{fig:model2:basic:verizon_down}
    \end{subfigure}
     \begin{subfigure}[h]{0.32\textwidth}
        \includegraphics[width=\textwidth]{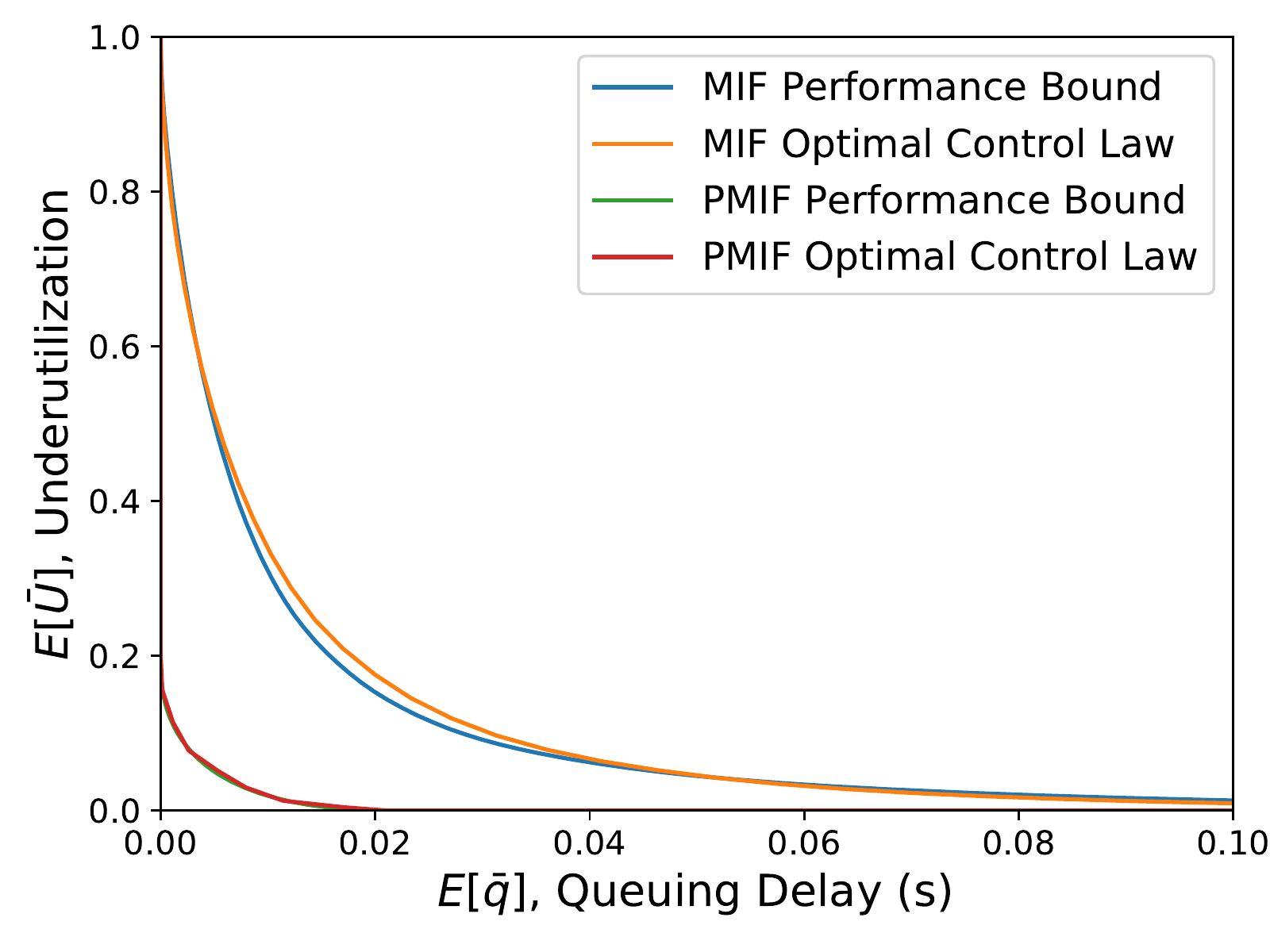}
        \vspace{-6mm}
        \caption{Verizon Uplink}
        \label{fig:model2:basic:verizon_up}
    \end{subfigure}
    \begin{subfigure}[h]{0.32\textwidth}
        \includegraphics[width=\textwidth]{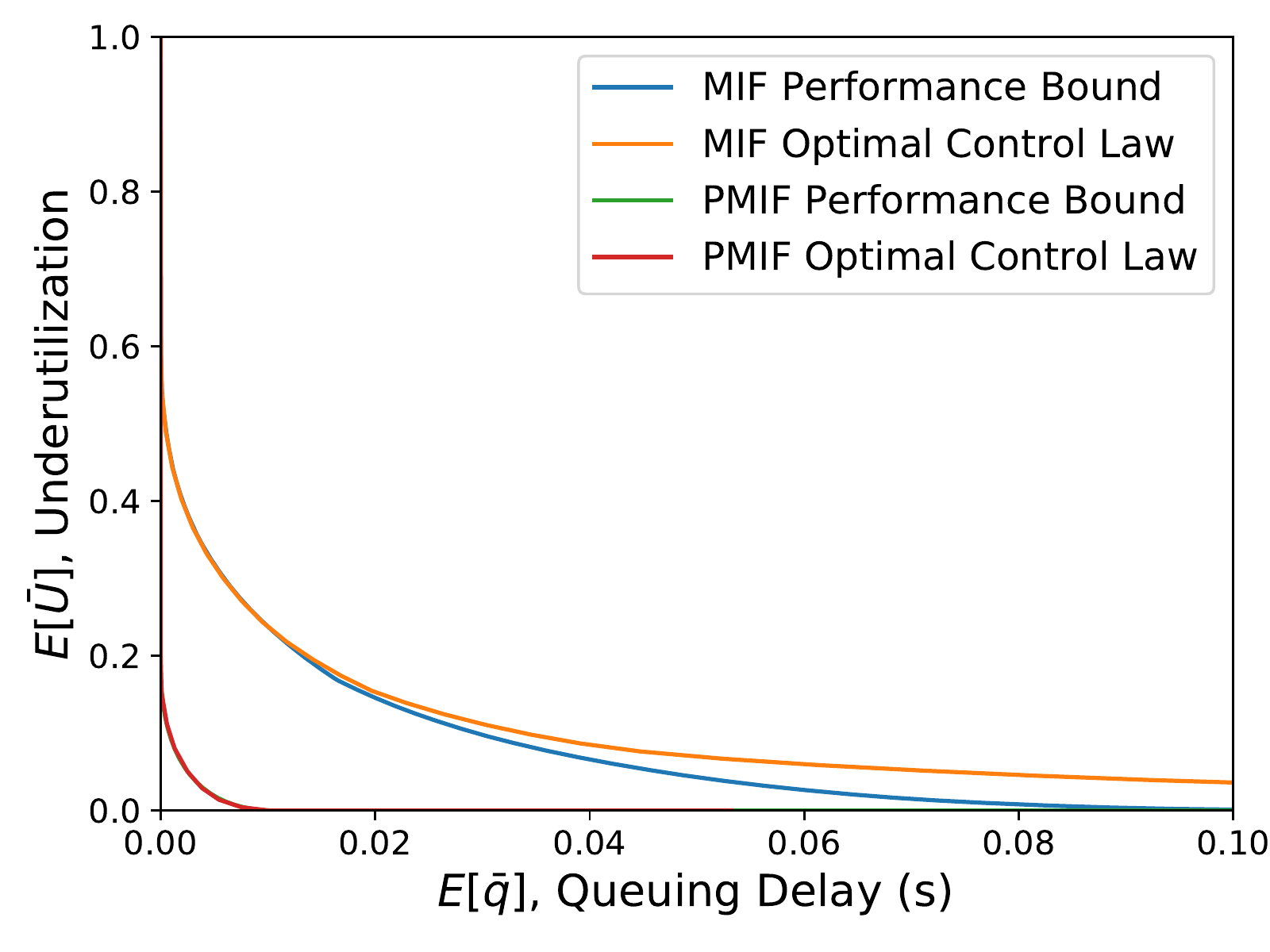}
         \vspace{-6mm}
        \caption{Synthetic Trace(MIF model)}
        \label{fig:model2:basic:custom}
    \end{subfigure}
    \vspace{-4mm}
    \caption{\small {\bf Impact of incorporating prediction ---} Incorporating prediction can substantially improve performance.}
    \label{fig:model2:basic}
    \vspace{-4mm}
\end{figure} 

\section{A General Markovian Link Model}
\label{s:model3}
\label{ss:generic_markov}
\label{ss:model3:model}
In the MIF model, tne link capacity varies independently from step to step and the relative variation is the same for every time step. In a more realistic setting, the variability may instead depend on the \emph{state of the link}. 
%\ma{For example, cellular links have a maximum rate, and the relative variability is likely to be higher at lower link rates compared to higher link rates} \ma{Is this 2nd part of this sentence ok?}.
\pg{For example, cellular links have a maximum rate, and the likelihood of relative increase in capacity is likely to be higher at lower link rates compared to higher link rates.}  In this section, we extend the MIF model to allow for a more expressive Markovian process governing the link capacity. We refer to this extension as the {\em State-dependent Multiplicative Factors (SMF)} model. 
%old
%In the MIF model, we treated the link capacity as a markov chain with restrictions on the transition probabilities. This restricted model might not accurately capture the nature of variability on links where the relative variation in link capacity (over the next time step) is not the same for every time step but is rather dependent on the \emph{state of the link} in the previous time step. To overcome this potential limitation, we will now extend the MIF model to allow for a more expressive markovian process governing the link capacity. We refer to this extension as the State-dependent Multiplicative Factors (SMF) model. %Our analysis of the MIF model forms the basis for studying more complex link models later in the paper. 

%Let $\textbf{S}$ be the set of all possible link states in the markov process. The link state $S(t)$ in itself is a set that includes all the quantities that impact how the link capacity varies in the next round. For example, in cellular networks where multiple users are competing at a single base station, the number of users at the base station might be a part of the state $S(t)$. Link state in itself will change over time. We model the link state as generic markov chain. The state in round $t-1$ ($S(t-1)$) governs the state in time step $t$ ($S(t)$). Formally, the uncertainty in link capacity is governed by,

%$S(t-1)$ also implicitly governs the link capacity $\mu(t)$.  
%In other words, this model takes into account all the quantities that can affect the link capacity under a markovian process. 

% We now introduce a new quantity. 
Let $S(t) \in \textbf{S}$ represent the state of the link in time step $t$, where $\textbf{S}$ denotes the set of all possible link states. \pg{The link state includes the link capacity $\mu(t)$ and any other quantity that can impact how the link capacity varies in the next time step.}
%\pg{for some reason MA wrote that time step}
%Let $\textbf{S}$ be the set of all possible link states in the markov process. 
%$S(t)$ in itself is a set, and includes the link capacity in that time step ($\mu(t)$). 
%$S(t)$ also includes all the quantities that impact how the link capacity varies in the next time step ($\mu(t+1)$). 
For example, in cellular networks where multiple users compete at a base station, the number of users at the base station might be part of the link state. 
% Link state in itself will change over time. 
To capture variations in the link state, we model the link state itself as a Markov chain. The link state in time step $t-1$ ($S(t-1)$) governs the probability distribution of the link state in time step $t$ ($S(t)$). $S(t-1)$ implicitly also governs the probability distribution of link capacity in the next time step ($\mu(t)$).  %In simple words, this model takes into account all the quantities that affect the link capacity. 

%With this model, we will establish a performance bound for causal congestion control protocols and draw insights about the optimal congestion control loop. Note that, the number of distinct states int the set $\textbf{S}$ can be infinite (i.e. continuous state space), for our proof we will assume the number of such states is finite (i.e. discrete state space). However, our argument holds for continuous state space as well \pg{Caution}. Let $\mathbb{A}$ be the transition matrix associated with the markov process. This matrix defines the probability $P^{A}(S(t)=k_1|S(t-1)=k_2)$. To calculate this trade-off, we will first convert this model to a from analogous to Equation ~\ref{eq:model1:model}. This markov model can be described in an alternate form given by the equations,

The number of distinct states in the set $\textbf{S}$ can be infinite. For ease of presentation, we will assume that the number of such states is finite, but our analysis holds for continuous state spaces as well.  Note that the above model includes as a special case any Markovian link rate process governed by a probability transition function $P(\mu(t+1)|\mu(t))$; the link state is simply $S(t) = \mu(t)$ in this case. 

Let $\mathbb{A}$ be the probability transition matrix for the link state. This matrix defines the probability $P^{A}(S(t)=k_1|S(t-1)=k_2)$. For ease of analysis, we first convert this model to a form analogous to Eq. ~\eqref{eq:model1:model}. The Markov chain for the link state can be described in an alternate form as follows. Let $S(t)$ conditioned on the event $\{S(t-1) = k\}$ have distribution $H^k$. We can write: $S(t) = Z_t^{k}$ if $S(t-1)=k$, where  $Z_{t}^{k}$ is an I.I.D random variable with distribution $H^k$ ($\forall k \in \textbf{S}, \forall t \in \mathbb{N}$). In short:
\begin{align}
    S(t) = Z_{t}^{S(t-1)}.
\end{align}
Since $\mu(t)$ is determined by $S(t)$, we can similarly write $\mu(t)$ as follows,
\begin{align}
    \mu(t) &= Y_{t}^{S(t-1)},
\end{align}
where $\forall k \in \textbf{S}, \forall t \in \mathbb{N}$, $Y_{t}^{k}$ is an I.I.D random variable. Similarly, since $\mu(t-1)$ is given by $S(t-1)$, we can rewrite the above equation as follows,
\begin{align}
    \mu(t) &= \mu(t-1) \cdot X_{t}^{S(t-1)},
    \label{eq:generic_markov_mu_def}
\end{align}
where $\forall k \in \textbf{S}, X_{t}^{k}$ is an I.I.D random variable. 

Let $f^k(\cdot)$ be the PDF associated with random variable $X^{k}$. $f^k$ can be calculated from the transition matrix $\mathbb{A}$. The equation above (governing the uncertainty in the link capacity at the sender) is now analogous to Eq.~\eqref{eq:model1:model}. We make the following assumptions:

\textit{Assumptions:} (1) Link state transitions are independent of the decisions of the congestion control protocol, i.e., the sending rate and queuing delay do not affect how the link state changes. 
(2) Link state has a stationary distribution in the steady state given by the function $\lambda(\cdot)$, i.e. $P(S(t)=k) \to \lambda(k)$, and $t \to \infty$.
% This implies that we cannot pick a arbitrary transition matrix, instead the transition matrix $\mathbb{A}$ must meet the criterion \pg{(cite XXX)} for stationary distribution in steady state; 
(3) The sender has access to the link state information in the previous time step ($S(t-1)$) for deciding the sending rate in time step $t$.
(4) The starting state ($S(0)$) has the same distribution as the stationary distribution $\lambda$. This implies, $P(S(t)=k) = \lambda(k), \forall t \in \mathbb{N}$.\footnote{\pg{This assumption is merely for convenience of analysis. In the more general case of starting from any arbitrary initial state, the performance bound we provide will hold after a sufficient ``mixing'' time has passed so that the link state distribution is close the stationary distribution. The optimal control law we derive holds regardless of the initial state.}} %\ma{This assumption is merely for convenience. Our result hold regardless of the initial but the analysis becomes more tedious.} 
5) The link capacity is always positive, i.e., we only consider $\mathbb{A}$ such that $P(\mu(t)<=0) = 0, \forall t \in \mathbb{N}$.

Similar to our analysis for the MIF model, we derive a performance bound between $E[\bar{q}(t)]$ and $E[\bar{U}(t)]$ (\S\ref{ss:model3:tperfbound}), use it to draw insights about how the optimal congestion controller (\S\ref{ss:model3:achievable_bound}), and derive the optimal control law by posing the task as a MDP (\S\ref{ss:model3:mdp}). Finally, we discuss the implications of our analysis and validate our findings (\S\ref{ss:model3:implications}). %Note that, we do not attempt to come up with a predictor ourselves, rather we only analyse the impact of error in prediction on time-varying links.

\subsection{Performance Bound}
\label{ss:model3:tperfbound}

\begin{theorem}
In the SMF model, for any causal congestion control protocol, for any $t \in N$,% there is a fundamental trade-off between $E[\bar{q}(t)]$ and $E[\bar{U}(t)]$. 
the point $(E[\bar{q}(t)], E[\bar{U}(t)])$ lies in a convex set $\mathbb{D}^A$. The point $(x,y) \in \mathbb{D}^A$ iff
\begin{align}
    x &= \sum_{k \in \textbf{S}} \lambda(k) \cdot x^k, & y = \sum_{k \in \textbf{S}} \lambda(k) \cdot y^k,
    \label{eq:thm:model3:perfbound}
\end{align}
where $\forall k \in \textbf{S}$, the point $(x^k, y^k) \in \mathbb{C}^{f^k}$. The set $\mathbb{C}^{f^k}$ is as defined in Theorem ~\ref{thm:model1:perfbound} with $f^k(\cdot)$ as the PDF %\ma{reference the relevant equation}\pg{we dont have a equation for $\mathbb{C}^{f^k}$}. %$f^k$ is the PDF of the random variable $X^k$, and is calculated using $\mathbb{A}$.
\label{thm:model3:perfbound}
\end{theorem}

The performance bound is dependent on the probability distribution of the link state ($\lambda$) and how each link state $S(t-1)$ impacts the relative variability in link capacity in the next time step ($f^{S(t-1)}$). The set $\mathbb{D}^A$ is a weighted average of the sets $\cup_{k \in \textbf{S}} \{\mathbb{C}^{f^k}\}$. $\mathbb{C}^{f^k}$ in itself is the performance bound when the relative variability in link capacity is the same in each time step according to link state $k$ (i.e., $\mu(t) = \mu(t-1) \cdot X^k_t, \forall t \in \mathbb{N}$).

%\subsubsection{Proof\\}
\begin{proof}
 Since link state ($S(t-1)$) governs the link capacity in the next time step ($\mu(t)$), the queuing delay and underutilization in time step $t$ also depend on $S(t-1)$. We can calculate $E[q(t)]$ and $E[U(t)]$ using the conditional expected values $E[q(t)|S(t-1)=k]$ and $E[U(t)|S(t-1)=k]$ as follows
\begin{align}
    E[q(t)] &= E_{S(t-1)}[E[q(t)|S(t-1)=k]] = \sum_{k \in \textbf{S}} \lambda(k) \cdot E[q(t)|S(t-1)=k],\nonumber\\
    E[U(t)] &= E_{S(t-1)}[E[U(t)|S(t-1)=k]] = \sum_{k \in \textbf{S}} \lambda(k) \cdot E[U(t)|S(t-1)=k].
    \label{eq:model3:generic_markov_equ}
\end{align}

To prove our theorem, we will prove the following lemma about the conditional expected values of queuing delay and underutilization,
\begin{lemma}
 $\forall k \in \textbf{S}$, the point $(E[q(t)|S(t-1)=k],E[U(t)|S(t-1)=k]) \in \mathbb{C}^{f^k}$.
\end{lemma}
\begin{proof}
The conditional queuing delay $(q(t)|(S(t-1)=k))$ and underutilization $(U(t)|S(t-1) = k)$ can be calculated as follows,
\begin{align}
    q(t)|(S(t-1)=k) &= T \cdot \left(\frac{\frac{Q(t-1)}{T} + s(t)}{\mu(t - 1) \cdot X_{t}^{k}} - 1 \right)^+ = T \cdot \left( \frac{\rho(t)}{X_{t}^k} -1 \right)^+,\nonumber\\    
    U(t)|(S(t-1)=k) &= \left( 1 - \frac{\frac{Q(t-1)}{T} + s(t)}{\mu(t - 1) \cdot X_{t}^{k}}\right)^+ = \left(1 -  \frac{\rho(t)}{X_{t}^k} \right)^+, 
    \label{eq:model3:proof_perfbound:generic_markov_system}
\end{align}
where $\rho(t)$ is as defined by Eq.~\eqref{eq:model1:proof_perfbound:rt}.  

Given the link state in the previous time step ($S(t-1)$), the relative link load ($\rho(t)$) solely governs the probability distribution of $q(t)$ and $U(t)$. The equation above is in fact analogous to Eq.~\eqref{eq:model1:proof_perfbound:qU_rt} from the MIF model. Thus, we can apply the same techniques as in the proof of Theorem ~\ref{thm:model1:perfbound}. The point $(E[q(t)|S(t-1)=k, \rho(t) = b], E[U(t)|S(t-1)=k, \rho(t) = b])$ lies on the curve $g^{f^k}(\cdot)$. Consequently the point $(E[q(t)|S(t-1)=k], E[U(t)|S(t-1)=k])$ lies
in the convex set $\mathbb{C}^{f^k}$.%, where $f^k$ is the PDF of the random variable $X^{k}$. 
\end{proof}

The Lemma along with Eq.~\eqref{eq:model3:generic_markov_equ} establishes that the point $(E[q(t)], E[U(t)]) \in \mathbb{D}^A$. The point $(E[\bar{q}(t)], E[\bar{U}(t)])$ also belongs to $\mathbb{D}^A$ (using the same technique as the multi-time step proof of Theorem ~\ref{thm:model1:perfbound}). Since $\mathbb{C}^{f^k}$ is convex $\forall k \in \textbf{S}$, the weighted average of $\cup_{k \in \textbf{S}}\{\mathbb{C}^{f^k}\}$ is convex, and $\mathbb{D}^A$ is convex. %Corollary proved. %\pg{Is this enough?}
\end{proof}

Next, we will show a simple method for computing the set $\mathbb{D}^A$. This method will provide us insights into the potential form of the optimal control law. 

\subsubsection{How to compute the set $\mathbb{D}^A$?\\}

To compute $\mathbb{D}^A$, we will calculate it's boundary function $g^A(\cdot)$. Point $(x,y)$ belongs to $\mathbb{D}^A$ iff $y \geq g^A(x)$. The boundary point $(x,g^A(x))$ must be a weighted average of points on the curves $\cup_{k \in \textbf{S}}\{g^{f^k}(\cdot)\}$. Formally,

\begin{align}
    x &= \sum_{k \in \textbf{S}} \lambda(k) \cdot x^k, & y = \sum_{k \in \textbf{S}} \lambda(k) \cdot g^{f^k}(x^k).
\end{align}
This is because for a given $x^k$, the min value of $y^k$, such that $(x^k,y^k) \in \mathbb{C}^{f^k}$ occurs at $y^k=g^{f^k}(x^k)$. 

\begin{proposition}
The point $(x,y) \in \mathbb{D}^A$, given by 
\begin{align}
    x &= \sum_{k \in \textbf{S}} \lambda(k) \cdot x^k, & y = \sum_{k \in \textbf{S}} \lambda(k) \cdot g^{f^k}(x^k),\nonumber\\
\end{align}
is on the boundary ($y=g^A(x)$) iff,
\begin{align}
    g^{f^i}{'}(x^i) = g^{f^j}{'}(x^j) &, \forall i, j \in \textbf{S}.  
\end{align}
\label{prop:model3:compute_da}
\end{proposition}
\begin{proof}
Please see Appendix~\ref{app:prop:compute_da} for proof details. 
\end{proof}

This proposition states that the point $(x,g^{A}(x))$ is a weighted average of points $\cup_{k \in \textbf{S}} \{(x^k, g^{f^k}(x^k))\}$ with the {\em same slope}. In other words, we are averaging points such that the marginal trade-off between expected queuing delay and underutilization is same on the curves $\cup_{k \in \textbf{S}}\{g^{f^k}(\cdot)\}$.

\subsection{Is the Performance Bound Achievable?}
\label{ss:model3:achievable_bound}

Recall, the point $(E[q(t)|S(t-1)=k, \rho(t) = b], E[U(t)|S(t-1)=k, \rho(t) = b])$ lies on the curve $g^{f^k}(\cdot)$. The proposition in the previous subsection establishes that we can achieve a particular point ($x, g^{A}(x)$) on the performance bound where $x = \sum_{k \in \textbf{S}}(\lambda(k) \cdot x^k)$ and %$g^{f^i}{'}(x^i) = g^{f^j}{'}(x^j), \forall i,j \in \textbf{S}$ by following a strategy of the form,
$g^{f^k}{'}(x^k) = g^{A}{'}(x), \forall k \in \textbf{S}$ by following a strategy of the form,

\if 0
\begin{align}
    \rho(t)|(S(t-1)=k) &= C^{simp}(g^{A}{'}(x), k) &, \forall t \in \mathbb{N}, \forall k \in \textbf{S}\nonumber\\
    s(t)|(S(t-1)=k) &= C^{simp}(g^{A}{'}(x), k) \cdot \mu(t-1) - \frac{Q(t-1)}{T} &, \forall t \in \mathbb{N}, \forall k \in \textbf{S} \nonumber\\
    x^k &= E[q(t)|S(t-1)=k, \rho(t) = C^{simp}(g^{A}{'}(x), k)]) &, \forall k \in \textbf{S}%\nonumber\\
    %g^{f^i}{'}(x^i) &= g^{f^j}{'}(x^j) &, \forall i, j \in \textbf{S} 
    \label{eq:model3:simplified_strategy}
\end{align}
\fi

\if 0
\begin{align}
    \rho(t)|(S(t-1)=k) &= C^A(k), & \forall t \in \mathbb{N}, \forall k \in \textbf{S},\nonumber\\
    s(t)|(S(t-1)=k) &= C^A(k)\cdot \mu(t-1) - \frac{Q(t-1)}{T}, & \forall t \in \mathbb{N}, \forall k \in \textbf{S},\nonumber\\
    x^k &= E[q(t)|S(t-1)=k, \rho(t) = C^A(k)]), & \forall t \in \mathbb{N}, \forall k \in \textbf{S}.%\nonumber\\
    %g^{f^i}{'}(x^i) &= g^{f^j}{'}(x^j) &, \forall i, j \in \textbf{S} 
    \label{eq:model3:simplified_strategy}
\end{align}
\fi
\begin{align}
    \rho(t)|(S(t-1)=k) &= C^A(k) \Longrightarrow  s(t)|(S(t-1)=k) = C^A(k)\cdot \mu(t-1) - \frac{Q(t-1)}{T}, &&  \forall k \in \textbf{S},
    %x^k &= E[q(t)|S(t-1)=k, \rho(t) = C^A(k)]), & \forall t \in \mathbb{N}, \forall k \in \textbf{S}.%\nonumber\\
    %g^{f^i}{'}(x^i) &= g^{f^j}{'}(x^j) &, \forall i, j \in \textbf{S} 
    \label{eq:model3:simplified_strategy}
\end{align}
where $\forall k \in \textbf{S}, C^A(k)$ satisfies $x^k = E[q(t)|S(t-1)=k, \rho(t) = C^A(k)])$.
%Note that, this also implies that $g^{f^i}{'}(E[q(t)|S(t-1)=i, \rho(t) = C^{simp}(i)]) = g^{f^j}{'}(E[q(t)|S(t-1)=j, \rho(t) = C^{simp}(j)]), \forall i, j \in \textbf{S}$. 

%The analysis in the previous section shows us that we can achieve a particular performance point on the lower bound by following a strategy of the form $\rho(t) = C(S(t-1))$ such that the sending rate is non-negative and $\forall k \in \textbf{S}$ s.t. $g^{f^k}'(E[q(t)|S(t-1)=k, \rho(t) = C(S(t-1))$ is same $\forall k \in \textbf{S}$. Alternatively, this strategy translates to $s(t) = C(S(t-1) \cdot \mu(t-1) - \frac{Q(t-1)}{T}$.

In simple terms, the above control law is arguing for setting the relative link load based solely on the link state in the previous time step. Intuitively this is as expected because the probability distribution of $q(t)$ and $U(t)$ is only dependent on $\rho(t)$ and $S(t-1)$ (Eq.~\eqref{eq:model3:proof_perfbound:generic_markov_system}).
Again, following this strategy might not always be possible. In particular, in the event of excessive queuing ($Q(t-1) > T \cdot C^A(S(t-1)) \cdot \mu(t-1)$), for $\rho(t) = C^A(S(t-1)$, the sender needs to pick a negative value for $s(t)$ (from Eq.~\eqref{eq:model1:proof_perfbound:rt}) which is not possible. %In such a scenario, the performance point ($E(\bar{q}(t)$, $E(\bar{U}(t)$) might not be on the boundary of set $\mathbb{D}^A$. 

We will now establish the necessary conditions under which the sender can follow the strategy $\rho(t) = C^A(S(t-1))$ exactly. %Let $X^k_{min}$ be the minimum value of $X$ given the PDF $f^k$ ($P(X^k<X^k_{min}) = 0$). 

\begin{proposition}

In the SMF model , it is possible to follow the strategy $\rho(t) = C^A(S(t-1))$  $\forall t \in \mathbb{N}$ iff
\if 0
\begin{align}
    \frac{C^{S(t-1)}}{X_{min}^{S(t-1)}} - 1 &\leq C^{S(t)} & \forall S(t-1), S(t) \in \textbf{S}
\end{align}
\fi
\begin{align}
    \frac{C^A(k_1)}{X_{min}^{k_1}} - 1 &\leq C^A(k_2), & \forall k_1, k_2 \in \textbf{S},
\end{align}
\label{prop:model3:achievable_bound}
where $\forall k \in \textbf{S}, X^k_{min}$ is the minimum value of the random variable $X^k$ ($P(X^k<X^k_{min}) = 0$). 
\end{proposition}
\begin{proof}
Please see Appendix~\ref{app:model3:achievable bound} for the proof. 
\end{proof}
This condition restricts excessive queue build up for all possible transitions in the link state $S(t-1) = k_1$ to $S(t) = k_2$. Again, the performance bound need not be tight.

\if 0
\begin{align}
    Q(t) &\leq T \cdot C_r^{S(t)} \cdot \mu(t) & \forall S(t) \in \textbf{S}\nonumber\\
    (C_r^{S(t-1)} \cdot \mu(t-1) - \mu(t))^+ &\leq C_r^{S(t)} \cdot \mu(t) & \forall S(t-1), S(t) \in \textbf{S} \nonumber\\
    (\frac{C_r^{S(t-1)} }{X^{S(t-1)}} - 1)^+ &\leq C_r^{I(t)} & \forall S(t-1), S(t) \in \textbf{S} \nonumber\\
    \frac{C_r^{S(t-1)} }{X_{min}^{S(t-1)}} - 1) &\leq C_r^{S(t)} & \forall S(t-1), S(t) \in \textbf{S}
\end{align}
\fi

\subsection{Optimal Control Law}
\label{ss:model3:mdp}
%\pg{do we need this? we should include this if we can do any experiments that use this math.}

What happens when the above condition is not met? What is the optimal control law? 
Will the following variant of the strategy in Proposition~\ref{prop:model3:achievable_bound} -- $s(t)|(S(t-1)=k) = \left(C^A(k) \cdot \mu(t-1) - \frac{Q(t-1)}{T}\right)^+$, where $C^A(k)$ is as defined in Eq. ~\eqref{eq:model3:simplified_strategy} -- be optimal? 

To answer this question, as before, we formulate an MDP. The MDP is similar to that in the MIF model but the state at time step $t$ is given by $(q(t-1), \mu(t-1), S(t-1))$. %The congestion control ``agent'' observes this state and selects an {\em action}\,---\,the sending rate $s_t \geq 0$. The state then transitions to the next state $(q(t), \mu(t))$, and the agent incurs a {\em cost} $(w\cdot q(t) + U(t))$ for this time step. To see that this is an MDP, recall that 
The state transitions occur according to:
\begin{align}
    \mu(t)&=\mu(t-1)\cdot X^{S(t-1)}_t  \quad \text{(per the SMF model)}, \nonumber\\
    S(t) &= Z^{S(t-1)}_t \quad \text{(per the SMF model)}, \nonumber\\
    q(t)&=T\cdot(\rho(t)/X^{S(t-1)}_t-1)^+  \quad \text{(per Eq.~\eqref{eq:model3:proof_perfbound:generic_markov_system})}. \nonumber
\end{align}

The goal of the agent is to minimize the objective function ($J(\cdot)$ from Eq.~\eqref{eq:model1:mdp:objective}) over all policies ($\pi$):
\if 0
The goal of the agent is to minimize the following objective function over all policies ($\pi^A$).
\begin{align}
     J^{A}(\pi^A) = E\left[\sum_{t=1}^{\infty}\gamma^{t-1} \Big( w \cdot q(t) + U(t)] \Big)\right], 
\end{align}
where $\gamma \in (0,1)$ is the discount factor. 
\fi
\begin{corollary}
    In the MDP defined above, the optimal control policy that minimizes $J(\pi)$ takes the form:  
    \begin{align}
        s(t) = \left( C^{A}(w, \gamma, S(t-1)) \cdot \mu(t-1) - \frac{Q(t-1)}{T} \right)^+,  
        \label{eq:cor:model3:mdp}
    \end{align}
where $\forall k \in \textbf{S}, C^A(w,\gamma,k)$ is a constant. 
    \label{cor:model3:mdp}
\end{corollary}
\begin{proof}
Please see Appendix~\ref{app:model3:cor_mdp}
\end{proof}

\subsection{Implications and Validation}
\label{ss:model3:implications}
We use the traces from \S\ref{ss:model1:implications} to demonstrate the performance of the optimal control law based on the SMF model. Since we only have access to link capacity information in the traces,
% For cellular traces, we only have access to the link capacity information. So to apply the SMF Model to these traces, 
we treat link capacity as the link state ($S(t) = \mu(t)$). In other words, the relative variation in link capacity in a time step in itself is dependent on the link capacity in the previous time step. Similar to \S\ref{ss:model1:implications} we use the link capacity from traces to generate the PDFs $\cup_{k \in \textbf{S}}\{f^{k}(\cdot)\}$. 
\Fig{model3} shows the performance bound ($g^A(\cdot)$) and the performance curve for the optimal control law\footnote{\pg{
For different values of $w$, we approximate the value of $C^A(w, \gamma, k)$ using the curve $g^A(\cdot)$ (Eq.~\eqref{eq:model3:simplified_strategy}). In particular, in the MDP formulation, $\gamma \to 0$ (or $J(\pi) = E[w \cdot q(1) +U(1)]$) corresponds to minimizing the cost in a single time step. Therefore, assuming $\gamma \to 0$, we can use the curve $g^A$ to compute $C^A(w, \gamma, k)$, formally, $g^{f^k}{'}([E[q(t)|S(t-1)=k, \rho(t)=C^A(w, \gamma, k)]) = -w, \forall k \in \textbf{S}$. 
%Mohammad or Tom, can you please check this? It's correct but not sure if it's clear.
}}.  As expected, on cellular traces, both the performance bound and the performance curve for the optimal control law based on the SMF model outperform those based on the MIF model. Compared to the MIF model, the SMF model provides more degrees of freedom to model the variations in link capacity. 
For example, unlike the MIF model, when the link capacity is low it is more likely to increase than if
the link capacity is high.
% there is more probability for the link capacity to increase in the next time step then when the current link capacity is high. 
With access to the physical layer information, we might be able to use the SMF model to better capture the nuances of variations in the link capacity and improve performance even further.
%Therefore, for $\gamma \to 0$, we can use the curve $g^A$ to find the optimal control law which as follows $\rho(t) = max(C, q(t)/T)$ and $g^f{'}(E[q(t)|\rho(t)]) = - w$.
% Additionally, the SMF model optimal control law achieves performance close to the SMF model performance bound, and approximating $C^A(w, \gamma, k)$ based on $g^A$ gets us very close to the optimal.
% Finally,
Note also that since the relative variation in link capacity for the synthetic trace (generated using the MIF model in \S\ref{ss:model1:implications}) is same in every time step, 
there is no improvement from using the SMF optimal control law over the MIF model on those traces. 
%, and we can use %insights from the performance bound to pick a point on the curve $g^A$ such that $g^{A}{'}(x) = -w$ and correspondingly pick an approximate value for $C(w, S(t-1))$ ($ = C^{simp}(S(t-1))$ from Equation~\ref{eq:model3:simplified_strategy}) to achieve a desired trade-off.

   \begin{figure}[t]
 \centering
      \begin{subfigure}[h]{0.32\textwidth}
        \includegraphics[width=\textwidth]{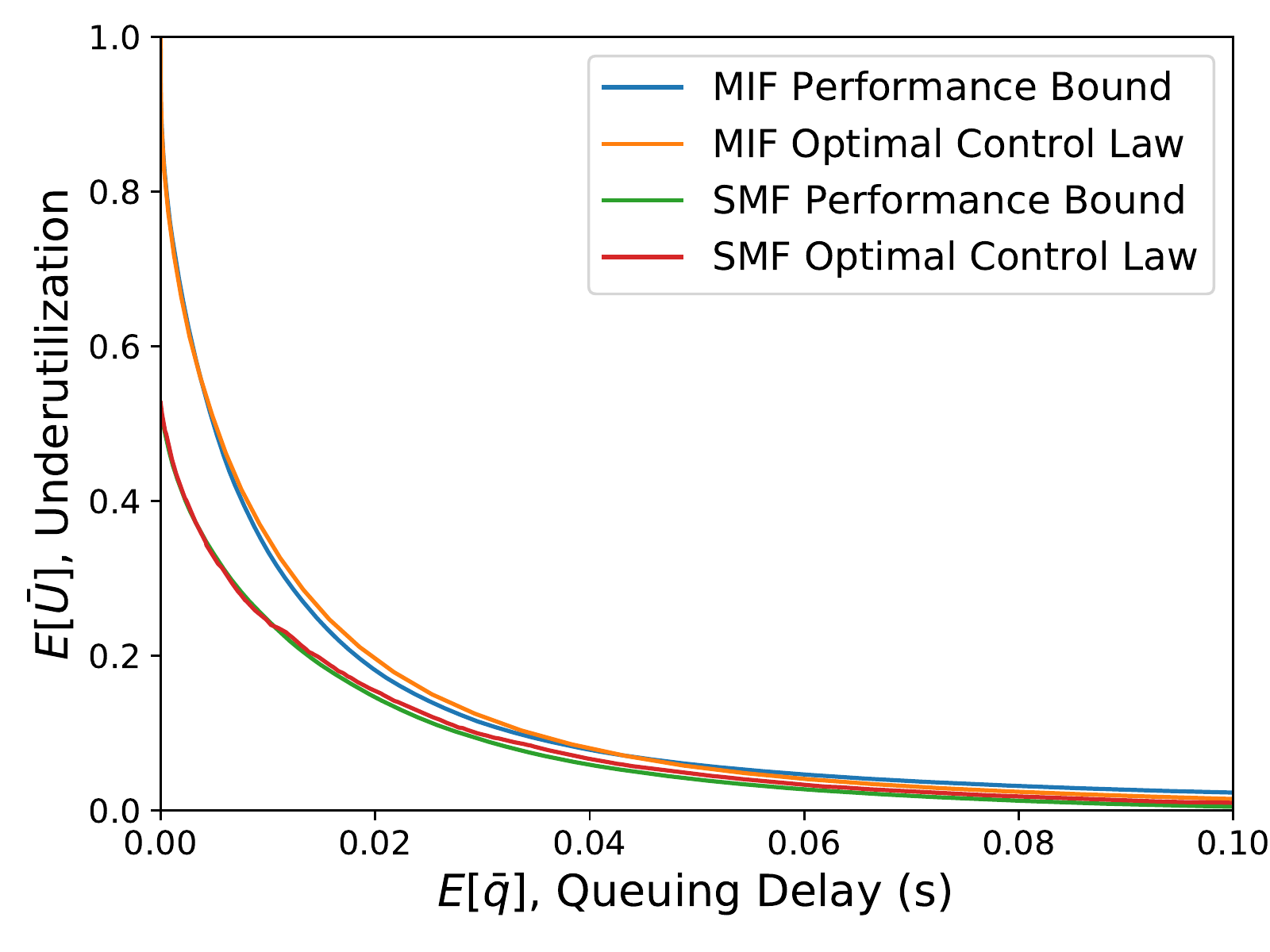}
        \vspace{-6mm}
        \caption{Verizon Downlink}
        \label{fig:model3:verizon_down}
    \end{subfigure}
     \begin{subfigure}[h]{0.32\textwidth}
        \includegraphics[width=\textwidth]{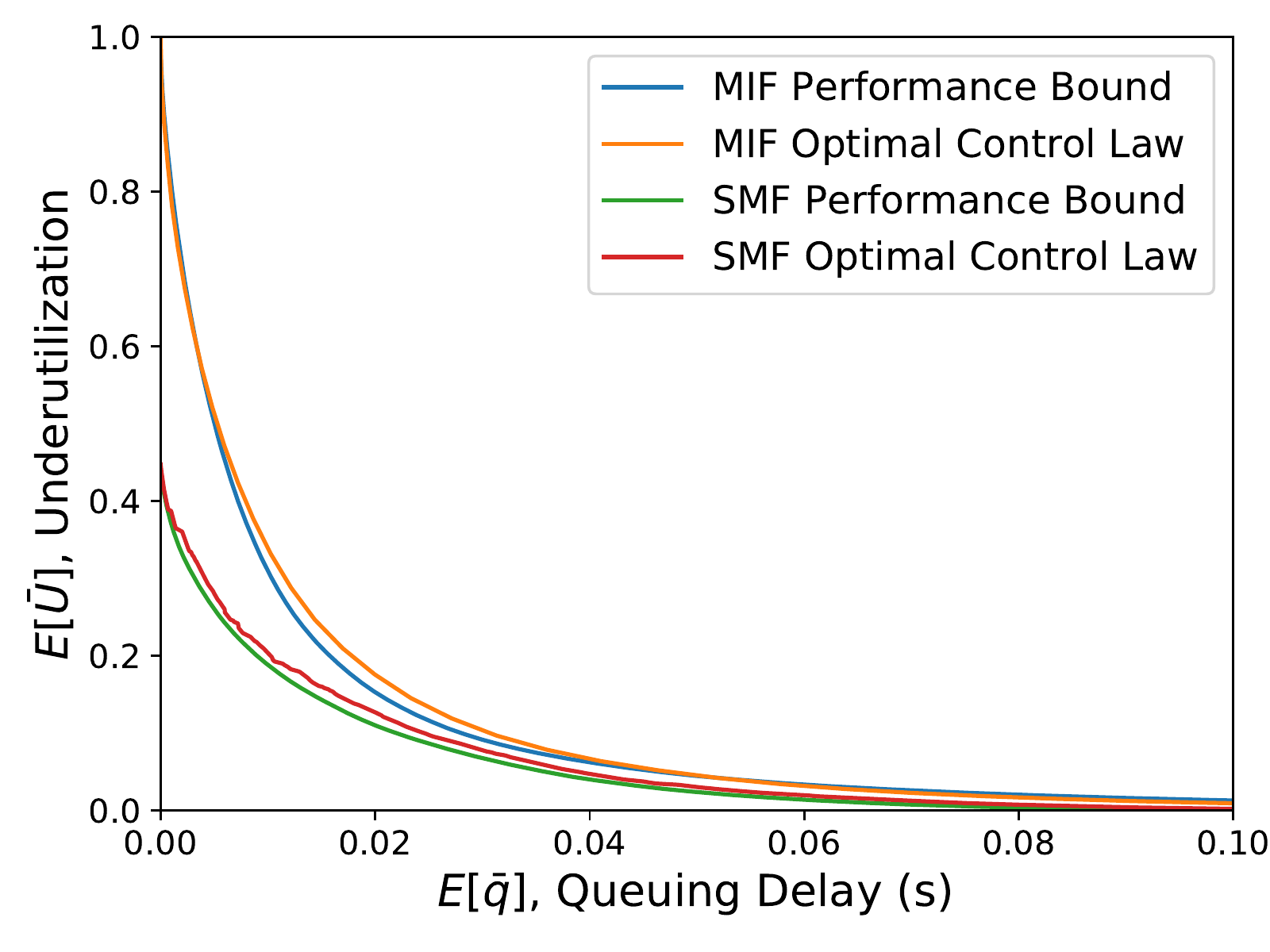}
        \vspace{-6mm}
        \caption{Verizon Uplink}
        \label{fig:model3:verizon_up}
    \end{subfigure}
    \begin{subfigure}[h]{0.32\textwidth}
        \includegraphics[width=\textwidth]{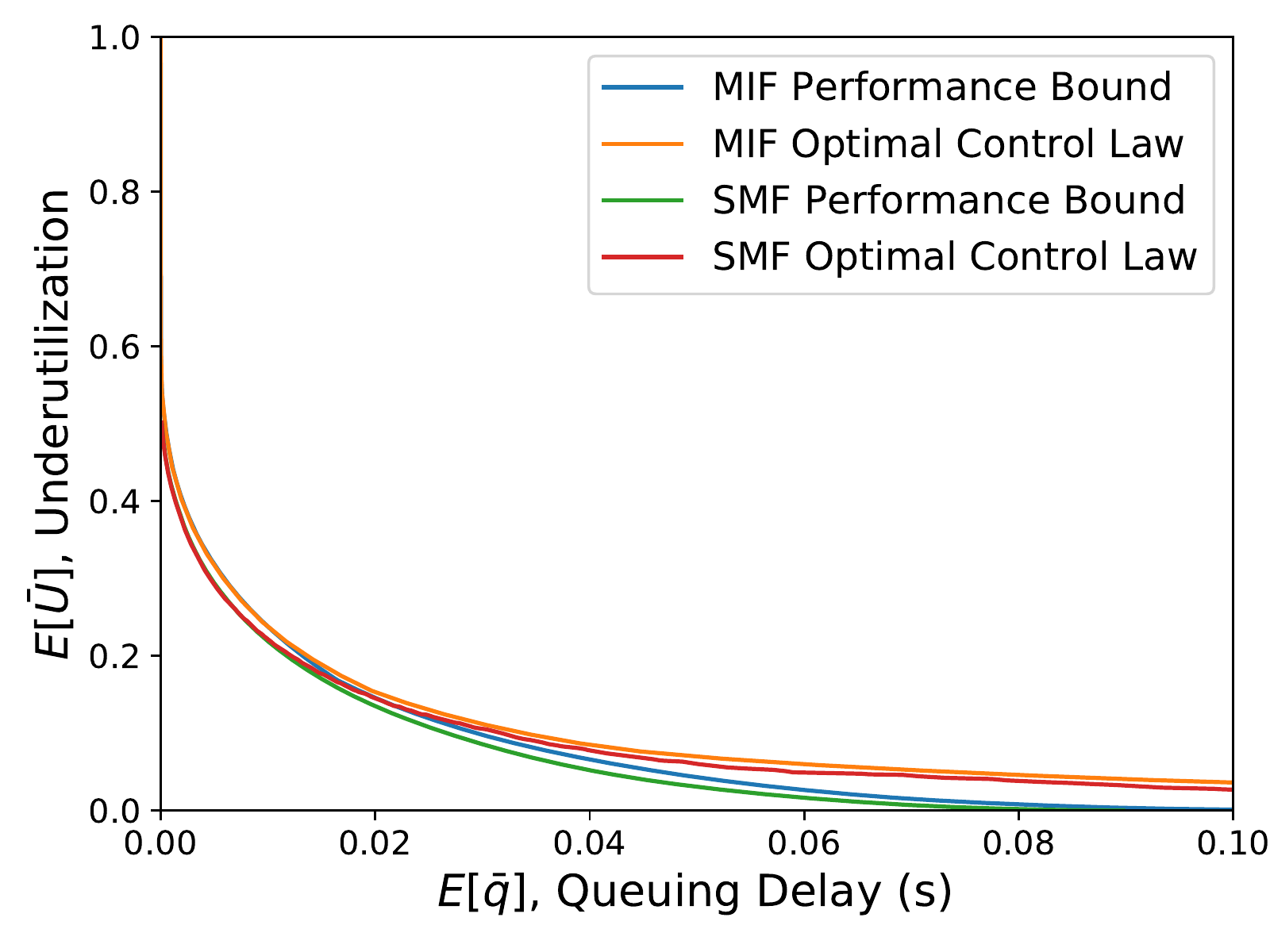}
         \vspace{-6mm}
        \caption{Synthetic trace (MIF model)}
        \label{fig:model3:custom}
    \end{subfigure}
    \vspace{-4mm}
    \caption{\small {\bf Performance bound and performance curve for the optimal control law based on the SMF model.}}
    \label{fig:model3}
    \vspace{-4mm}
\end{figure} 
\section{Discussion}
\label{s:discussion}
%\pg{Moved some things to appendix already. Should I move the remainder.}
%\smallskip
%\noindent

\noindent
%\textbf{Stability in Continuous-time Domain} 
\textbf{Continuous-time Domain.} 
%\label{s:discussion:stability}
In our analysis -- for our models and system equations -- we used the discrete-time domain. However, obviously wireless links operate in the  continuous-time domain. Thus, we might miss sub-RTT effects. We leave extending our analysis to the continuous-time domain as future work. Such an analysis is likely to be substantially more complex than the one we presented here. That said, are the optimal control laws from our analysis \emph{stable} in the continuous-time domain? We show that these control laws are indeed globally asymptotically stable. In particular, we show that if the link capacity stops changing after a time instant $t_0$, regardless of the initial conditions, the sending rate and the queuing delay at the link will converge to certain values $s^{*}$ and $q^{*}$ that depend on the control law parameters. Please see Appendix~\ref{app:disc:stability} for details.

\smallskip
\noindent
\textbf{Implications for Learning-based Congestion Control.}
%\label{s:discussion:learning_based}
%\pg{move to \S\ref{ss:model1:implications}?}
Our analysis reveals the optimal form of the control law for various link models. If the link rate generation process is known, one can derive the exact parameter values in the control law by solving the Bellman equations of the corresponding MDPs. In practice, the transition probabilities for the link capacity are typically not known apriori and they may change over time. In such cases, we envision an online learning system that continuously adjusts the control law in accordance with the performance objective. We believe that compared to existing learning-based approaches for congestion control~\cite{jay2019deep, li2018qtcp, yan2018pantheon}, knowledge about the form of the optimal control law can reduce reduce sample complexity and enable faster adaptation. 
\if 0
\subsection{Miscellaneous}
%The form of the optimal control loop in our models presents an opportunity for congestion control using online learning.  
\pg{Maybe not necessary, the results are mehish}
We tried a simple multi-armed bandit strategy for achieving a particular tradeoff for minimizing $w*E[\bar{q}] + E[\bar{U}]$. We pick parameter values in each round based on performance in the last few seconds while maintaining a constant rate of exploration. We try both Model3/Model1 like strategy and XCP.

   \begin{figure}[ht]
 \centering
    \begin{subfigure}[h]{0.45\textwidth}
        \includegraphics[width=\textwidth]{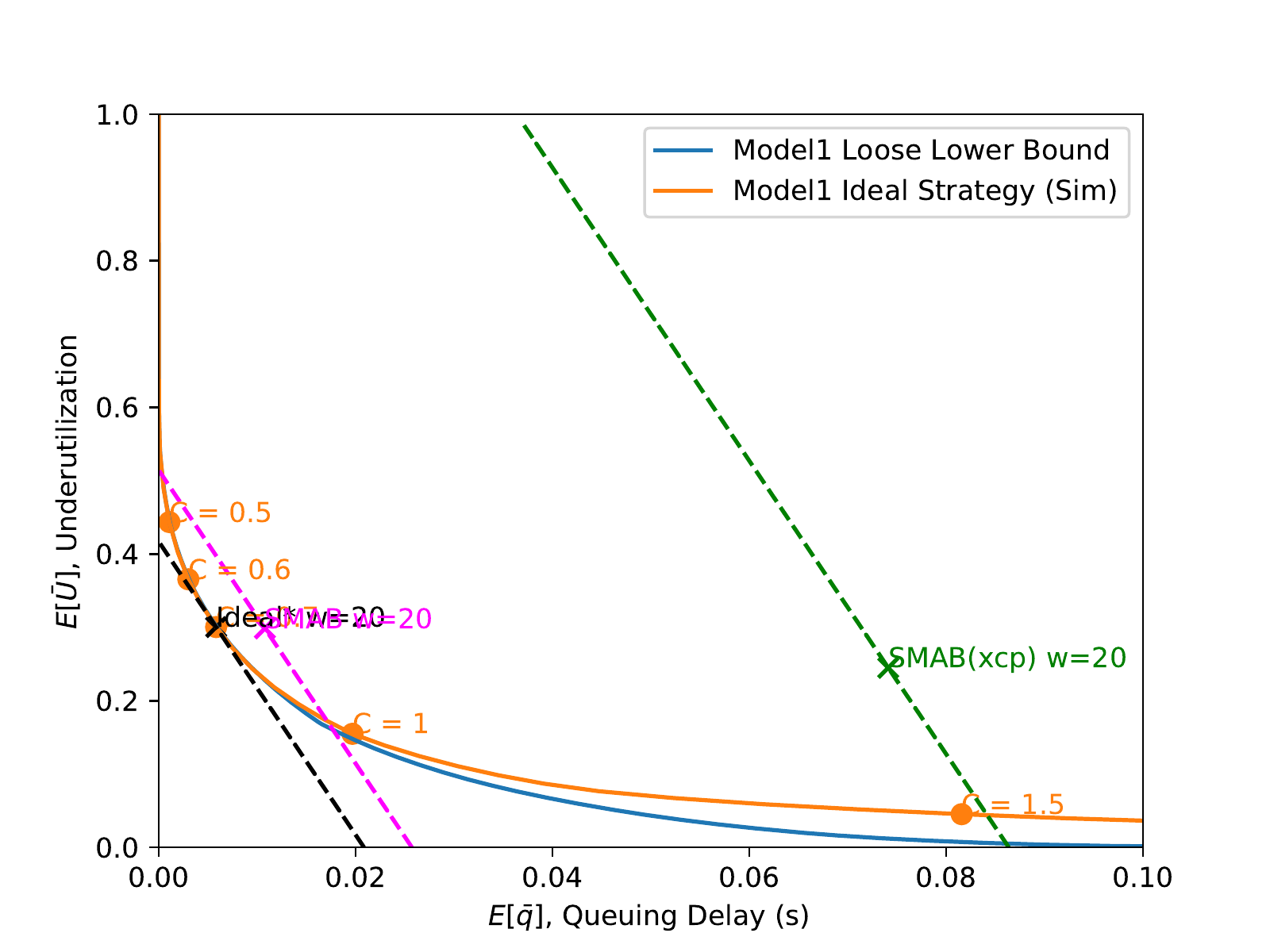}
         \vspace{-7mm}
        \caption{Model1 Trace}
        \label{fig:smab:custom}
    \end{subfigure}
    \begin{subfigure}[h]{0.45\textwidth}
        \includegraphics[width=\textwidth]{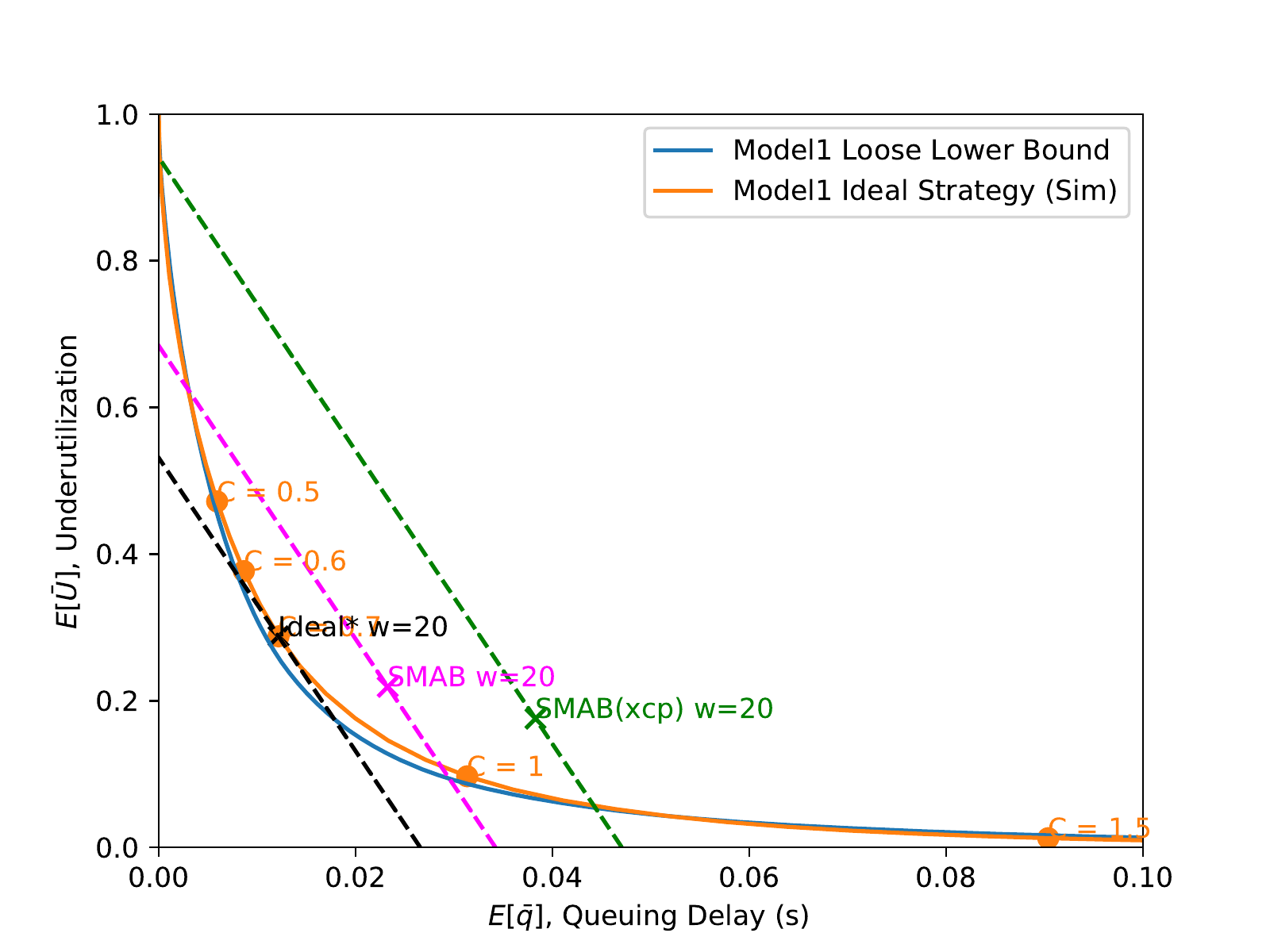}
        \vspace{-7mm}
        \caption{Verizon Uplink (short)}
        \label{fig:smab:verizon_up}
    \end{subfigure}
    \vspace{-3mm}
    \caption{\small Simple multi armed bandit style online learning.}
    \label{fig:smab}
    \vspace{-4mm}
\end{figure} 

\smallskip
\noindent
\pg{LEAVE out?}\textbf{Analysis for metrics expressed not in expectation?}
What if we want to calculate 99th percentile queuing delay? Maybe given a (optimal) control loop, we can construct a Dicrete time markov chain model to analyse various statistics of the control loop.
\fi
\section{Conclusion}
In this paper, we present a theoretical framework for analysing congestion control behaviour on time-varying links. For tractability, we model the link capacity using three simple yet general discrete-time Markov Chains that capture the key components of variability in the capacity. For these models, we show that there is a fundamental bound on the possible performance for \emph{any} causal congestion control protocol. 
This bound is a consequence of the fact that because of feedback delay, the sender does not know the current link capacity, and it must guess to make congestion control decisions.
Finally, we derive the optimal control law for each model by posing the congestion control task as a Markov Decision Process. The optimal control law differs from that used by existing protocols and depends on the characteristics of variability in the link capacity.
We demonstrate the performance improvement of our optimal control laws over existing protocols using real-world cellular traces.
%\input{tradeoff2}
%\input{extensions2}
%\label{p:end}

%\pagebre
%\clearpage
\def\bibfont{\normalfont}
\bibliographystyle{abbrv}
\bibliography{limitscc}

\appendix
\section{Analysis for Unused Link Capacity} 
\label{s:app:unused_mu}
We presented our analysis with queuing delay and underutilization as our performance metrics. Our analysis can also be extended to other performance metrics. For example, we can derive the optimal control law when we use \textit{unused link capacity} as the performance metrics instead of link underutilization. Unused link capacity in time step $t$ ($L(t)$) can be given by,

%\subsection{Calculating lost throughput instead of underutilization}

%The extension above (\S\ref{s:model3}) can also be used to calculate the expected lost throughput and queuing delay trade-off. Assuming the link capacity variation model in \S\ref{s:model1}, let $L(t)$ be the lost throughput in round $t$. Then,

\begin{align}
     L(t) &\triangleq \frac{(\mu(t) \cdot T - s(t) \cdot T - Q(t - 1))}{T}^+.
\end{align}

Assuming the MIF model for variations in link capacity (i.e., the relative variation in link capacity is same in every time step), we can write an equation analogous to Eq.~\eqref{eq:model1:proof_perfbound:single_round_givenrt} for conditional expected value of $L(t)$ as follows,
\begin{align}
     E[L(t)| r(t)=b] &= \mu(t-1) \cdot \int_{b}^{\infty}\left( a - b\right) \cdot f(a) \cdot da. 
\end{align}
 
Notice that unlike underutilization, the conditional expected value of unused link capacity in time step $t$ depends on the link capacity in the previous time step ($\mu(t-1)$). Thus, we need to use the analysis from the SMF model to calculate a performance bound and the optimal control law with $E[\bar{L(t)}]$ and $E[\bar{q(t)}]$ as the performance metrics. In particular, we treat link capacity as the link state ($S(t-1) = \mu(t-1)$). The optimal control law is of the following form in this case, 
\begin{align}
    s(t) = \left( C^{\mu(t-1)}(w) \cdot \mu(t-1) - \frac{Q(t-1)}{T} \right)^+.
\end{align}

\Fig{lost_tput} compares the performance curves for the optimal control law based on link underutilization and unused link capacity. We find that the control law based on $L(t)$ outperforms the one based on $U(t)$. In particular, we find that the optimal control law based on $L(t)$ uses a lower value of $C^{\mu(t-1)}(w)$ when $\mu(t-1)$ is lower. This is expected, as when $\mu(t-1)$ is low, varying $r(t)$  has a relatively lower impact on lost throughput compared to queuing delay. 

  \begin{figure}[t]
 \centering
      \begin{subfigure}[h]{0.4\textwidth}
        \includegraphics[width=\textwidth]{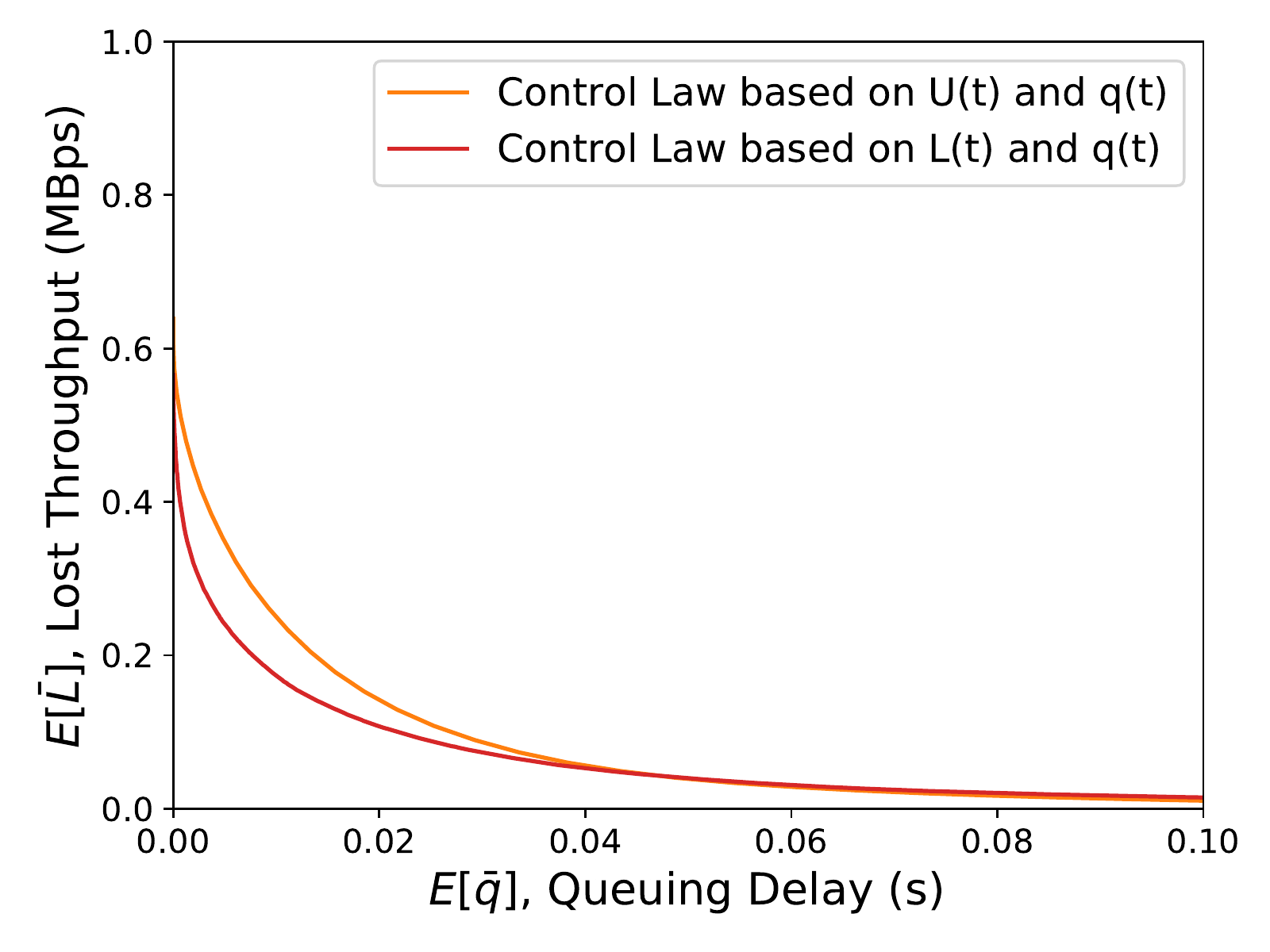}
        \vspace{-7mm}
        \caption{Verizon Downlink}
        \label{fig:lost_tpu:verizon_down}
    \end{subfigure}
     \begin{subfigure}[h]{0.4\textwidth}
        \includegraphics[width=\textwidth]{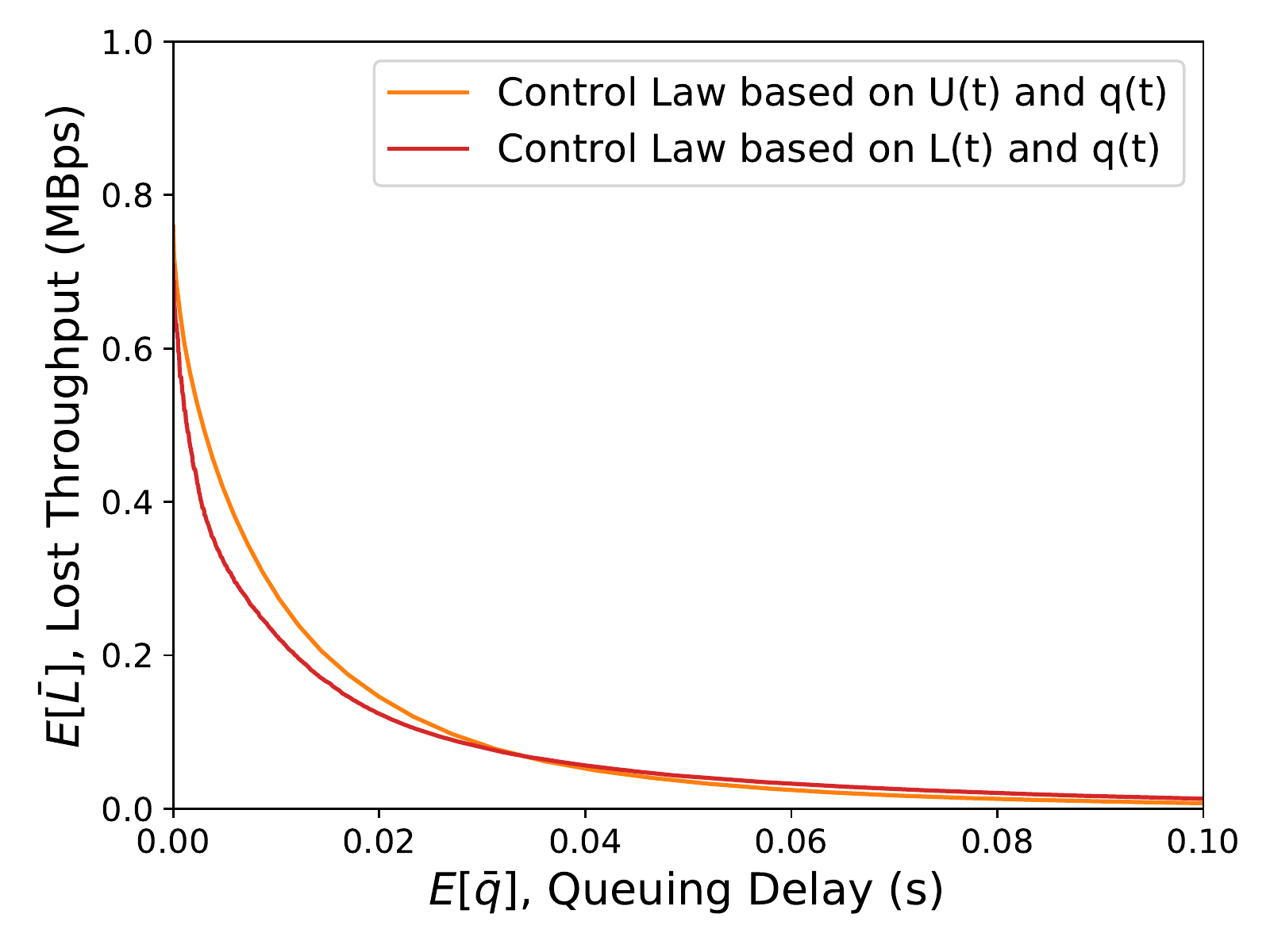}
        \vspace{-7mm}
        \caption{Verizon Uplink}
        \label{fig:lost_tput:verizon_up}
    \end{subfigure}
    \vspace{-4mm}
    \caption{\small {\bf Performance curve of the optimal control law based on unused link capacity ($L(t)$).}}
    \label{fig:lost_tput}
    \vspace{-4mm}
\end{figure}

\section{Analysis for Datacenter Networks}
\label{s:app:analysis_datacenter}
Datacenter networks have fixed capacity wired links that do not exhibit variations in the link capacity. However, datacenter networks are still highly variable: majority of the traffic is contributed by short flows that last only a few RTTs. There is significant churn in the number of competing flows at a link~\cite{bfc}. The optimal sending rate for a particular flow is a function of the number of competing flows at the bottleneck at the current time ($N(t)$). Because of the feedback delay, there is an uncertainty in the value of $N(t)$ at a sender at time $t$. We can potentially analyse such networks -- where the cause of variability is the churn in traffic -- by modeling $N(t)$ as a markov chain.%, and potentially derive the optimal control loop. 
We leave extending our analysis to such network environments as future work.

\section{Implications for Link Layer Design}
\label{app:implications_linklayer}

%\smallskip
%\noindent
%\textbf{:}

Our analysis establishes that there is a fundamental performance bound which depends on the variability in the link capacity. The more variable the link is the worse the achievable performance.
To our knowledge, wireless link layer protocols largely ignore this principle. For example, when arbitrating access to shared
media, scheduling algorithms that provide a relatively smoother link capacity might be more desirable for latency-sensitive applications. Additionally, cellular network operators may be able to offer more accurate 
service level agreements (SLAs) to users based on our performance bounds.

\smallskip
\noindent \textbf{Impact of reducing $T$:}
With 5G, cellular operators are trying to reduce non-congestion related delay incurred in the radio access network. Content providers are also striving to move their servers closer to the receiver and reduce delays.
One would expect that any reductions in the base RTT will reduce the end-to-end delay linearly. Indeed, our analysis for the performance bound also suggests that for a fixed value of underutilization, queuing delay is linear in $T$ assuming that the characteristics of link variability does not change with $T$. However, this assumption need not hold. \Fig{impact_t} shows the performance curve for the optimal control law from the MIF model on a Verizon LTE trace for different values of $T$. We see that the gains in performance are not linear and reducing $T$ has diminishing returns. This is because, when we apply MIF model to this trace, we find that reducing $T$ makes the link capacity more variable on the base RTT time scale and the performance gains are sub-linear. 
\if 0
\begin{wrapfigure}{r}{0.45\textwidth}
    %\centering
    \vspace{-5mm}
    \includegraphics[width=0.43\textwidth]{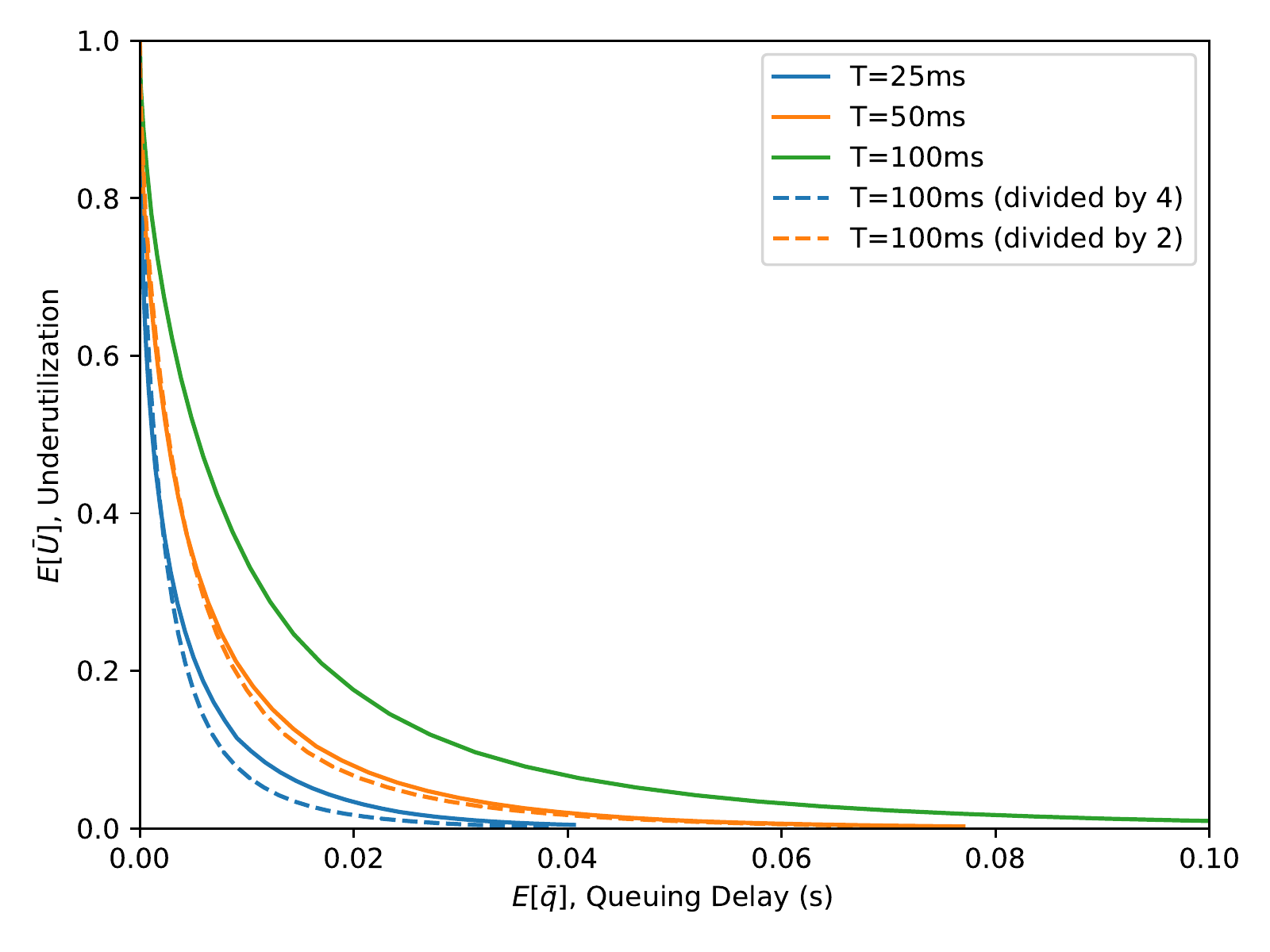}
    \vspace{-4mm}
    \caption{\small {\bf Performance curve for optimal control law from the MIF model for different values of $T$---} The difference in solid and the dashed line for a particular color shows that reducing $T$ has diminishing returns.}
    \label{fig:impact_t}
    \vspace{-5mm}
\end{wrapfigure}
\fi

\begin{figure}[t]%{0.45\textwidth}
    \centering
    \includegraphics[width=0.43\textwidth]{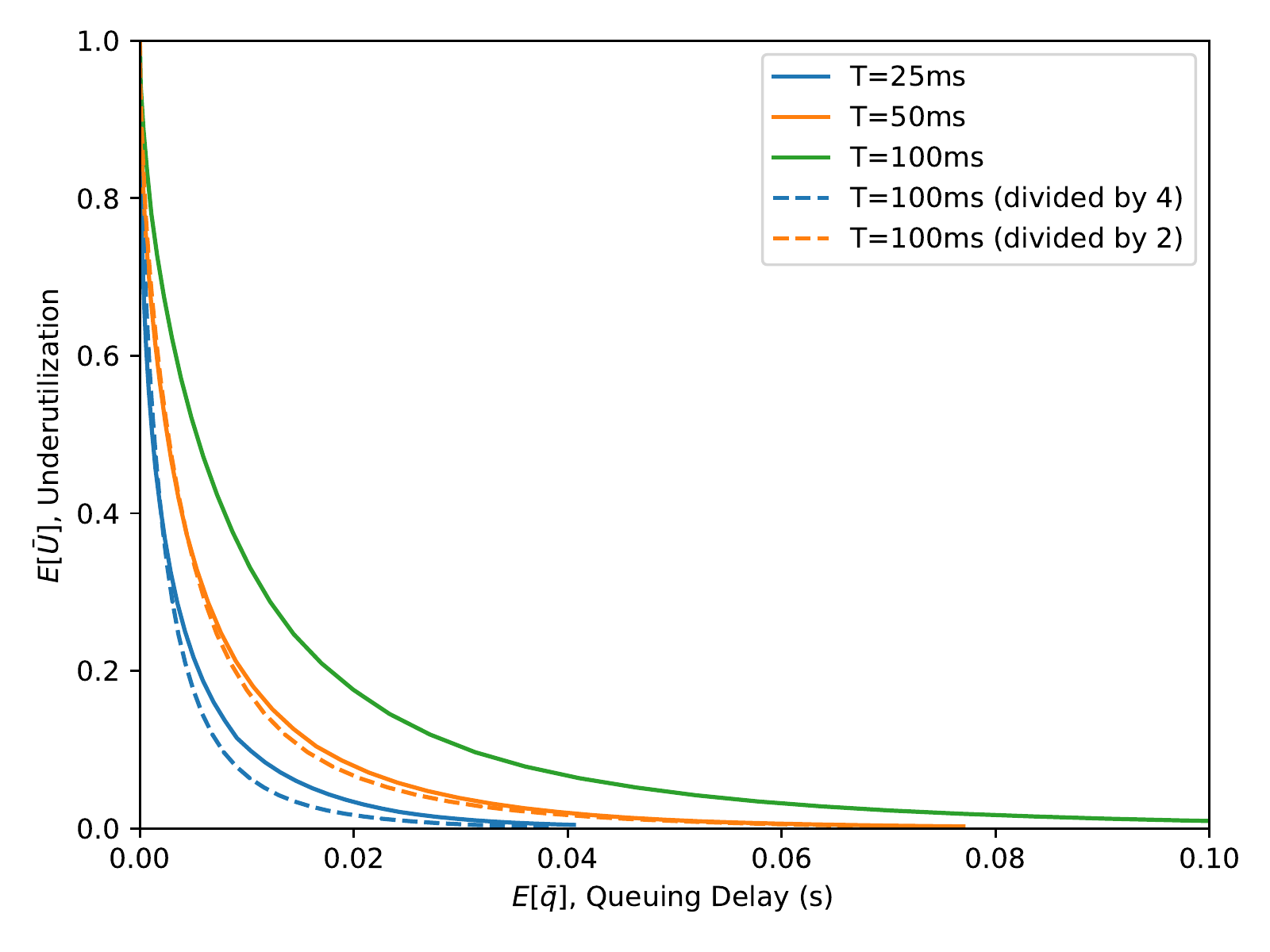}
    \vspace{-4mm}
    \caption{\small {\bf Performance curve for optimal control law from the MIF model for different values of $T$---} The difference in solid and the dashed line for a particular color shows that reducing $T$ has diminishing returns.}
    \label{fig:impact_t}
    \vspace{-5mm}
\end{figure}

\section{MIF Model Proofs}

This section includes proofs from main body for the MIF model (\S\ref{s:model1}). Assuming the terminology established  in \S\ref{s:setup} and \S\ref{s:model1}.

\subsection{Proof of Lemma~\ref{lemma:model1:convexity_of_bound}}
\label{app:model1:perfbound}
\begin{proof}

%We will now prove Lemma~\ref{lemma:model1:convexity_of_bound}. 
We will first calculate the slope of the function $g^f(\cdot)$. Differentiating $x$ and $g^f(x)$ with respect to $b$ using the Leibniz integral rule,
\begin{align}
    \frac{dx}{db} &= T \cdot \frac{d\left(\int_{0}^{b}\left(\frac{b}{a} - 1 \right) \cdot f(a) \cdot da\right)}{db},\nonumber\\
                                        &= T \cdot \left(\int_{0}^{b}\frac{f(a)}{a} \cdot dx \right),\nonumber\\
    \frac{d(g^f(x))}{db} &= \frac{d\left(\int_{b}^{\infty}\left(1 -\frac{b}{a}\right) \cdot f(a) \cdot da\right)}{db},\nonumber\\
                                         &= -\int_{b}^{\infty}\frac{f(a)}{a
    } \cdot da.
\end{align}
Combining the equations above we get,
\begin{align}
    \frac{d(g^f(x))}{dx} &= - \frac{\int_{b}^{\infty}\frac{f(a)}{a
    } \cdot da}{T \cdot \left(\int_{0}^{b}\frac{f(a)}{a} \cdot da \right)}.
    \label{eq:model1:proof_perfbound:slope}
\end{align}
Note that in the above equation the denominator is increasing with $b$ while the absolute value of the numerator is decreasing. This implies that as we increase $b$ (and consequently $x$) the absolute slope of the function $g^f(\cdot)$ is decreasing (or the slope is increasing) with $x$.  Therefore, the function is convex irrespective of the PDF $f$. Since the curve $g^f(\cdot)$ forms the boundary of $\mathbb{C}^f$, $\mathbb{C}^f$is convex.
\end{proof}

\subsection{Proof of Proposition~\ref{prop:model1:achievable_bound}}
\label{app:model1:prop}
\begin{proof}

%\ma{I think this proof could go into the appendix. This proposition isn't central to the story.} \pg{Will move.}

We will first show that if the Proposition condition holds then with the policy $\rho(t) = C$, $Q(t)  \leq T \cdot C \cdot \mu(t)$ $\forall t \in \mathbb{N}$.  We will prove this by induction.

The base case holds because $Q(0) = 0$. Lets assume that $Q(t - 1) 
 \leq T \cdot C \cdot \mu(t-1)$, then we can pick a non negative $s(t)$, such that $\rho(t) = C$. Then in the next time step $t$,
 \begin{align}
     Q(t) &= (Q(t-1) + s(t) \cdot T- \mu(t) \cdot T)^+,\nonumber\\
          &= (T \cdot C \cdot \mu(t - 1) - T \cdot \mu(t))^+,\nonumber\\
          &=  T \cdot C \cdot \mu(t) \cdot \left(\frac{1}{X_t} - \frac{1}{C}\right)^+, \nonumber\\
          &\leq T \cdot C \cdot \mu(t) \cdot \left(\frac{1}{X_{min}} - \frac{1}{C}\right)^+,\nonumber\\
          &\leq T \cdot \mu(t) \cdot C,
 \end{align}
 where in the last step we used Eq.~\eqref{eq:prop:model1:achievable_bound}.
 
 The above equation also shows that if the Proposition condition does not hold then, $Q(t)$ will exceed $T \cdot C \cdot \mu(t)$ with non-zero probability($\geq P(X_t<C/(C+1)$). Consequently, the sender will not always be able follow $\rho(t+1) = C$ in time step $t+1$. Proposition proved.

\end{proof}

\subsection{Proof of Lemma~\ref{lemma:model1:combined_mdp_lemma}}
\label{app:model1:combined_mdp_lemma}

\begin{proof}

To prove this Lemma, we will first reconsider a variant of the MDP over a finite number of time steps $n$.% Later, we will consider the case $n \to \infty$ which corresponds to the original MDP.
For the finite time steps version, we can define the optimal value function at the start of time step $t$ ($V_{t}(\cdot)$) as follows, 

\begin{align}
    V_{t}(q,\mu) = \min_\pi E\left[\sum_{i=t}^{n}\gamma^{i-t} \Big( w \cdot q(i) + U(i)] \Big) \Bigg| q(t-1) = q, \mu(t-1) = \mu, s(t) \sim \pi \right]. 
\end{align}
%We will now write the Bellman Equation governing this optimization problem. Where $V_{t-1}(\cdot)$\cut{$q(t-1),\mu(t-1))$} is the value function at time step $t-1$.

We will show that $V_t(q, u)$ is independent of $\mu$ (Lemma~\ref{lemma:model1:round_value_function_independence}) and convex as a function of $q$ (Lemma~\ref{lemma:model1:mdp_convex}). We close the proof by considering this MDP over infinite horizon ($n \to \infty$). By definition, as $n \to \infty, V_{t}(q,\mu) \to V(q,u)$ and Lemma proved.
\end{proof}
%and show that: as $n \to \infty, V_{t}(q,\mu) \to V(q,u)$.% to close our proof.

The bellman equation for the finite time steps MDP is as follows,
\begin{align}
    V_{t}(q,\mu) &= \min_{s(t) \geq 0} E \left[ w \cdot q(t) + U(t) + \gamma \cdot V_{t+1}(q(t),\mu(t)) \big| q(t-1) = q, \mu(t-1) = \mu \right] \nonumber, & \forall t \leq n,\\
    %V_{t-1}(q(t-1),\mu(t-1)) &= \min_{s(t) \geq 0} \Bigg( w \cdot E[q(t)] + E[U(t)] + \gamma \cdot \int_{0}^{\infty}f(a) \cdot V_{t}(q(t),\mu(t)) \cdot da\Bigg), \nonumber & \forall t \leq n,\\
   V_{t}(q,\mu) &= 0, & t=n+1.
\end{align}

The condition $s(t)\geq 0$ is equivalent to $\rho(t)\geq\frac{q(t-1)}{T}$. Rewriting,
\begin{align}
    V_{t}(q,\mu) &= \min_{\rho(t) \geq \frac{q}{T}} E\left[ w \cdot q(t) + U(t) + \gamma \cdot V_{t+1}(q(t),\mu(t)) \big| q(t-1) = q, \mu(t-1) = \mu\right] \nonumber, & \forall t \leq n,\\
     %V_{t-1}^1(\frac{q(t-1)}{T},\mu(t-1)) &= \min_{\rho(t) \geq \frac{q(t-1)}{T}} \Bigg( w \cdot E[\frac{q(t)}{T}] + E[U(t)]\Bigg) &, t=n
     \label{eq:model1:conjecture}
\end{align}

\begin{lemma}
    $V_t(q, u)$ is independent of $\mu$, $\forall t \leq n$.
    \label{lemma:model1:round_value_function_independence}
\end{lemma}
\begin{proof}
In the Eq.~\eqref{eq:model1:conjecture}, the first two terms on the right hand side ($w \cdot q(t)$ and $U(t)$) depend only on the value of $\rho(t)$. Since $q(t)$ only depends on $\rho(t)$, this implies that if $V_{t+1}(q(t), \mu(t))$ is only a function of $q(t)$ and independent of $\mu(t)$, then $V_{t}(q, \mu)$ is also only a function of $q$. Since $V_{n+1}(\cdot)$ is independent of $\mu(n)$, then, $\forall t \leq n$ ,we can describe the function $V_{t}(q, \mu)$ using only $q$. Lemma proved.

\end{proof}

Using the above Lemma, we can rewrite Eq.~\eqref{eq:model1:conjecture},  

\begin{align}
    V_{t}(q) &= \min_{\rho(t) \geq \frac{q}{T}} E\left[ w \cdot q(t) + U(t) + \gamma \cdot V_{t+1}(q(t)) \big| q(t-1) = q\right], & \forall t \leq n.
    % V_{t-1}^1(\frac{q(t-1)}{T}) &= \min_{\rho(t) \geq \frac{q(t-1)}{T}} \Bigg( w \cdot E[\frac{q(t)}{T}] + E[U(t)]\Bigg) &, t=n
    \label{eq:model1:v_def}
\end{align}

%Since the quantity over which we are taking the minima only depends on $\rho(t)$, $V_t$ is non-decreasing. 

%Consequently, the best strategy in a given round $t$ is of the form $\rho(t) = C(t)$ or $s(t) = max\left(0, C(t) \cdot \mu(t-1) - \frac{Q(t-1)}{T}\right)$. 

%\todo{To show that the best strategy in every round is the same we need to show that the value function in every round is the same.} \pg{Not sure if the math in comments below to prove For all t, $V_{t-1}$ is convex is required at all?}

Next, we will show that $V_{t}(\cdot)$ is convex.
\begin{lemma}
$\forall t \leq n$, $V_{t}(\cdot)$ is convex.
\label{lemma:model1:mdp_convex}
\end{lemma}
\begin{proof}

We give an induction based proof for the Lemma. We will show that if $V_{t+1}(\cdot)$ is convex then $V_t(\cdot)$ is convex. Our proof exploits the fact that $E[q(t)|\rho(t)]$ and $E[U(t)|\rho(t)]$ are both convex as a function of $\rho(t)$.

To prove this Lemma, we will first prove three other Sub Lemmas. To help our analysis, we define a new helper function  $W_{t}(
\cdot)$ as follows,
\begin{align}
    W_{t}(\rho) &= E\left[ w \cdot q(t) + U(t) + \gamma \cdot V_{t+1}(q(t)) \big| \rho(t) = \rho\right] \nonumber, & \forall t \leq n,\nonumber\\
    %W_{t}(\rho(t))  &= \Bigg( w \cdot E[q(t)] + E[U(t)] + \gamma \cdot \int_{0}^{\infty}f(a) \cdot V_{t+1}(T \cdot (\frac{\rho(t)}{a} - 1)^+) \cdot da\Bigg), & \forall t \leq n, \nonumber\\
    V_{t}(q) &= \min_{\rho \geq \frac{q}{T}} (W_{t}(\rho)), & \forall t \leq n.\\
    %V_{t}^2(\rho(t)) &= \Bigg( w \cdot E[\frac{q(t)}{T}] + E[U(t)] \nonumber\\ 
    %&+ \int_{0}^{\infty}f(a) \cdot V_{t-1}^1((\frac{\rho(t)}{a} - 1)^+) \cdot da\Bigg)\nonumber &, \forall t<n\nonumber\\
\end{align}

Expanding $V_{t+1}(q(t))$ in the equation above, we can rewrite $W_{t}$ as follows,

\begin{align}
    W_{t}(\rho) &= ( w \cdot E[q(t)|\rho(t)=\rho] + E[U(t)|\rho(t)=\rho]) + \gamma \cdot \int_{0}^{\infty} \left(f(a) \cdot V_{t+1}\left(T \cdot \left(\frac{\rho}{a} - 1\right)^+\right) \cdot da\right), & \forall t \leq n.
    \label{eq:model1:w_def_round}
\end{align}

\begin{sublemma}
$\forall t \leq n$, if $W_{t}(\cdot)$ is convex then $V_{t}(\cdot)$ is convex.
\end{sublemma}

\begin{proof}
To prove this Sub Lemma we will show 
\begin{align}
    V_{t}(\alpha \cdot x + (1 - \alpha) \cdot y) &\leq \alpha \cdot V_{t}(x) + (1 - \alpha) V_{t}(y), & \forall \alpha \in [0,1], x,y\in\mathbb{R}^+, x<y.
\end{align}

Let's assume that minima of $W_{t}(\cdot)$ occurs at $\rho=C$ (need not be unique). Since $W_{t}(\cdot)$ is convex, 
\begin{align}
W_{t}(x)&\geq W_{t}(C), &\forall x\geq C,\nonumber\\
\Rightarrow   V_{t}(T \cdot x) &= W_{t}(x), &\forall x\geq C,\nonumber\\
W_{t}(x)&\leq W_{t}(C), &\forall x< C,\nonumber\\
\Rightarrow   V_{t}(T \cdot x) &= W_{t}(C), &\forall x< C.
\end{align}

Lets define a new function $h(\cdot)$,
\begin{align}
    h(x) &= \frac{x}{T}, &\forall x \geq C \cdot T,\nonumber\\
    h(x) &= \frac{C}{T}, &\forall x < C \cdot T,\nonumber\\
     \Rightarrow V_{t}(x) &= W_{t}(h(x)).
\end{align}
Note that $h(\cdot)$ is convex. Since, $h(x) \geq C \forall x \in \mathbb{R}^+$ and $W_{t}(x) \geq W_{t}(y), \forall x \geq y \geq C$, we can write,
\begin{align}
    h(\alpha \cdot x + (1 - \alpha) \cdot y) &\leq \alpha \cdot h(x) + (1-\alpha) \cdot h(y), \nonumber\\
    \Rightarrow W_{t}(h(\alpha \cdot x + (1 - \alpha) \cdot y)) &\leq W_{t}(\alpha \cdot h(x) + (1-\alpha) \cdot h(y)).
\end{align}

Using this,
\begin{align}
    V_{t}(\alpha \cdot x + (1 - \alpha) \cdot y) &= W_{t}(h(\alpha \cdot x + (1 - \alpha) \cdot y)), \nonumber\\
     &\leq W_{t}(\alpha \cdot h(x) + (1-\alpha) \cdot h(y)),\nonumber\\
    &\leq \alpha \cdot W_{t}(h(x)) + (1-\alpha) \cdot W_{t}(h(y)),\nonumber\\
    &= \alpha \cdot V_{t}(x) + (1-\alpha) \cdot V_{t}(y).
\end{align}

Sub Lemma proved.
\end{proof}

\begin{sublemma}
$\forall t \leq n$, $V_{t}(\cdot)$ is non-decreasing, i.e., if $x \leq y$, then, $V_{t}(x) \leq V_{t}(y)$.

\end{sublemma}
\begin{proof}
Lets assume that for $V_{t}(y)$ minima occurs at $\rho(t)=C_y$. Since $C_y\geq \frac{y}{T}$, $C_y\geq \frac{x}{T}$. Thus, for $V_{t}(x)$ we can always pick $\rho(t)=C_y$, and, $V_{t}(x)$ is upper bounded by $V_{t}(y)$. Sub lemma proved.
\end{proof}

\begin{sublemma}
$\forall t \leq n$, if $V_{t+1}(\cdot)$ is convex, then $W_{t}(\cdot)$ is convex.
\label{sublemma:model1_helper}
\end{sublemma}
\begin{proof}

Lets define $h_a(\cdot)$ as follows,
\begin{align}
    h_a(x) = T \cdot \left(\frac{x}{a}-1\right)^+.
\end{align}
The functions $h_a(\cdot)$ are convex $\forall a \in \mathbb{R}^+$. Rewriting Eq.~\eqref{eq:model1:w_def_round},
\begin{align}
    W_{t}(x) &= ( w \cdot E[q(t)|\rho(t)=x] + E[U(t)|\rho(t)=x]) + \int_{0}^{\infty}\left(f(a) \cdot V_{t+1}(h_a(x)) \cdot da\right), & \forall t \leq n.\nonumber\\
\end{align}
Since $h_a(\cdot)$ is convex, and $V_{t+1}(\cdot)$ in convex and non-decreasing, $V_{t+1} \circ h_a(\cdot)$ is convex. Additionally $E[q(t)|\rho(t)=x]$ and $E[U(t)|\rho(t)=x]$ are also convex in $x$ (Lemma \ref{lemma:model1:convexity_of_bound} and Eq.~\eqref{eq:model1:proof_perfbound:single_round_givenrt}). Since summation of convex functions is convex, $W_{t}(\cdot)$ is convex. Sub lemma proved.
\end{proof}

The last three Sub lemmas, prove that $\forall t \leq n$ if $V_{t+1}(\cdot)$ is convex then $V_t(\cdot)$ is convex. Since $V_{n+1}(\cdot)$ is convex by definition, Lemma proved. %Conjecture proved.
\end{proof}

\section{PMIF Model Proofs}

This section includes proofs from main body for the PMIF model (\S\ref{s:model2}). Assuming the terminology established  \S\ref{s:setup} and \S\ref{s:model2}.

\subsection{Proof of Corollary~\ref{cor:model2:mdp}}
\label{app:model2:cor_mdp}
\begin{proof}
Define the optimal value function for our MDP as:
\begin{align}
    V^p\left(q,\mu, Pred\right) = \min_{\pi} E\left[\sum_{t=1}^{\infty}\gamma^{t-1} \Big( w \cdot q(t) + U(t)] \Big) \Bigg| q(0) = q, \mu(0) = \mu, Pred(0)=Pred, s(t) \sim \pi \right]. 
\end{align}

We only consider $f^p(\cdot)$ for which $V(q,\mu, Pred)$ exists. %Otherwise, if $V(q,\mu, Pred) \to \infty$, then for the state $(q, \mu, Pred)$, every control law has the same value function ($\infty$) and the optimal control law is not defined.

For ease of analysis, we introduce an additional variable $Q/Pred$ in the optimal value function as the state. Rewriting the equation above,

\begin{align}
    V^p\left(q,\mu, Pred, \frac{Q}{Pred}\right) = \min_{\pi} E\Bigg[&\sum_{t=1}^{\infty}\gamma^{t-1} \Big( w \cdot q(t) + U(t)] \Big) \Bigg|  \nonumber\\
    & q(0) = q, \mu(0) = \mu, Pred(0)=Pred, \frac{Q(0)}{Pred(0)} = \frac{Q}{Pred}, s(t) \sim \pi \Bigg]. 
\end{align}

The condition $s(t) \geq 0$ translates to $\rho^p(t)\geq\frac{Q(t-1)}{T \cdot Pred(t-1)}$. The optimal value function thus satisfies the following Bellman Equation:
\begin{align}
    V^p\left(q,\mu, Pred, \frac{Q}{Pred}\right) = \min_{\rho^p(1) \geq \frac{Q}{T \cdot Pred}} E \Bigg[& w\cdot q(1) + U(1) + \gamma V^p\left(q(1), \mu(1), Pred(1), \frac{Q(1)}{Pred(1)}\right) \Bigg| \nonumber\\
    & q(0) = q, \mu(0) = \mu, Pred(0)=Pred, \frac{Q(0)}{Pred(0)} = \frac{Q}{Pred} \Bigg]. 
\end{align}

\begin{lemma}
$V^p(q,\mu, Pred, \frac{Q}{Pred})$ is a convex function of $\frac{Q}{Pred}$ and does not depend on $q, \mu, Pred$. 
\label{lemma:model2:combined_mdp_lemma}
\end{lemma}
\begin{proof}
Again, we reconsider a variant of this MDP over a finite number of time steps $n$.% Later, we will consider the case $n \to \infty$ which corresponds to the original MDP.
For the finite time steps version, we can define the optimal value function at the start of time step $t$ ($V^p_{t}(\cdot)$) as follows, 

\begin{align}
    V^p_{t}\left(q,\mu, Pred, \frac{Q}{Pred}\right) = \min_{\pi} E\Bigg[&\sum_{i=t}^{n}\gamma^{i-t} \Big( w \cdot q(i) + U(i)] \Big) \Bigg| \nonumber\\
    &q(t-1) = q, \mu(t-1) = \mu, Pred(t-1)=Pred, \frac{Q(t-1)}{Pred(t-1)}=\frac{Q}{Pred}, s(t) \sim \pi \Bigg]. 
\end{align}
%We will now write the Bellman Equation governing this optimization problem. Where $V_{t-1}(\cdot)$\cut{$q(t-1),\mu(t-1))$} is the value function at time step $t-1$.

We will show that $V^p_t(q, u, Pred, \frac{Q}{Pred})$ is independent of $q, \mu, pred$ (Lemma~\ref{lemma:model2:round_value_function_independence}) and convex as a function of $\frac{Q}{Pred}$ (Lemma~\ref{lemma:model2:mdp_convex}). We close the proof by considering this MDP over infinite horizon ($n \to \infty$). By definition, as $n \to \infty, V^p_{t}(q,\mu, Pred, \frac{Q}{Pred}) \to V^p(q,\mu, pred, \frac{Q}{Pred})$ and Lemma proved.
\end{proof}

The Bellman Equation for the finite time steps MDP is as follows,
\begin{align}
    V^p_{t}\left(q,\mu, Pred, \frac{Q}{Pred}\right) = \min_{\rho^p(t) \geq \frac{Q}{T \cdot Pred}} E \Bigg[ &w \cdot q(t) + U(t) + \gamma \cdot V^p_{t+1}\left(q(t),\mu(t), Pred(t), \frac{Q(t)}{Pred(t)}\right) \Bigg| \nonumber\\
    & q(t-1) = q, \mu(t-1) = \mu, \nonumber\\
    & Pred(t-1)=Pred, \frac{Q(t-1)}{Pred(t-1)}=\frac{Q}{Pred}\Bigg] \nonumber, & \forall t \leq n,\\
    %V_{t-1}(q(t-1),\mu(t-1)) &= \min_{s(t) \geq 0} \Bigg( w \cdot E[q(t)] + E[U(t)] + \gamma \cdot \int_{0}^{\infty}f(a) \cdot V_{t}(q(t),\mu(t)) \cdot da\Bigg), \nonumber & \forall t \leq n,\\
   V^p_{t}\left(q,\mu, Pred, \frac{Q}{Pred}\right) &= 0, & t=n+1.
   \label{eq:model2:conjecture}
\end{align}

\begin{lemma}
    $V^p_t(q, u,  pred, \frac{Q}{pred})$ is independent of $q, \mu, pred$, $\forall t \leq n$.
    \label{lemma:model2:round_value_function_independence}
\end{lemma}
\begin{proof}
Now, 
\begin{align}
    \frac{Q(t)}{Pred(t)}=  T \cdot \left(\frac{\rho^p(t)}{X^p_t} - 1\right)^+  \cdot \frac{X^p_t}{X^{pred}_t}
\end{align}

In the Eq.~\eqref{eq:model2:conjecture}, the first two terms on the right hand side ($w \cdot q(t)$ and $U(t)$) depend only on the value of $\rho^p(t)$. Since $Q(t)/Pred(t)$ also only depends on $\rho^p(t)$, this implies that if $V^p_{t+1}(q(t), \mu(t), Pred(t), Q(t)/Pred(t)$ is only a function of $Q(t)/Pred(t)$, then $V^p_{t}(q, \mu)$ is also only a function of $q$. Since $V^p_{n+1}(\cdot)$ is independent of $q(n), \mu(n), Pred(n)$, then, $\forall t \leq n$ ,we can describe the function $V^p_{t}(q, \mu, Pred, \frac{Q}{Pred})$ using only $Q/Pred$. Lemma proved.

\end{proof}

Using the above Lemma, we can rewrite Eq.~\eqref{eq:model2:conjecture},  

\begin{align}
    V^p_{t}\left(\frac{Q}{Pred}\right) &= \min_{\rho^p(t) \geq \frac{Q}{T \cdot Pred}} E\left[ w \cdot q(t) + U(t) + \gamma \cdot V_{t+1}\left(\frac{Q(t)}{Pred(t)}\right) \Bigg| \frac{Q(t-1)}{Pred(t-1)} = \frac{Q}{Pred}\right], & \forall t \leq n.
    % V_{t-1}^1(\frac{q(t-1)}{T}) &= \min_{\rho(t) \geq \frac{q(t-1)}{T}} \Bigg( w \cdot E[\frac{q(t)}{T}] + E[U(t)]\Bigg) &, t=n
    \label{eq:model2:v_def}
\end{align}

\begin{lemma}
$\forall t \leq n$, both $V^p_{t}(\cdot)$ is convex.
\label{lemma:model2:mdp_convex}
\end{lemma}

\begin{proof}
To help our analysis, we define a new helper function  $W_{t}(
\cdot)$ as follows,
\begin{align}
    W^p_{t}(\rho^p) &= E\left[ w \cdot q(t) + U(t) + \gamma \cdot V^p_{t+1}\left(\frac{Q(t)}{Pred(t)}\right) \Big| \rho^p(t) = \rho^p\right] \nonumber, & \forall t \leq n,\nonumber\\
    %W_{t}(\rho(t))  &= \Bigg( w \cdot E[q(t)] + E[U(t)] + \gamma \cdot \int_{0}^{\infty}f(a) \cdot V_{t+1}(T \cdot (\frac{\rho(t)}{a} - 1)^+) \cdot da\Bigg), & \forall t \leq n, \nonumber\\
    V^p_{t}\left(\frac{Q}{Pred}\right) &= \min_{\rho^p \geq \frac{Q}{T \cdot Pred}} (W^p_{t}(\rho^p)), & \forall t \leq n.\\
    %V_{t}^2(\rho(t)) &= \Bigg( w \cdot E[\frac{q(t)}{T}] + E[U(t)] \nonumber\\ 
    %&+ \int_{0}^{\infty}f(a) \cdot V_{t-1}^1((\frac{\rho(t)}{a} - 1)^+) \cdot da\Bigg)\nonumber &, \forall t<n\nonumber\\
\end{align}

Expanding $V^p_{t+1}(\frac{Q(t)}{Pred(t)})$ in the equation above, we can rewrite $W^p_{t}$ as follows,

\begin{align}
    W^p_{t}(\rho^p) = w \cdot E[&q(t)|\rho^p(t)=\rho^p] + E[U(t)|\rho^p(t)=\rho^p]\nonumber\\ 
    &+ \gamma \cdot \int_{0}^{\infty}\int_{0}^{\infty}\Big(f^p(a) \cdot f^{pred}(b) \cdot V^p_{t+1}\left(T \cdot \left(\frac{\rho^p}{a} - 1\right)^+ \cdot \frac{a}{b}\right) \cdot da \cdot db\Big)\nonumber &, \forall t \leq n\nonumber\\
    % V_{t-1}^1(\frac{q(t-1)}{T}) &= \min_{r(t) \geq \frac{q(t-1)}{T}} \Bigg( w \cdot E[\frac{q(t)}{T}] + E[U(t)]\Bigg) &, t=n
\end{align}

The remainder of proof for this lemma is similar to proof of Lemma~\ref{lemma:model1:mdp_convex}, except we need to redefine $h_a(\cdot)$ from Sub lemma~\ref{sublemma:model1_helper} as follows.
\begin{align}
    h_{a,b}(x) = T \cdot \left(\frac{x}{a}-1\right)^+ \cdot \frac{a}{b}
\end{align}
$V^p_{t+1} \circ h_{a,b}$ is also convex. 
\end{proof}

To close the proof of the corollary, we exploit the fact that $q(t), U(t)$ and $\frac{Q(1)}{Pred(1)}$ only depend on the value of $\rho^p(t)$, and define a helper function $W(\cdot)$ as follows,
\begin{align}
    W(\rho^p) &= E\left[ w \cdot q(1) + U(1) + \gamma \cdot V^p\left(\frac{Q(1)}{Pred(1)}\right) \Bigg| \rho^p(1) = \rho^p\right] \nonumber\\
    %W_{t}(\rho(t))  &= \Bigg( w \cdot E[q(t)] + E[U(t)] + \gamma \cdot \int_{0}^{\infty}f(a) \cdot V_{t+1}(T \cdot (\frac{\rho(t)}{a} - 1)^+) \cdot da\Bigg), & \forall t \leq n, \nonumber\\
    V^p\left(\frac{Q}{Pred}\right) &= \min_{\rho^p \geq \frac{Q}{T \cdot Pred}} (W^p(\rho^p)).
    %V_{t}^2(\rho(t)) &= \Bigg( w \cdot E[\frac{q(t)}{T}] + E[U(t)] \nonumber\\ 
    %&+ \int_{0}^{\infty}f(a) \cdot V_{t-1}^1((\frac{\rho(t)}{a} - 1)^+) \cdot da\Bigg)\nonumber &, \forall t<n\nonumber\\
    \label{eq:model2:mdp_helper_def}
\end{align}
In other words, $W^p(\rho^p)$ is the optimal value function restricting the first action $\rho^p(1)$ to $\rho^p$.

\begin{lemma}
    $W^p(\rho^p)$ is a convex function in $\rho^p$.
\end{lemma}
\begin{proof}
The proof directly follows from proof of Sub lemma~\ref{sublemma:model1_helper} in Appendix~\ref{app:model1:combined_mdp_lemma}.
\end{proof}

Let the minimum value of $W^p(\cdot)$ occur at $\rho^{p*}$. Then, the convexity of $W^p(\cdot)$ combined with Eq.~\eqref{eq:model2:mdp_helper_def} establishes that the optimal control law is of the form $\rho^p(t) = max(\rho^{p*}, \frac{Q(t-1)}{T \cdot Pred(t-1)})$, where $\rho^{p*}$ depends on $w$ and $\gamma$. Defining $C^p(w, \gamma) = \rho^{p*}$, we get the optimal control law from the corollary.

\end{proof}

 \section{SMF Model Proofs}

This section includes proofs from main body for the SMF model (\S\ref{s:model3}). Assuming the terminology established  \S\ref{s:setup} and \S\ref{s:model3}.

\subsection{Proof of Proposition~\ref{prop:model3:compute_da}}
\label{app:prop:compute_da}
\begin{proof}

To prove this Proposition, we will first show that if $y=g^A(x)$, then $g^{f^i}{'}(x^i) = g^{f^j}{'}(x^j), \forall i,j \in \textbf{S}$. We will prove this by contradiction. Lets assume, $g^{f^i}{'}(x^i) < g^{f^j} {'}(x^j)$. Then, we can pick a new point $(x_1, y_1)$ ($x_1 = \sum_{k \in \textbf{S}} \lambda(i) \cdot x_1^k, y_1 = \sum_{k \in \textbf{S}} \lambda(i) \cdot y_1^k$) such that
\begin{align}
    y_1^k &= g^{f^k}(x_1^k), &\forall k \in \textbf{S},\nonumber\\
    x_1^k &= x^k, & \forall k \in \textbf{S} - \{i,j\},\nonumber\\
    \Rightarrow y_1^k &= g^{f^k}(x^k), &\forall k \in \textbf{S}-\{i,j\},\nonumber\\
    x_1^i &= x^i + \frac{\epsilon}{\lambda(i)}, & \nonumber\\
    \Rightarrow y_1^i &= y^i + \frac{\epsilon}{\lambda(i)}g^{f^i}{'}(x^i), & if \epsilon \to 0 ,\nonumber\\
    x_1^j &= x^j - \frac{\epsilon}{\lambda(j)}, &\nonumber\\
    \Rightarrow y_1^j &= y^j - \frac{\epsilon}{\lambda(j)}g^{f^j}{'}(x^j), & if \epsilon \to 0.
\end{align}

If $\epsilon \to 0$, then $(x_1, y_1)$ can be given by
\begin{align}
    x_1 &= \left(\sum_{k \in \textbf{S}}\lambda(k) \cdot x^k\right) + \lambda(i) \cdot \frac{\epsilon}{\lambda(i)} - \lambda(j) \cdot \frac{\epsilon}{\lambda(j)}, \nonumber\\
       &= x, \nonumber\\
    y_1 &= \left(\sum_{k \in \textbf{S}}\lambda(k) \cdot g^i(x^i)\right) + \lambda(i) \cdot \frac{\epsilon}{\lambda(i)} \cdot g^{f^i}{'}(x^i) - \lambda(j) \cdot \frac{\epsilon}{\lambda(j)} \cdot g^{f^j}{'}(x^j),\nonumber\\
     &= y + \epsilon \cdot (g^{f^i}{'}(x^i) - g^{f^j}{'}(x^j)).
\end{align}
Since $x_1 = x$ and $y_1 < y$, the point $(x,y)$ cannot be on the boundary. Contradiction.

Finally, we will show that if the point $(x,y)$ satisfies the condition $g^{f^k}{'}(x^k) = C \forall k \in \textbf{S}$, then it will be on the boundary. To do this we will show that there does not exist another point $(x_1, y_1) (x_1 = \sum_{k \in \textbf{S}} \lambda(i) \cdot x_1^k, y_1 = \sum_{k \in \textbf{S}} \lambda(i) \cdot g^{f^k}(x^k_1)$) which satisfies the condition $g^{f^k}{'}(x_1^k) = C_1$, $\forall k \in \textbf{S}$ and $x=x_1$ but $y_1<y$. Lets assume $C < C1$. Since $g^{f^k}$ in convex with increasing slope,
\begin{align}
    g^{f^k}{'}(x^k) &< g^{f^k}{'}(x_1^k), &\forall k \in \textbf{S},\nonumber\\
    \Rightarrow x^k &< x_1^k, &\forall k \in \textbf{S},\nonumber\\
    \Rightarrow x &< x_1. &
\end{align}
Since $x<x_1$, we have a contradiction. We can show a similar contradiction if $C>C1$ (we will get $x>x_1$). For $C=C_1$, since $g^{f^k}(\cdot)$ is convex, if $x=x_1$, then $y=y_1$. This is because
\begin{align}
    %y_1^k &= g^{f^k}(x_1^k), & \forall k \in \textbf{S},\nonumber\\
    g^{f^k}(x_1^k)  &= g^{f^k}(x^k + (x_1^k - x^k)),& \forall k \in \textbf{S},\nonumber\\
          &= g^{f^k}(x^k) + g^{f^k}{'} \cdot (x_1^k - x^k),& \forall k \in \textbf{S},\nonumber\\
          &= g^{f^k}(x^k) + C \cdot (x_1^k - x^k),& \forall k \in \textbf{S},\nonumber\\
    y_1 &= \sum_{k \in \textbf{S}}g^{f^k}(x^k_1) =y + C \cdot \sum_{k \in \textbf{S}}(x_1^k - x^k), &\nonumber\\
    \Rightarrow y_1 &= y. &
\end{align}
Proposition proved!
\end{proof}
\subsection{Proof of Proposition~\ref{prop:model3:achievable_bound}}
\label{app:model3:achievable bound}
\begin{proof}
First, we will show that if the Theorem condition holds then with the policy $\rho(t) = C^A(S(t-1))$, $Q(t)  \leq T \cdot C^A(S(t)) \cdot \mu(t)$ $\forall t \in \mathbb{N}, \forall S(t-1) \in \textbf{S}$. We will prove the Proposition using induction

Base case holds as Q(0) = 0. Lets assume that $Q(t - 1) 
 \leq T \cdot C^A(S(t-1)) \cdot \mu(t-1)$ $\forall S(t-1) \in \textbf{S}$, then $\forall S(t) \in \textbf{S}$, we can pick a non negative $s(t)$, such that $\rho(t) = C^A(S(t-1))$. Then in the next time step $t$,
 \begin{align}
     Q(t) &= (Q(t-1) + s(t) \cdot T- \mu(t) \cdot T)^+,\nonumber\\
          &= (T \cdot C^A(S(t-1)) \cdot \mu(t - 1) - T \cdot \mu(t))^+,\nonumber\\
          &=  T \cdot C^A(S(t)) \cdot \mu(t) \cdot \left(\frac{C^A(S(t-1))}{C^A(S(t)) \cdot X^{S(t-1)}_t} - \frac{1}{C^A(S(t))}\right)^+, \nonumber\\
          &\leq T \cdot C^A(S(t)) \cdot \mu(t) \cdot \left(\frac{C^A(S(t-1))}{C^A(S(t)) \cdot X^{S(t-1)}_{min}} - \frac{1}{C^A(S(t))}\right)^+.
 \end{align}

Rewriting the Proposition condition using $k_1 = S(t-1)$ and $k_2 = S(t)$,
 \begin{align}
     \frac{C^A(S(t-1))}{C^A(S(t)) \cdot X^{S(t-1)}_{min}} - \frac{1}{C^A(S(t))} &\leq 1, \nonumber\\
     \Rightarrow Q(t) &\leq T \cdot \mu(t) \cdot C^A(S(t)).
 \end{align}

The above analysis also shows that if the Proposition condition does not hold then, $Q(t)$ will exceed $T \cdot C^A(S(t)) \cdot \mu(t)$ with non-zero probability. Consequently, the sender will not always be able follow $\rho(t+1) = C^A(S(t))$ in time step $t+1$. Proposition proved. 
 
 %This conditions depends not only on the current link state but 
 \end{proof}
 
 \subsection{Proof of Corollary~\ref{cor:model3:mdp}}
 \label{app:model3:cor_mdp}
 \begin{proof}
 Define the optimal value function for our MDP as:
\begin{align}
    V^A(q,\mu, k) = \min_{\pi} E\left[\sum_{t=1}^{\infty}\gamma^{t-1} \Big( w \cdot q(t) + U(t)] \Big) \Big| q(0) = q, \mu(0) = \mu, S(0)=k, s(t) \sim \pi \right]. 
\end{align}

We only consider $\cup_{k \in \textbf{S}}f^k(\cdot)$ such that $V^A(q,\mu, k)$ exists. %Otherwise, if $V(q,\mu, k) \to \infty$, then for the state $(q, \mu, k)$, every control law has the same value function ($\infty$) and the optimal control law is not defined.

The optimal value function satisfies the following Bellman Equation:
\begin{align}
    V^A(q,\mu, k) = \min_{\rho(1) \geq \frac{q}{T}} E \left[ w\cdot q(1) + U(1) + \gamma V^A(q(1), \mu(1), S(1)) \Big| q(0) = q, \mu(0) = \mu, S(0)=k, \right]. 
\end{align}

\begin{lemma}
$V^A(q,\mu, k)$ is a convex function of $q$ and does not depend on $\mu$. 
\label{lemma:model3:combined_mdp_lemma}
\end{lemma}
\begin{proof}
The proof for this lemma is analogous to proof of Lemma~\ref{lemma:model1:combined_mdp_lemma}. Again, we can define the optimal value function in time step $t$ $V^A_t(\cdot)$ for the finite time steps version of the MDP. The Bellman Equation for $V^A_t(\cdot)$ is as follows,

\begin{align}
    V^A_t(q,\mu, k) = \min_{\rho(t) \geq \frac{q}{T}} E\Bigg[\Big(&w \cdot q(t) +U(t) + \gamma \cdot \sum_{S(t) \in \textbf{S}} P^A(S(t)|S(t-1)) \cdot V^A_{t+1}(q(t),\mu(t), S(t))\Big) \Bigg| \nonumber\\
    & q(t-1) = q, \mu(t-1)=\mu, S(t-1)=k\Bigg],  & \forall t \leq n, \forall k \in \textbf(S)\nonumber\\
    %V_{t-1}(q(t-1),\mu(t-1)) &= \min_{s(t) \geq 0} \Bigg( w \cdot E[q(t)] + E[U(t)] + \gamma \cdot \int_{0}^{\infty}f(a) \cdot V_{t}(q(t),\mu(t)) \cdot da\Bigg) \nonumber &, \forall t < n\\
   V^A_{t}(q,\mu, k) &= 0, & t=n+1, \forall k \in \textbf{S}.
\end{align}

Using the proof of Lemma~\ref{lemma:model1:mdp_convex}, we can show that the $V^A_t(q, \mu, k)$ is a convex function of $q$ and does not depend on $\mu$. The key change in the induction step is that if $V^A_{t+1}(q, \mu, k)$ is convex in $q, \forall k \in \textbf{S}$, then, $V^A_{t}(q, \mu, k)$ is convex in $q, \forall k \in \textbf{S}$.

Consequently,  $V^A(q,\mu, k)$ is a convex function of $q$ and does not depend on $\mu$.
\end{proof}

 To close the proof of the corollary, we exploit the fact that $q(t),U(t)$ and $S(t)$ only depend on the value of $\rho(t) and S(t-1)$ and define a helper function $W^A(\cdot)$ as follows,
\begin{align}
    W^A(\rho, k) &= E\left[ w \cdot q(1) + U(1) + \gamma \cdot V^A(q(1), S(1)) \big| \rho(1) = \rho, S(0)=k\right] \nonumber\\
    %W_{t}(\rho(t))  &= \Bigg( w \cdot E[q(t)] + E[U(t)] + \gamma \cdot \int_{0}^{\infty}f(a) \cdot V_{t+1}(T \cdot (\frac{\rho(t)}{a} - 1)^+) \cdot da\Bigg), & \forall t \leq n, \nonumber\\
    V^A(q, k) &= \min_{\rho \geq \frac{q}{T}} (W^A(\rho, k)).
    %V_{t}^2(\rho(t)) &= \Bigg( w \cdot E[\frac{q(t)}{T}] + E[U(t)] \nonumber\\ 
    %&+ \int_{0}^{\infty}f(a) \cdot V_{t-1}^1((\frac{\rho(t)}{a} - 1)^+) \cdot da\Bigg)\nonumber &, \forall t<n\nonumber\\
    \label{eq:model3:mdp_helper_def}
\end{align}
In other words, $W^A(\rho, k)$ is the optimal value function restricting the first action $\rho(1)$ to $\rho$ and link state $S(0)$ to $k$.

\begin{lemma}
    $\forall k \in \textbf{S}, W^A(\rho, k)$ is a convex function in $\rho$.
\end{lemma}
\begin{proof}
The proof is analogous to proof of Sub lemma~\ref{sublemma:model1_helper} in Appendix~\ref{app:model1:combined_mdp_lemma}.
\end{proof}

Given a $k$, let the minimum value of $W^A(x, k)$ occur at $\rho^{k*}$. Then, the convexity of $W^A(\cdot)$ combined with Eq.~\eqref{eq:model3:mdp_helper_def} establishes that the optimal control law is of the form $\rho(t)|(S(t-1)=k) = max(\rho^{k*}, \frac{q(t-1)}{T})$, where $\rho^{k*}$ depends on $w$ and $\gamma$. Defining $C^A(w, \gamma, k) = \rho^{k*}$, we get the optimal control law from the corollary.

 \end{proof}
 
 %\section{Proof of Theorem~\ref{thm:disc:stability}}
\section{Stability in Continuous-time domain}
\label{app:disc:stability}

 Let $\mu_{I}(t), s_{I}(t), q(t)_{I}(t), Q_{I}(t)$ be the instantaneous link capacity, sending rate, queuing delay and queue size respectively. Assuming $\mu_{I}(t) = \mu$ for $t > t_0$, then for $t > t_0 + T$, the optimal control law in all our models will be of the form,
\begin{align}
    s_{I}(t) = \left(C  - \frac{Q_{I}(t-T)}{T}\right)^+,
    \label{eq:disc:stability_stdef}
\end{align}
where $C$ is a positive constant depending on the $\mu$ and the model.

\begin{proposition}
    For a single bottleneck time-varying link, the control law from Eq.~\eqref{eq:disc:stability_stdef} is globally asymptotically stable, $\forall C \in \mathbb{R}^+$.
    \label{thm:disc:stability}
\end{proposition}
 \begin{proof}
To prove this Proposition, we leverage the proof of stability for the ABC control law (see Appendix C in~\cite{abc}). Similar to ABC, ignoring the boundary condition (${q}_{I}(t)$ must be $\geq 0$), we can describe the rate of change of queuing delay as follows,
\begin{align}
    \dot{q}_{I}(t) &=  \frac{s_{I}(t) - \mu}{\mu}, \nonumber\\
    \dot{q}_{I}(t) &= \left(\frac{C}{\mu} - \frac{q_{I}(t-T)}{T}\right)^+ - 1,\nonumber\\
    \dot{q}_{I}(t) &= max\left(\left(\left(\frac{C}{\mu} - 1\right) - \frac{q_{I}(t-T)}{T}\right), - 1\right).
\end{align}

Lets define $x(t) = q_{I}(t) - T \cdot (\frac{C}{\mu} - 1)$, then,
\begin{align}
    \dot{x}(t) = max\left(-\frac{x(t-T)}{T}, - 1\right),\nonumber\\
    \dot{x}(t) = -min\left(\frac{x(t-T)}{T}, 1\right),\nonumber\\
    \dot{x}(t) = -g(x(t-T)),
\end{align}
where $g(u) = min(\frac{u}{T}, 1)$.

In ~\cite{yorke1970asymptotic} (Corollary 3.1), Yorke established that delay-differential equations of this type are globally asymptotically stable (i.e., ${x}(t) \to 0$ as $t \to \infty$ irrespective of the initial condition), if the following conditions are met:
\begin{enumerate}
    \item \textbf{H$_{1}$:} g is continuous.
    \item \textbf{H$_{2}$:}  There exists some $\alpha$, s.t. $\alpha \cdot u^2 > ug(u) > 0$ for all $u \neq 0$.
    \item \textbf {H$_{3}$:} $\alpha \cdot T < \frac{3}{2}$.
\end{enumerate}

The function $g(\cdot)$ trivially satisfies \textbf{H$_{1}$}. $\alpha \in (\frac{1}{T}, \frac{3}{2 \cdot T})$ satisfies both \textbf{H$_2$} and \textbf{H$_3$}.

%As $t \to \infty$, $q_{Ins.}(t) \to T \cdot (\frac{C}{\mu} - 1) (= q^*)$ and  $s_{Ins.}(t) \to $.
\end{proof}
%\end{sloppypar}
%\theendnotes
\end{document}